\documentclass[a4paper,UKenglish,cleveref, autoref, thm-restate]{lipics-v2021}

\hideLIPIcs  

\usepackage{graphicx}
\usepackage{microtype}
\usepackage{url}
\usepackage{algpseudocode}
\usepackage{enumerate}
\usepackage[scr=rsfso]{mathalpha}
\usepackage{multirow}
\usepackage{hyperref}

\usepackage{silence}
\title{Dynamic r-index: An Updatable Self-Index in LCP-bounded Time} 

\author{Takaaki Nishimoto}{RIKEN Center for Advanced Intelligence Project, Japan}{takaaki.nishimoto@riken.jp}{}{}
\author{Yasuo Tabei}{RIKEN Center for Advanced Intelligence Project, Japan}{yasuo.tabei@riken.jp}{}{}
\authorrunning{T. Nishimoto and Y. Tabei}
\Copyright{T. Nishimoto and Y. Tabei}



\category{} 

\relatedversion{} 

\nolinenumbers 

\EventEditors{John Q. Open and Joan R. Access}
\EventNoEds{2}
\EventLongTitle{42nd Conference on Very Important Topics (CVIT 2016)}
\EventShortTitle{CVIT 2016}
\EventAcronym{CVIT}
\EventYear{2016}
\EventDate{December 24--27, 2016}
\EventLocation{Little Whinging, United Kingdom}
\EventLogo{}
\SeriesVolume{42}
\ArticleNo{23}

\begin{document}

\maketitle

\newcommand{\floor}[1]{\left \lfloor #1 \right \rfloor}
\newcommand{\ceil}[1]{\left \lceil #1 \right \rceil}

\newcommand{\argmax}{\mathop{\rm arg~max}\limits}
\newcommand{\argmin}{\mathop{\rm arg~min}\limits}
\newcommand{\polylog}{\mathop{\rm polylog}\limits}


\newcommand{\modulo}{\mathsf{mod}}
\newcommand{\SA}{\mathsf{SA}}
\newcommand{\ISA}{\mathsf{ISA}}

\newcommand{\LCE}{\mathsf{LCE}}
\newcommand{\LCP}{\mathsf{LCP}}
\newcommand{\lcp}{\mathsf{lcp}}


\newcommand{\unproven}{\color{red}UNPROVEN\color{black}}
\newcommand{\DESCRIPTION}{\color{red}DESCRIPTION\color{black}}
\newcommand{\DUMMY}{\color{red}DUMMY DUMMY \color{black}}

\newcommand{\BWT}{\mathsf{BWT}}

\newcommand{\CS}{\mathsf{CS}}

\newcommand{\SUM}{\mathsf{sum}}
\newcommand{\SEARCH}{\mathsf{search}}
\newcommand{\UPDATE}{\mathsf{update}}
\newcommand{\MERGE}{\mathsf{merge}}
\newcommand{\DIVIDE}{\mathsf{divide}}
\newcommand{\LCS}{\mathsf{LCS}}
\newcommand{\PA}{\mathsf{PA}}
\newcommand{\IPA}{\mathsf{IPA}}
\newcommand{\RLE}{\mathsf{RLE}}
\newcommand{\DI}{\mathsf{DI}}
\newcommand{\first}{\mathsf{first}}
\newcommand{\prev}{\mathsf{prev}}
\newcommand{\pred}{\mathsf{pred}}

\newcommand{\predecessor}{\mathsf{pred}}

\newcommand{\MR}{\mathsf{MR}}
\newcommand{\avg}{\mathsf{avg}}
\newcommand{\ins}{\mathsf{ins}}
\newcommand{\del}{\mathsf{del}}
\newcommand{\LF}{\mathsf{LF}}
\newcommand{\occ}{\mathsf{occ}}

\newcommand{\CM}{\mathsf{CM}}

\newcommand{\tmp}{\mathsf{tmp}}

\newcommand{\compSA}{\mathsf{comp}\mathchar`-\mathsf{SA}}
\newcommand{\compISA}{\mathsf{comp}\mathchar`-\mathsf{ISA}}
\newcommand{\compT}{\mathsf{comp}\mathchar`-\mathsf{T}}

\newcommand{\rank}{\mathsf{rank}}
\newcommand{\lexCount}{\mathsf{lex}\mathchar`-\mathsf{count}}
\newcommand{\lexSearch}{\mathsf{lex}\mathchar`-\mathsf{search}}
\newcommand{\runAccess}{\mathsf{run}\mathchar`-\mathsf{access}}
\newcommand{\runIndex}{\mathsf{run}\mathchar`-\mathsf{index}}
\newcommand{\access}{\mathsf{access}}

\newcommand{\accessSA}{\mathsf{access}}
\newcommand{\orderSA}{\mathsf{order}}
\newcommand{\countSA}{\mathsf{count}}
\newcommand{\insertSA}{\mathsf{insert}}
\newcommand{\deleteSA}{\mathsf{delete}}
\newcommand{\incrementSA}{\mathsf{increment}}
\newcommand{\decrementSA}{\mathsf{decrement}}

\newcommand{\sort}{\mathsf{sort}}
\newcommand{\select}{\mathsf{select}}

\newcommand{\sumPS}{\mathsf{sum}}
\newcommand{\searchPS}{\mathsf{search}}
\newcommand{\insertPS}{\mathsf{insert}}
\newcommand{\deletePS}{\mathsf{delete}}
\newcommand{\mergePS}{\mathsf{merge}}
\newcommand{\dividePS}{\mathsf{divide}}
\newcommand{\replacePS}{\mathsf{replace}}

\newcommand{\insertL}{\mathsf{insert}}
\newcommand{\deleteL}{\mathsf{delete}}

\newcommand{\accessPI}{\mathsf{access}}

\newcommand{\paccess}{\mathsf{access}}
\newcommand{\pinvAccess}{\mathsf{inv}\mathchar`-\mathsf{access}}
\newcommand{\pindex}{\mathsf{index}}

\newcommand{\incrementInsert}{\mathsf{increment}\mathchar`-\mathsf{insert}}
\newcommand{\decrementDelete}{\mathsf{decrement}\mathchar`-\mathsf{delete}}

\newcommand{\pinsert}{\mathsf{insert}}
\newcommand{\pdelete}{\mathsf{delete}}

\newcommand{\insertRLE}{\mathsf{insert}}
\newcommand{\deleteRLE}{\mathsf{delete}}
\newcommand{\splitRLE}{\mathsf{split}}
\newcommand{\mergeRLE}{\mathsf{merge}}
\newcommand{\incrementRLE}{\mathsf{increment}}
\newcommand{\decrementRLE}{\mathsf{decrement}}

\newcommand{\DRI}{\textsf{DRI}}
\newcommand{\DFMI}{\textsf{DFMI}}
\newcommand{\RI}{\textsf{RI}}


\begin{abstract}
A self-index is a compressed data structure that supports locate queries—reporting all positions where a given pattern occurs in a string while maintaining the string in compressed form.  
While many self-indexes have been proposed, developing dynamically updatable ones supporting string insertions and deletions remains a challenge.
The r-index (Gagie et al., JACM'20) is a representative static self-index based on the run-length Burrows–Wheeler transform (RLBWT), designed for highly repetitive strings.
We present the dynamic r-index, a dynamic extension of the r-index that achieves updates in LCP-bounded time.
The dynamic r-index supports count queries in $\mathcal{O}(m \log r / \log \log r)$ time and locate queries in $\mathcal{O}(m \log r / \log \log r + \mathsf{occ} \log r)$ time, using $\mathcal{O}(r)$ words of space, where $m$ is the length of a query with $\mathsf{occ}$ occurrences and $r$ is the number of runs in the RLBWT.
Crucially, update operations are supported in $\mathcal{O}((m + L_{\mathsf{max}}) \log n)$ time for a substring of length $m$, where $L_{\mathsf{max}}$ is the maximum LCP value; the average running time is $\mathcal{O}((m + L_{\mathsf{avg}}) \log n)$, where $L_{\mathsf{avg}}$ is the average LCP value.  
This LCP-bounded complexity is particularly advantageous for highly repetitive strings where LCP values are typically small.
We experimentally demonstrate the practical efficiency of the dynamic r-index on various highly repetitive datasets.
\end{abstract}

\section{Introduction}\label{sec:introduction}
Efficient search in large string datasets is a critical challenge in many applications.
Particularly, datasets composed of highly repetitive strings—those with many repeated substrings—has grown significantly in recent years.
Typical examples of such datasets include web pages collected by crawlers~\cite{FerraginaM10}, version-controlled documents such as those on Wikipedia~\cite{wikipedia}, and, perhaps most notably, biological sequences including human genomes~\cite{Przeworski00}.

To address these search challenges, compressed data structures called self-indexes have been developed. A self-index allows locate queries—that is, reporting all positions where a given pattern occurs in the string. Among these, the FM-index~\cite{DBLP:journals/jacm/FerraginaM05,DBLP:journals/talg/FerraginaMMN07} is an efficient self-index based on the Burrows–Wheeler Transform (BWT)~\cite{burrows1994block}, a reversible permutation of the input string. 
The FM-index consistently delivers high performance across various types of strings, supporting locate queries in $\mathcal{O}(m + \occ \log_{\sigma} n)$ time using $\mathcal{O}(n (1 + \log \sigma/\log n))$ words of space~\cite{DBLP:journals/talg/BelazzouguiN14}, where $n$ is the length of text $T$, $m$ is the length of the query string, $\sigma$ is the alphabet size, and $\occ$ is the number of occurrences of the query in $T$.

However, for highly repetitive strings, standard FM-indexes suffer from limited space efficiency.  
To address this, various specialized self-indexes based on compressed formats have been proposed, including grammar-based~\cite{DBLP:conf/spire/ClaudeN12a,DBLP:journals/talg/ChristiansenEKN21,9961143,ViceVersa},  
LZ-based~\cite{DBLP:conf/latin/ChristiansenE18,DBLP:conf/latin/GagieGKNP14}, and Block-Tree-based~\cite{DBLP:journals/tcs/NavarroP19} approaches.  
Among them, the run-length Burrows–Wheeler transform (RLBWT)—a run-length encoded BWT—has attracted attention for its space efficiency and high search performance for highly repetitive strings~\cite{DBLP:journals/jcb/MakinenNSV10,10.1145/3375890}.  
Gagie et al.~\cite{10.1145/3375890} proposed the r-index, an RLBWT-based self-index whose space usage is linear in the number $r$ of runs in the RLBWT.  
It supports locate queries in $\mathcal{O}(r)$ words and  
$\mathcal{O}(m \log \log_{B}(\sigma + n/r) + \occ \log \log_{B}(n/r))$ time, for word size $B = \Theta(\log n)$.  
OptBWTR~\cite{DBLP:conf/icalp/NishimotoT21} improves upon the r-index, achieving optimal time and space for constant alphabets.  
A practical implementation is presented in~\cite{DBLP:conf/wea/Bertram0N24}.  
However, most self-indexes for highly repetitive strings are static, meaning they do not support updates such as string insertions or deletions.

Recently, there has been a growing demand for real-time string processing in low-resource environments, especially in dynamic settings where input evolves continuously~\cite{Mehmood2020RealTime}.
Such needs often arise in edge computing, where data must be processed locally — on end-user devices or nearby systems — due to constraints such as latency, bandwidth, or privacy.
To meet these challenges, it is essential to develop dynamic, update-friendly data structures that can efficiently handle modifications to input strings without requiring full reconstruction.

\begin{table}[htb] 
\caption{Summary of state-of-the-art dynamic indexes supporting count and locate queries.
    $n$ is the length of the input string $T$; 
    $m$ is the length of a given pattern $P$; 
    $\occ$ is the number of occurrences of $P$ in $T$; 
    $\occ_{c} \geq \occ$ is the number of candidate occurrences of $P$ 
    obtained by the ESP-index~\cite{DBLP:journals/jda/MaruyamaNKS13}.
    $m^{\prime}$ is the length of a given string inserted into $T$ 
    or a substring deleted from $T$.
    $\sigma = n^{\mathcal{O}(1)}$ is the alphabet size of $T$; 
    $r$ is the number of runs in the BWT of $T$; 
    $\epsilon > 0$ is an arbitrary constant;
    $L_{\max}$ is the maximum value in the LCP array of $T$;
    $L_{\avg}$ is the average of the values in the LCP array;
    $g$ is the size of a compressed grammar deriving $T$, 
    and $g = \mathcal{O}(z \log n \log^{*} n)$~\cite{DBLP:journals/dam/NishimotoIIBT20}, 
    where $z$ is the number of factors in the LZ77 of $T$~\cite{LZ76}; 
    $q \geq 1$ is a user-defined parameter; 
    $g^{\prime}$ is the total size of the $q$-truncated suffix tree of $T$ and 
    a compressed grammar deriving $T$, and $g^{\prime} = \mathcal{O}(z(q^{2} + \log n \log^{*} n))~\cite{DBLP:journals/iandc/NishimotoTT20}$;     
    $\DRI$ and $\DFMI$ represent the dynamic r-index and dynamic FM-index, respectively.}
\label{tab:dat}

\tiny

\begin{tabular}{p{15em}|c|c|c}
\multirow{2}{*}{Method}  & \multirow{2}{*}{Format} & Space          & \multirow{2}{*}{Locate time}  \\
       &        & (words)        &   \\
\hline
$\DFMI$~\cite{DBLP:journals/jda/SalsonLLM10,DBLP:journals/jda/LeonardMS12} & BWT & $\mathcal{O}(n \frac{\log \sigma}{\log n})$   & $\mathcal{O}((m+\occ) \log^{2+\epsilon} n)$  \\
\hline
\multirow{2}{*}{SE-index~\cite{DBLP:journals/dam/NishimotoIIBT20}}  & \multirow{2}{*}{Grammar} & \multirow{2}{*}{$\mathcal{O}(g)$}  & $\mathcal{O}(m (\log \log n)^{2} + \occ \log n $  \\
 &  &  & $+ \log m (\log n \log^{*} n)^{2})$   \\
\hline
\multirow{3}{*}{TST-index-d~\cite{DBLP:journals/iandc/NishimotoTT20}} & \multirow{3}{*}{Grammar} & \multirow{3}{*}{$\mathcal{O}(g^{\prime})$}   & $\mathcal{O}(m (\log\log \sigma)^{2} + \occ \log n)$ ($m \leq q$)  \\
\cline{4-4}
    &   &   & $\mathcal{O}(m (\log\log n)^{2}$   \\
    &   &   & $+ \occ_{c} \log m \log n \log^{*} n)$ ($m > q$)  \\

\hline
Gawrychowski et al.~\cite{DBLP:journals/corr/GawrychowskiKKL15} & Grammar & $\Omega(n)$ & $\mathcal{O}(m + \log^{2} n + \occ \log n)$  \\
\hline
\hline
$\DRI$ (this study) & RLBWT & $\mathcal{O}(r)$   & $\mathcal{O}((m+\occ) \log n)$  \\
\end{tabular}

\vspace{5mm}

\begin{tabular}{p{15em}|c|c|c}
Method & Format & Space (words)        & Count time \\
\hline
$\DFMI$~\cite{DBLP:journals/jda/SalsonLLM10,DBLP:journals/jda/LeonardMS12} & BWT & $\mathcal{O}(n \frac{\log \sigma}{\log n})$   & $\mathcal{O}(m \log n)$ \\
\hline
\hline
$\DRI$ (this study) & RLBWT & $\mathcal{O}(r)$   & $\mathcal{O}(m \log n)$ \\
\end{tabular}

\vspace{5mm}

\begin{tabular}{p{15em}|c|c|c}
Method & Format & Space (words)        & Update time \\
\hline
\multirow{3}{*}{$\DFMI$~\cite{DBLP:journals/jda/SalsonLLM10,DBLP:journals/jda/LeonardMS12}} & \multirow{3}{*}{BWT} & \multirow{3}{*}{$\mathcal{O}(n \frac{\log \sigma}{\log n})$}   & $\mathcal{O}((m^{\prime} + L_{\max}) \log n)$ \\
 \cline{4-4}
 &  &  & $\mathcal{O}((m^{\prime} + L_{\avg}) \log n)$ \\
&  &  & (average case) \\

\hline
\multirow{2}{*}{SE-index~\cite{DBLP:journals/dam/NishimotoIIBT20}}  & \multirow{2}{*}{Grammar} & \multirow{2}{*}{$\mathcal{O}(g)$}  & $\mathcal{O}(m^{\prime}(\log n \log^{*} n)^{2}$ \\
 &  &  & $+ \log n(\log n \log^{*} n)^{2})$ \\
\hline
\multirow{2}{*}{TST-index-d~\cite{DBLP:journals/iandc/NishimotoTT20}} & \multirow{2}{*}{Grammar} & $\mathcal{O}(g^{\prime})$  & $\mathcal{O}((\log \log n)^{2} (m^{\prime}q $ \\
    &   &   & $+ q^{2} + \log n \log^{*} n) )$ \\
\hline
Gawrychowski et al.~\cite{DBLP:journals/corr/GawrychowskiKKL15} & Grammar & $\Omega(n)$ & $\mathcal{O}(m^{\prime} \log n + \log^{2} n)$ \\
\hline
\hline
 \multirow{3}{*}{$\DRI$ (this study)} & \multirow{3}{*}{RLBWT} & \multirow{3}{*}{$\mathcal{O}(r)$}   & $\mathcal{O}((m^{\prime} + L_{\max}) \log n)$ \\
 \cline{4-4}
           &     &                & $\mathcal{O}((m^{\prime} + L_{\avg}) \log n)$ \\
           &     &                & (average case) \\           
\end{tabular}

\end{table}

Only a few dynamic self‑indexes that support update operations have been proposed so far. 
Table~\ref{tab:dat} summarizes the time and space complexities of representative dynamic string indexes. 
The dynamic FM‑index~\cite{DBLP:journals/jda/SalsonLLM10} extends the FM‑index to support string insertions and deletions in general strings. Grammar‑based approaches include the SE‑index~\cite{DBLP:journals/dam/NishimotoIIBT20}, the TST‑index‑d~\cite{DBLP:journals/iandc/NishimotoTT20},  
and a string index proposed by Gawrychowski et al.~\cite{DBLP:journals/corr/GawrychowskiKKL15}.
Despite these advances, no self‑index with the combination of practicality, scalability, and full dynamic functionality required by real‑world applications has yet been realized. Consequently, developing such a practical, dynamically updatable self‑index remains an important open challenge.

{\bf Contribution.} 
We present the dynamic r-index, a dynamic extension of the r-index.
The original r-index uses three sampled subsequences from the RLBWT and suffix array, along with static data structures supporting rank and lex-count queries, to efficiently answer locate queries.  
The main challenge in making it dynamic lies in updating these subsequences, which store partial information from the BWT and suffix array that must remain consistent under modifications.  
To address this, we design an efficient update algorithm using dynamic data structure techniques for partial sums~\cite{DBLP:journals/algorithmica/BilleCCGSVV18},  
rank/select~\cite{DBLP:conf/esa/MunroN15}, and permutation access~\cite{DBLP:journals/jda/SalsonLLM10}.  
We analyze the time and space complexity and show that they depend on the maximum or average values in the longest common prefix (LCP) array.  
The dynamic r-index supports count queries in $\mathcal{O}(m \log r / \log \log r)$ time and locate queries in $\mathcal{O}(m \log r / \log \log r + \occ \log r)$ time, using $\mathcal{O}(r)$ words of space.  
Update operations are supported in $\mathcal{O}((m + L_{\max}) \log n)$ time for processing a string of length $m$,  
where $L_{\max}$ is the maximum LCP value; the average running time is $\mathcal{O}((m + L_{\avg}) \log n)$, where $L_{\avg}$ is the average LCP value.  
We experimentally evaluated the dynamic r-index on highly repetitive strings and demonstrated its practicality.
A C++ reference implementation of the dynamic r-index will be released as open-source together with the camera-ready version.

%

\section{Preliminaries}\label{sec:preliminary}
\begin{figure}[t]
\begin{center}
	\includegraphics[width=0.3\textwidth]{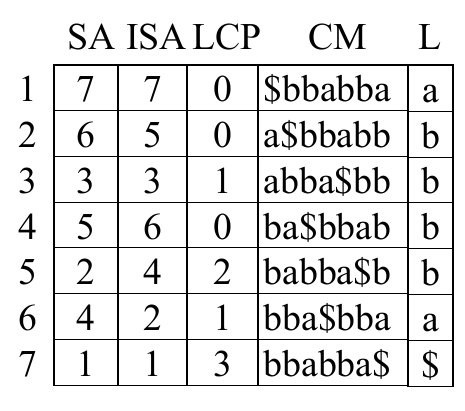}
\caption{Table illustrating the suffix array ($\SA$), inverse suffix array ($\ISA$), LCP array ($\LCP$), conceptual matrix ($\CM$), and BWT ($L$) of $T = \texttt{bbabba\$}$.}
 \label{fig:bwt}
\end{center}
\end{figure}

Let $\mathbb{N}_{0}$ be the set of non-negative integers.
An \emph{interval} $[b, e]$ for integers $b \leq e$ represents the set $\{b, b+1, \ldots, e \}$. 
For a sequence $S$ of values from $\mathbb{N}_{0}$ with length $d$, 
let $|S|$ be the length of $S$ (i.e., $|S| = d$). 
For an integer $i \in \{ 1, 2, \ldots, d \}$, 
let $S[i]$ denote the value at position $i$ in $S$. 
To simplify boundary conditions, 
we define $S[0] := S[|S|]$ and $S[|S| + 1] := S[0]$.  
For any $i, j \in \mathbb{N}_{0}$, if $1 \leq i \leq j \leq |S|$, then $S[i..j]$ denotes the subsequence $S[i], S[i+1], \ldots, S[j]$ of length $j - i + 1$;  
otherwise, $S[i..j]$ denotes the empty sequence of length $0$.

Let $T$ be a string of length $n$ over an alphabet $\Sigma = \{1, 2, \ldots, n^{\mathcal{O}(1)}\}$ of size $\sigma$.  
For $1 \leq i \leq j \leq n$, the substring $T[i..j]$ refers to the characters from position $i$ to $j$ in $T$.  
Let $P$ be a string of length $m$, called the \emph{pattern}.  
We assume that the special character $\$$, which is the smallest in $\Sigma$, appears only at the end of $T$ (i.e., $T[n] = \$$ and $T[i] \neq \$$ for $1 \leq i < n$).

For characters $c, c' \in \Sigma$, $c < c'$ means $c$ is smaller than $c'$.  
$T$ is lexicographically smaller than $P$ ($T \prec P$) if and only if one of the following holds:  
(i) there exists $i \in \{1, \ldots, n\}$ such that $T[1..(i-1)] = P[1..(i-1)]$ and $T[i] < P[i]$;  
(ii) $T$ is a proper prefix of $P$, i.e., $T = P[1..|T|]$ and $|T| < |P|$.

We use the base-2 logarithm throughout the paper.  
Our computation model is the standard unit-cost word RAM with a word size $B = \Theta(\log n)$ bits~\cite{DBLP:conf/stacs/Hagerup98}.  
Space complexity is measured in words; bit-level bounds can be obtained by multiplying by $B$.


\textbf{Suffix array, inverse suffix array, and LCP array.}
A suffix array~\cite{DBLP:journals/siamcomp/ManberM93} $\SA$ of $T$ is an integer sequence of length $n$ where  
$\SA[i]$ is the starting position of the $i$-th smallest suffix of $T$ in lexicographic order.  
Formally, $\SA$ is a permutation of $1, 2, \ldots, n$ satisfying $T[\SA[1]..n] \prec \cdots \prec T[\SA[n]..n]$.  
The elements of $\SA$ are called \emph{sa-values}.  
The inverse suffix array ($\ISA$) of $T$ is a sequence of length $n$ such that $\ISA[i]$ gives the position of suffix $T[i..n]$ in $\SA$;  
that is, $\ISA[\SA[i]] = i$ for all $i \in \{1, \ldots, n\}$.  
The \emph{LCP array} $\LCP$ is a sequence of length $n$ where $\LCP[i]$ is the length of the longest common prefix between $T[\SA[i]..n]$ and $T[\SA[i-1]..n]$.  
Formally, $\LCP[1] = 0$, and for $i \in \{2, \ldots, n\}$,  
$\LCP[i]$ is the largest $\ell \geq 0$ such that  
$T[\SA[i]..\SA[i] + \ell - 1] = T[\SA[i-1]..\SA[i-1] + \ell - 1]$.

\textbf{Conceptual matrix.}
For each $i \in \{1, \dots, n\}$, the $i$-th \emph{circular shift} of $T$ is $T^{[i]} = T[i..n]\,T[1..i-1]$. The \emph{conceptual matrix} $\CM$ is the $n \times n$ matrix of all such shifts sorted lexicographically, where the $i$-th row is $T^{[i]}$.
For any $i, j \in \{1, \dots, n\}$, the $i$-th circular shift precedes the $j$-th circular shift if and only if $i$ precedes $j$ in the suffix array of $T$, since the last character of $T$ is $\$$.



\textbf{BWT and RLBWT.}
BWT~\cite{burrows1994block} $L$ of $T$ is obtained by concatenating the last characters of all rows in the conceptual matrix $\CM$ of $T$ (i.e., $L = \CM[1][n], \CM[2][n], \ldots, \CM[n][n]$).  

A \emph{run} $T[i..j]$ is a maximal substring of repeated character $c$, i.e.,  
(i) $T[i..j]$ consists of $c$ only,  
(ii) $T[i-1] \neq c$ or $i = 1$, and  
(iii) $T[j+1] \neq c$ or $j = n$.  
The BWT $L$ consists of $r$ runs: $L[t_1..t_2{-}1], L[t_2..t_3{-}1], \dots, L[t_r..t_{r+1}{-}1]$ ($t_{1} = 1$ and $t_{r+1} = n+1$).  
The RLBWT, denoted $L_{\RLE}$, encodes each run as $(c_i, \ell_i)$,  
where $c_i$ is the character of the run and $\ell_i = t_{i+1} - t_i$ is its length. $L_{\RLE}$ uses $r(\log n + \log \sigma)$ bits.





Figure~\ref{fig:bwt} illustrates the suffix array, inverse suffix array, LCP array, conceptual matrix, and BWT of $T = \texttt{bbabba\$}$.  

\section{Review of r-index}\label{sec:review_r_index}
The r-index~\cite{10.1145/3375890} is a static self-index that supports both count and locate queries.  
A count query for $P$ on $T$ returns the number of its occurrences in $T$, while a locate query returns all positions where $P$ appears.

The r-index computes these queries using the RLBWT $L_{\RLE}$ of $T$ and two sampled suffix arrays, $\SA_s$ and $\SA_e$,  
obtained by sampling $\SA$ at the start and end positions of each run in the BWT $L$ of $T$.  
Given the $r$ runs $L[t_1..t_2{-}1], L[t_2..t_3{-}1], \ldots, L[t_r..t_{r+1}{-}1]$,  
the sampled arrays are defined as $\SA_s = \SA[t_1], \SA[t_2], \ldots, \SA[t_r]$ and  
$\SA_e = \SA[t_2{-}1], \SA[t_3{-}1], \ldots, \SA[t_{r+1}{-}1]$.

The r-index answers count and locate queries via backward search~\cite{10.1145/3375890,DBLP:journals/jacm/FerraginaM05},  
which computes the SA-interval $[sp, ep]$ such that $\SA[sp..ep]$ contains all occurrences of pattern $P$ in $T$.  
It supports count in $\mathcal{O}(m \log\log(\sigma + n/r))$ time and locate in $\mathcal{O}((m + \occ) \log\log(\sigma + n/r))$ time,  
using $\mathcal{O}(r)$ words, where $m = |P|$ and $\occ = ep - sp + 1$.  

\begin{figure}[t]
\begin{center}
	\includegraphics[width=0.4\textwidth]{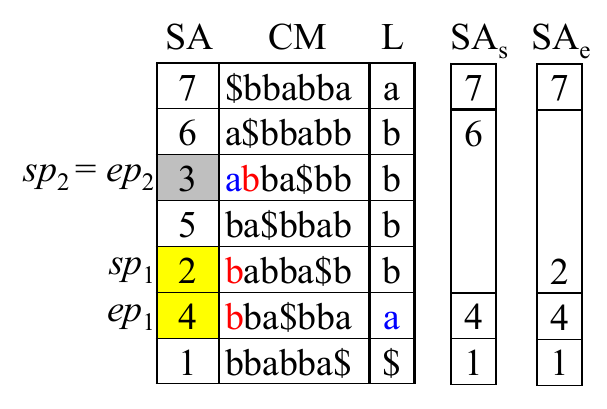}
\caption{(Left) Table illustrating the suffix array ($\SA$), conceptual matrix ($\CM$), 
and BWT ($L$) of the string $T$ used in Figure~\ref{fig:bwt}. 
The yellow rectangle represents the sa-interval $[sp_{1}, ep_{1}]$ of $\texttt{b}$. 
Similarly, the gray rectangle represents the sa-interval $[sp_{2}, ep_{2}]$ of $\texttt{ab}$. 
(Right) The sampled suffix arrays $\SA_{s}$ and $\SA_{e}$ for $\SA$.}
 \label{fig:bs}
\end{center}
\end{figure}

The details of the r-index is as follows. 
The algorithms for answering count and locate queries are described in Appendix~\ref{app:r_index_algo}.
The components of the r-index is described in Appendix~\ref{app:r_index_component}. 
The time complexity for count and locate queries is discussed in Section~\ref{app:r_index_time_space}. 

\subsection{Details of Count and Locate Queries.}\label{app:r_index_algo}
{\bf Count query.} 
The count query is computed by the backward search technique~\cite{10.1145/3375890,DBLP:journals/jacm/FerraginaM05}, which determines the interval $[sp, ep]$ (known as the sa-interval of $P$) such that $\SA[sp..ep]$ contains all occurrence positions of $P$ in $T$.
The backward search is computed using \emph{rank} and \emph{lex-count} queries on $L_{\RLE}$, which are defined as follows.
\begin{itemize}
    \item A rank query $\rank(L_{\RLE}, i, c)$ returns the number of occurrences of a character $c$ in 
    the prefix $L[1..i]$ of the BWT $L$ represented as $L_{\RLE}$, 
    i.e., $\rank(L_{\RLE}, i, c) = |\{ j \in \{ 1, 2, \ldots, i \} \mid L[j] = c \}|$.
    \item A lex-count query $\lexCount(L_{\RLE}, c)$ returns the number of characters smaller than a given character $c \in \Sigma$ in $L$ (i.e. $\lexCount(L_{\RLE}, c) = |\{ 1 \leq i \leq n \mid L[i] < c \}|$). 
\end{itemize}

The sa-interval of $P$ can be computed iteratively as follows:
(i) starting from the sa-interval $[sp_1,ep_1]$ of $P[m]$ as $sp_{1} = \lexCount(L_{\RLE}, P[m]) + 1$ and $ep_{1} = \lexCount(L_{\RLE}, P[m]) + \rank(L_{\RLE}, n, P[m])$; 
(ii) for $i \geq 2$, the sa-interval $[sp_i, ep_i]$ of $P[(m-i+1)..m]$ is computed as 
$sp_{i} = \lexCount(L_{\RLE}, P[m-i+1]) + \rank(L_{\RLE}, sp_{i-1} - 1, P[m-i+1]) + 1$ 
and 
$ep_{i} = \lexCount(L_{\RLE}, P[m-i+1]) + \rank(L_{\RLE}, ep_{i-1}, P[m-i+1])$, 
where $m$ is the length of $P$. 
Finally, $[sp_{m}, ep_{m}]$ is returned as the sa-interval of $P$. 
For more details on the backward search, refer to \cite{DBLP:journals/jacm/FerraginaM05}. 


\textbf{Locate query.}
The locate query is computed using the function $\phi^{-1}$ that returns $\phi^{-1}(\SA[i]) = \SA[i+1]$ for $i \in \{ 1, 2, \ldots, n \}$. 
The following lemma ensures that $\phi^{-1}$ function can be computed on $\SA_{s}$ and $\SA_{e}$ without $\SA$, as follows:

\begin{lemma}[\cite{10.1145/3375890}]\label{lem:phi_inv}
For an integer $i \in \{ 1, 2, \ldots, n \}$, 
let $j$ be the position in $\SA_{e}$ such that $\SA_{e}[j]$ is the largest value among ones less than  $(i+1)$ in $\SA_{e}$ (i.e., $j$ is the position of $u$ in $\SA_{e}$, where $u = \max \{ \SA_{e}[k] \mid 1 \leq k \leq r \text{ s.t. }  \SA_{e}[k] < i+1 \}$). 
For simplicity, we define $\SA_{s}[r+1] = \SA_{s}[1]$. 
Then, $\phi^{-1}(i) = (i - \SA_{e}[j]) + \SA_{s}[j+1]$. 
\end{lemma}

The answer to the locate query for $P$ is $\SA[sp]$, $\SA[sp+1]$,...,$\SA[ep]$ 
for the sa-interval $[sp,ep]$ of $P$.
These positions are computed as follows: 
(i) the sa-interval $[sp, ep]$ and the first position $\SA[sp]$ are computed using the result of the backward search 
(refer to \cite{DBLP:journals/tcs/BannaiGI20} for the detailed algorithm computing $\SA[sp]$); 
(ii) for $i = 1, 2, \ldots, (ep - sp)$, 
$\SA[sp + i]$ is iteratively computed by applying $\phi^{-1}$ to $\SA[sp+i-1]$. 
Therefore, count and locate queries can be supported using $L_{\RLE}$, $\SA_{s}$, and $\SA_{e}$. 

Figure~\ref{fig:bs} illustrates the following: 
the sa-interval $[sp_{1}, ep_{1}] = [5, 6]$ of character $\texttt{b}$; 
the sa-interval $[sp_{2}, ep_{2}] = [3, 3]$ of string $\texttt{ab}$; 
the two sampled suffix arrays $\SA_{s} = 7, 6, 4, 1$ and $\SA_{e} = 7, 2, 4, 1$ created by sampling from the suffix array shown in Figure~\ref{fig:bwt}. 
Using the backward search, 
the two positions $sp_{2}$ and $ep_{2}$ can be computed as 
$sp_{2} = \lexCount(L_{\RLE}, \texttt{a}) + \rank(L_{\RLE}, sp_{1} - 1, \texttt{a}) + 1$ 
and 
$ep_{2} = \lexCount(L_{\RLE}, \texttt{a}) + \rank(L_{\RLE}, ep_{1}, \texttt{a})$, respectively.
$\phi^{-1}(5)$ can be computed as follows: 
Let $j$ be the position in $\SA_{e}$ such that $\SA_{e}[j]$ is the largest value in $\SA_{e}$ less than $6$. 
Here, $j = 3$, $\SA_{e}[j] = 4$, and $\SA_{s}[j+1] = 1$. 
$\phi^{-1}(5) = 2$ is obtained from $\phi^{-1}(5) = (5 - \SA_{e}[j]) + \SA_{s}[j+1]$. 

\subsection{Three Components of r-index}\label{app:r_index_component}
For simplicity, 
we slightly modify the data structures used in the r-index 
under the assumption that each character in $\Sigma$ occurs in $T$. 
This index consists of three components of data structures:  
$D_{\RLE}(L_{\RLE})$, $D_{\DI}(\SA_{s})$, and $D_{\DI}(\SA_{e})$. 

\textbf{$D_{\RLE}(L_{\RLE})$.}
Four sequences $S_1$, $S_2$, $S_3$ and $S_4$ are created from $L_{\RLE} = (c_{1}, \ell_{1}), (c_{2}, \ell_{2})$, $\ldots$, $(c_{r}, \ell_{r})$ as follows: 
\begin{align*}
S_{1} &= c_{1}, c_{2}, \ldots, c_{r}; \\
S_{2} &= t_{1}, t_{2}, \ldots, t_{r}; \\
S_{3} &= \rank(L_{\RLE}, t_{1} + \ell_{1} - 1, c_{1}), \rank(L_{\RLE}, t_{2} + \ell_{2} - 1, c_{2}), \ldots, \rank(L_{\RLE}, t_{r} + \ell_{r} - 1, c_{r}); \\
S_{4} &= \lexCount(L_{\RLE}, 1), \lexCount(L_{\RLE}, 2), \ldots, \lexCount(L_{\RLE}, \sigma). 
\end{align*}
Here, each $t_{i}$ is the starting position of the $i$-th run $(c_{i}, \ell_{i})$ in $L$. 
These sequences are represented as four arrays. 
In addition, we build two component data structures that organize $D_{\RLE}(L_{\RLE})$. 
The first component is Belazzougui and Navarro's static data structure~\cite{DBLP:journals/talg/BelazzouguiN15} 
that supports access, rank, and select queries on the string $S_{1}$. 
Here, these three queries are defined as follows: 
\begin{itemize}
    \item 
    Similar to $\rank(L_{\RLE}, i, c)$, 
    a rank query $\rank(S_{1}, i, c)$ returns the number of occurrences of a character $c$ in $S_{1}[1..i]$. 
    \item A select query $\select(S_{1}, i, c)$ returns the position of the $i$-th occurrence of $c$ in $S_{1}$~(i.e., it returns the smallest integer $j \geq 1$ such that $\sum_{k=1}^{j} I\{S_{1}[k] = c \}$ is equal to $i$, where $I\{\cdot\}$ returns $1$ if and only if the condition inside $\{ \cdot \}$ is true).
    If such integer $i$ does not exist, then this query returns $-1$. 
    \item An access query $\mathsf{access}(S_{1}, i)$ returns the $i$-th character of $S_{1}$.     
\end{itemize}

The second component is Belazzougui and Navarro's static data structure~\cite{DBLP:journals/talg/BelazzouguiN15} 
that supports \emph{predecessor} query on $S_{2}$. 
A predecessor query $\predecessor(S_{2}, i)$ on the sorted sequence $S_{2}$ returns 
the largest integer $j \geq 1$ satisfying $S_{2}[j] < i$. 
These four arrays and two static data structures require $\mathcal{O}(r + \sigma)$ words. 
Here, $\sigma \leq r$ follows from the assumption.

\textbf{$D_{\DI}(\SA_{s})$.}
$D_{\DI}(\SA_{s})$ is an array that stores the values in $\SA_{s}$. 

\textbf{$D_{\DI}(\SA_{e})$.}
$D_{\DI}(\SA_{e})$ has three sequences $\SA_{e}$, $\Pi$, and $S'$. 
Here, $\Pi$ is the sorted sequence of $r$ integers $1, 2, \ldots, r$ satisfying $\SA_{e}[\Pi[1]] < \SA_{e}[\Pi[2]] < \cdots < \SA_{e}[\Pi[r]]$; 
$S' = \SA_{e}[\Pi[1]], \SA_{e}[\Pi[2]], \ldots, \SA_{e}[\Pi[r]]$. 
These three sequences are represented as three arrays, 
and 
we build the Belazzougui and Navarro's static data structure on $S'$ to support predecessor query. 
These three arrays and the static data structure require $\mathcal{O}(r)$ words. 
Therefore, the r-index can be stored in $\mathcal{O}(r)$ words. 

\subsection{Time Complexity for Count and Locate Queries.}\label{app:r_index_time_space}
As already explained, 
the r-index computes the sa-interval $[sp, ep]$ of a given pattern $P$ using rank and lex-count queries on $L_{\RLE}$. 
These queries are supported using $D_{\RLE}(L_{\RLE})$ as follows.

\textbf{Supporting $\rank(L_{\RLE}, i, c)$.}
The following equation follows from the fact that $L_{\RLE}$ is the run-length encoding of $L$. 
\begin{align}    
  \rank(L_{\RLE}, i, c) &= 
  \begin{cases}
    S_{3}[x] - (S_{2}[x+1] - 1 - i) & \text{if $c_{x} = c$} \\
    S_{3}[x'] & \text{otherwise} 
  \end{cases}
\end{align}
where $x = \predecessor(S_{2}, i+1)$, $c_{x} = \mathsf{access}(S_{1}, x)$, $k = \rank(S_{1}, x, c)$, 
and $x' = \select(S_{1}, k, c)$. 
Therefore, the rank query on $L_{\RLE}$ can be supported using access, rank, select, and predecessor queries on $S_{1}$ and $S_{2}$. 

\textbf{Supporting $\lexCount(L_{\RLE}, c)$.}
The answer to this query is stored as the $c$-th integer in $S_{4}$. 
Therefore, we can support any lex-count query in $\mathcal{O}(1)$ time. 

Next, the r-index answers locate query using $\SA[sp]$ and $\phi^{-1}(i)$ after computing the sa-interval $[sp, ep]$ of $P$. 
$\SA[sp]$ and $\phi^{-1}(i)$ are computed as follows.

\textbf{Computing $\SA[sp]$.}
Consider the $m$ sa-intervals $[sp_{1}, ep_{2}]$, $[sp_{1}, ep_{2}]$, $\ldots$, $[sp_{m}, ep_{m}]$ computed by the backward search, where $m$ is the length of $P$. 
The following lemma is called \emph{toehold lemma}, which states the relationship among 
$\SA[sp_{i+1}]$, $\SA[sp_{i}]$ and $\SA_{s}$ for each integer $i \in \{ 1, 2, \ldots, m-1 \}$. 
\begin{lemma}[Section 2.3 in \cite{DBLP:journals/tcs/BannaiGI20}]
    For each integer $i \in \{ 1, 2, \ldots, m-1 \}$, 
    let $(c_{v}, \ell_{v})$ be the run that contains the $sp_{i}$-th character in $L$. 
    If $P[m - i] = c_{v}$, 
    then $\SA[sp_{i+1}] = \SA[sp_{i}] - 1$; 
    otherwise, $\SA[sp_{i}] = \SA_{s}[v'] - 1$, 
    where $v' \in \{ 1, 2, \ldots, r \}$ is the smallest integer 
    satisfying $v' \geq v+1$ and $c_{v'} = P[m - i]$.
\end{lemma}

The following corollary follows from the toehold lemma. 
\begin{corollary}\label{cor:toe}
Let $i' \in \{ 1, 2, \ldots, m-1 \}$ be the largest integer satisfying 
$P[m - i'] \neq c_{v}$, 
where $(c_{v}, \ell_{v})$ is the run that contains the $sp_{i'}$-th character in $L$.
If such integer $i'$ exists, 
then $\SA[sp] = \SA_{s}[v'] - (m - i')$, 
where $v' \in \{ 1, 2, \ldots, r \}$ is the smallest integer 
satisfying $v' \geq v+1$ and $P[m - i'] = c_{v'}$. 
Otherwise, 
$\SA[sp] = \SA_{s}[u] - m$, 
where $u \in \{ 1, 2, \ldots, r \}$ is the smallest integer 
satisfying $c_{u} = P[m]$. 
\end{corollary}

Corollary~\ref{cor:toe} shows that $\SA[sp]$ can be computed using the largest integer $i'$ 
satisfying $P[m - i'] \neq c_{v}$, 
where $(c_{v}, \ell_{v})$ is the run that contains the $sp_{i'}$-th character in $L$.
Such integer $i'$ can be found using $\mathcal{O}(1)$ queries on $S_{1}$ and $S_{2}$. 
After computing the integer $i'$, 
$v'$ can be computed by $\select(S_{1}, 1 + d, P[m - i'])$, 
where $d = \rank(S_{1}, v, P[m - i'])$. 
Therefore, $\SA[sp]$ can be computed using $D_{\RLE}(L_{\RLE})$ and $D_{\DI}(\SA_{s})$. 

\textbf{Supporting $\phi^{-1}(i)$.}
Lemma~\ref{lem:phi_inv} shows that 
$\phi^{-1}(i)$ can be computed using $\SA_{e}[j]$ and $\SA_{s}[j+1]$.
$j$ is computed by $\Pi[\predecessor(S', i+1)]$. 
Therefore, $\phi^{-1}(i)$ can be computed using $D_{\DI}(\SA_{s})$ and $D_{\DI}(\SA_{e})$. 

Finally, the r-index can support count query in $\mathcal{O}(m \log\log(\sigma + n/r))$ time 
and locate query in $\mathcal{O}((m + \occ) \log\log(\sigma + n/r))$ time. 
See \cite{10.1145/3375890,DBLP:journals/tcs/BannaiGI20} for more details on the r-index.

\section{Dynamic r-index: Design and Implementation}\label{sec:dynamic_r_index_unified}
We present the \emph{dynamic r-index}, which extends the static r-index to support efficient updates while maintaining fast query performance. This section describes the complete design, from the underlying data structures to the update algorithms and their time complexities.

\subsection{Design Overview and Challenges}\label{subsec:design_overview}
The static r-index uses three key components: $L_{\RLE}$, $\SA_s$, and $\SA_e$. 
To make this index dynamic, we face two fundamental challenges:

\textbf{RLBWT updates.}
When a character is inserted or deleted, the BWT changes, which may split existing runs, merge adjacent runs, extend or shrink runs, or create new runs. Each of these operations must be performed efficiently while maintaining the compressed representation $L_{\RLE}$.

\textbf{Sampled SA consistency.}
$\SA_s$ and $\SA_e$ are intrinsically linked to the run structure of the BWT. When runs change, new runs require new sampled positions, and deleted runs require removing their sampled positions. 
Similar to the RLBWT update, 
each of these operations must be performed efficiently while maintaining the sampled SAs.

\textbf{Computational time.}
For achieving an LCP bounded update time, 
each update time of RLBWT and sampled SAs is bounded by $\mathcal{O}(\log r)$. In addition, the number of iterations for updating RLBWT and sampled SA should be bounded by the average length of LCPs.

\textbf{Our approach.}
We address these challenges by designing (i) an efficient update algorithm of $L_{\RLE}$, $\SA_s$, and $\SA_e$ and (ii) 
three coordinated dynamic data structures: $\mathscr{D}_{\RLE}(L_{\RLE})$ for managing 
the RLBWT, and $\mathscr{D}_{\DI}(\SA_s)$, $\mathscr{D}_{\DI}(\SA_e)$ for maintaining the sampled SAs. 
Each structure supports efficient operations in $\mathcal{O}(\log r)$ time, and the key insight is that updates to the input string $T$ typically affect only a small number of runs if the average LCP value is small.

\subsection{Dynamic Data Structures}\label{subsec:dynamic_r_index}
We present the \emph{dynamic r-index}, composed of three dynamic structures:  
$\mathscr{D}_{\RLE}(L_{\RLE})$, $\mathscr{D}_{\DI}(\SA_{s})$, and $\mathscr{D}_{\DI}(\SA_{e})$.  
These data structures require $\mathcal{O}(r)$ words of space for the $r$ runs in $L_{\RLE}$ 
and support queries on $L_{\RLE}$, $\SA_{s}$, and $\SA_{e}$ in $\mathcal{O}(\log r)$ time. 
Therefore, 
the dynamic r-index uses $\mathcal{O}(r)$ words of space, matching the static r-index and making it practical for highly repetitive strings where $r \ll n$.

\subsubsection{Details of \texorpdfstring{$\mathscr{D}_{\RLE}(L_{\RLE})$}{DRLE}}\label{app:RLE_dynamic_data_structure}
\begin{table}[htb]
\caption{Summary of the six queries supported by $\mathscr{D}_{\RLE}(L_{\RLE})$. Here, $L$ is the BWT of length $n$ represented as $L_{\RLE}$ of $r$ runs; $i \in \{ 1, 2, \ldots, n \}$ is an integer; 
$j \in \{ 1, 2, \ldots, r \}$ is an integer.}
\label{tab:queries_on_L}

\begin{tabular}{|c|l|c|}
\hline
Query & Description & Time complexity \\
\hline
$\rank(L_{\RLE}, i, c)$        & return the number of occurrences of $c$ in $L[1..i]$       & \multirow{9}{*}{$\mathcal{O}(\log r / \log\log r)$} \\
\multirow{2}{*}{$\select(L_{\RLE}, i, c)$}  & return the position of the $i$-th occurrence of $c$ in $L$           &  \\
 & If no such occurrence exists, it returns $-1$.           &  \\

$\lexCount(L_{\RLE}, c)$     & return the number of characters smaller than $c \in \Sigma$ in $L$         &  \\
\multirow{2}{*}{$\lexSearch(L_{\RLE}, i)$}    & return a character $c \in \Sigma$ that satisfies    &  \\
                            & $\lexCount(L, c) < i \leq \lexCount(L, c) + \rank(L, |L|, c)$   &  \\

$\runAccess(L_{\RLE}, j)$         & return the $j$-th run and its starting position in $L$   &  \\
\multirow{2}{*}{$\runIndex(L_{\RLE}, i)$}            & return the index $x$ of the run $(c_{x}, \ell_{x})$ that       &  \\
                    & contains the $i$-th character in $L$      &  \\
\hline
\end{tabular}
\end{table}

\begin{table}[htb]
\caption{Summary of the four update operations supported by $\mathscr{D}_{\RLE}(L_{\RLE})$. 
Here, $L_{\RLE} = (c_{1}, \ell_{1}), (c_{2}, \ell_{2}), \ldots, (c_{r}, \ell_{r})$; 
$c \in \Sigma$ is a character;
$i \in \{ 1, 2, \ldots, r \}$ is an integer; 
$t \in \{ 1, 2, \ldots, \ell_{i} - 1 \}$ is an integer.}
\label{tab:operations_on_L}
\begin{tabular}{|c|l|c|}
\hline
Operation & Description & Time complexity \\
\hline
\multirow{2}{*}{$\insertRLE(L_{\RLE}, c, i)$}         & insert $c$ into $L_{\RLE}$ as a new $i$-th run $(c, 1)$ if $c \neq c_i$;      & \multirow{7}{*}{$\mathcal{O}(\log r / \log\log r)$} \\
         & otherwise, increment the length of the $i$-th run by $1$      &  \\

\multirow{2}{*}{$\deleteRLE(L_{\RLE}, i)$}      & delete the $i$-th run $(c_i, \ell_i)$ if $\ell_i = 1$;           &  \\
      & otherwise, decrement $\ell_i$ by $1$           &  \\

$\splitRLE(L_{\RLE}, i, t)$     & split the $i$-th run $(c_{i}, \ell_{i})$ into two runs $(c_{i}, t)$ and $(c_{i}, \ell_{i} - t)$         &  \\
\multirow{2}{*}{$\mergeRLE(L_{\RLE}, i)$}    & merge two consecutive runs $(c_{i}, \ell_{i})$ and $(c_{i+1}, \ell_{i+1})$ into      &  \\
    & a new run $(c_{i}, \ell_{i} + \ell_{i+1})$ if $c_{i} = c_{i+1}$     &  \\

\hline
\end{tabular}
\end{table}

\begin{table}[htb]
\caption{Summary of the two update operations supported by the dynamic data structure proposed by Munro and Nekrich~\cite{DBLP:conf/esa/MunroN15}, which is built on $S_1$ of $r$ characters. 
Here, $c \in \Sigma$ is a character; $i \in \{ 1, 2, \ldots, r \}$ is an integer.}
\label{tab:operations_on_S}
\begin{tabular}{|c|l|c|}
\hline
Operation & Description & Time complexity \\
\hline
$\insertL(S_{1}, i, c)$         & insert $c$ into $S_{1}$ at position $i$;      & \multirow{2}{*}{$\mathcal{O}(\log r / \log\log r)$} \\
$\deleteL(S_{1}, i)$         & delete the character at position $i$ from $S_{1}$      &  \\
\hline
\end{tabular}
\end{table}

\begin{table}[htb]
\caption{Summary of the two queries supported by the dynamic data structure proposed by Bille et al.~\cite{DBLP:journals/algorithmica/BilleCCGSVV18}, which is built on $S$. Here, $i$ is a position in $S$; $t \geq 0$ is an integer.}
\label{tab:queries_PS}

\begin{tabular}{|c|l|c|}
\hline
Query & Description & Time complexity \\
\hline
$\sumPS(S, i)$        & return the sum of the first $i$ elements in $S$       & \multirow{3}{*}{$\mathcal{O}(\log |S| / \log (B/\delta))$} \\
\multirow{2}{*}{$\searchPS(S, t)$}  & return the smallest integer $1 \leq i \leq |S|$ such that $\sum_{j = 1}^{i} S[j] \geq t$          &  \\
  &     If no such $i$ exists, it returns $|S| + 1$.      &  \\

\hline
\end{tabular}
\end{table}

\begin{table}[htb]
\caption{Summary of the four update operations supported by the dynamic data structure proposed by Bille et al.~\cite{DBLP:journals/algorithmica/BilleCCGSVV18}, which is built on $S$. Here, $\delta \geq 1$ is a user-defined parameter; $i$ is a position in $S$; $\Delta \in \{ 0, 1, \ldots, 2^{\delta}-1 \}$ is an integer smaller than $2^{\delta}$; $s \in \{ 0, 1, \ldots, S[i] \}$ is an integer smaller than $(S[i]+1)$; $B$ is machine word size.}
\label{tab:operations_on_PS}
\begin{tabular}{|c|l|c|}
\hline
Operation & Description & Time complexity \\
\hline
$\insertPS(S, i, \Delta)$         & insert $\Delta$ into $S$ at position $i$      & \multirow{4}{*}{$\mathcal{O}(\log |S| / \log (B/\delta))$} \\
$\deletePS(S, i)$         & delete the element $S[i]$ at position $i$ from $S$ if $S[i] \leq 2^{\delta}-1$      &  \\
$\mergePS(S, i)$         & merge $S[i]$ and $S[i+1]$ into $(S[i] + S[i+1])$      &  \\
$\dividePS(S, i, s)$         & divide $S[i]$ into $s$ and $(S[i]-s)$ in $S$     &  \\
\hline
\end{tabular}
\end{table}

The dynamic structure $\mathscr{D}_{\RLE}(L_{\RLE})$ supports six queries on $L_{\RLE}$:  
\emph{rank}, \emph{select}, \emph{lex-count}, \emph{lex-search}, \emph{run-access}, and \emph{run-index},  
and four update operations: \emph{insertion}, \emph{deletion}, \emph{split}, and \emph{merge}. 
These queries and update operations are summarized in Table~\ref{tab:queries_on_L} and Table~\ref{tab:operations_on_L}. 

We construct four sequences from $L_{\RLE}$ and build a dynamic data structure for each sequence  
to support the six queries and four update operations in $\mathscr{D}_{\RLE}(L_{\RLE})$ as described below.
Let $\Pi$ be the sorted sequence of $r$ integers $1, 2, \ldots, r$ such that any pair of two integers $\Pi[i]$ and $\Pi[j]$ ($i < j$) 
satisfies either of the following two conditions:
(a) $c_{\Pi[i]} < c_{\Pi[j]}$, or
(b) $c_{\Pi[i]} = c_{\Pi[j]}$ and $\Pi[i] < \Pi[j]$. 
Sequences $S_1$, $S_2$, $S_3$ and $S_4$ are defined as follows: 
\begin{align*}
S_{1} &= c_{1}, c_{2}, \ldots, c_{r}; \\
S_{2} &= \ell_{1}, \ell_{2}, \ldots, \ell_{r}; \\
S_{3} &= \ell_{\Pi[1]}, \ell_{\Pi[2]}, \ldots, \ell_{\Pi[r]}; \\
S_{4} &= (c_{\Pi[1]} - 0), (c_{\Pi[2]} - c_{\Pi[1]}), \ldots, (c_{\Pi[r]} - c_{\Pi[r-1]}), ((\max \Sigma) - c_{\Pi[r]}). 
\end{align*}

To realize $\mathscr{D}_{\RLE}(L_{\RLE})$, we construct two dynamic data structures on the sequences $S_1$, $S_2$, $S_3$, and $S_4$.
As the first dynamic data structure, we use the one proposed by Munro and Nekrich~\cite{DBLP:conf/esa/MunroN15}, which is built on $S_1$.  
It supports access, rank and select queries on $S_1$, as well as \emph{insertion} and \emph{deletion} operations, as described in Table~\ref{tab:operations_on_S}.
This dynamic data structure supports all the queries in $\mathcal{O}(\log r / \log \log r)$ time while using $r \log \sigma + o(r \log \sigma)$ bits of space. 

As the second dynamic data structure, we use the one proposed by Bille et al.~\cite{DBLP:journals/algorithmica/BilleCCGSVV18},  
which is applied to the sequences $S_2$, $S_3$, and $S_4$.  
It supports \emph{sum} and \emph{search} queries, as well as four update operations \emph{insertion}, \emph{deletion}, \emph{split}, and \emph{merge} on an integer sequence $S$, as described in Table~\ref{tab:queries_PS} and Table~\ref{tab:operations_on_PS}.

The complexity of all these queries is $\mathcal{O}(\log |S| / \log (B/\delta))$ time and 
$\mathcal{O}(|S|)$ words of space, where $B$ is machine word size and $\delta \geq 1$ is a user-defined parameter. 

We construct the Bille et al.'s data structure by setting $\delta=1$ for each of $S_2$, $S_3$ and $S_4$. This configuration supports the queries and the four update operations in $\mathcal{O}(\log r / \log B)$ time, while using $\mathcal{O}(r)$ words of space. 
Since $r \leq n$ and $B = \Theta(\log n)$, 
the time complexity $\mathcal{O}(\log r / \log B)$ can be bounded by $\mathcal{O}(\log r / \log \log r)$. 

The lemmas below provide time complexities for the queries and updates supported by $\mathscr{D}_{\RLE}(L_{\RLE})$.
\begin{lemma}\label{lem:dyn_rle_update}
$\mathscr{D}_{\RLE}(L_{\RLE})$ supports 
(i) rank, select, lex-count, lex-search, run-access, and run-index queries, 
and (ii) insertion, deletion, split, and merge operations in $\mathcal{O}(\log r / \log \log r)$ time.
\end{lemma}
\begin{proof}
The proof of Lemma~\ref{lem:dyn_rle_update} is as follows. 

\textbf{Supporting rank, select, lex-count, lex-search, run-access, and run-index queries.}
If one can perform (i) access, rank, and select queries on $S_{1}$, and (ii) sum and search queries on each of $S_{2}$, $S_{3}$, and $S_{4}$, each in $\mathcal{O}(\beta)$, then the six queries on the RLBWT $L_{\RLE}$ --- namely, $\rank(L_{\RLE}, i, c)$, $\select(L_{\RLE}, i, c)$, $\lexCount(L_{\RLE}, c)$, $\lexSearch(L_{\RLE}, j)$, $\runAccess(L_{\RLE}, i)$, and $\runIndex(L_{\RLE}, j)$ --- can be answered in $\mathcal{O}(\beta)$ time \cite{DBLP:journals/jda/OhnoSTIS18,DBLP:journals/algorithmica/PolicritiP18}.
The two dynamic data structures of $\mathscr{D}_{\RLE}(L_{\RLE})$ support (i) access, rank, and select queries on $S_1$, and (ii) sum and search queries on each of $S_2$, $S_3$, and $S_4$, each in time $\beta = \mathcal{O}(\log r/\log \log r)$. 
Thus, $\mathscr{D}_{\RLE}(L_{\RLE})$ can answer all six queries in $\mathcal{O}(\log r/\log \log r)$ time.

\textbf{Supporting $\insertRLE(L_{\RLE}, c, i)$.}
If either $i = n+1$ or $c_{i} \neq c$, 
then a new run $(c, 1)$ inserted into $L_{\RLE}$ as the $i$-th run. 
In this case, 
we update the four sequences $S_{1}$, $S_{2}$, $S_{3}$, $S_{4}$ as follows: 
\begin{itemize}
    \item Insert $c$ into $S_{1}$ as the $i$-th character.
    \item Insert $1$ into $S_{2}$ as the $i$-th integer.
    \item Insert $1$ into $S_{3}$ at an appropriate position $j$. 
    $j$ is equal to largest integer $j' \geq 2$ satisfying 
    satisfies either (a) $c_{\Pi[j'-1]} < c$ or (b) $c_{\Pi[j'-1]} = c$ and $\Pi[j'-1]] < i$. If such $j'$ does not exist, then $j'$ is defined as 1. 
    We can show that 
    $j' = 1 + \lexCount(S_{1}, c) + \rank(S_{1}, i-1, c)$ using the lex-count and rank queries on $S_{1}$. 
    Here, the lex-count query $\lexCount(S_{1}, c)$ can be answered by computing $(\searchPS(S_{4}, c) - 1)$, 
    the rank query is supported by the Munro and Nekrich's dynamic data structure built on $S_{1}$. 
    \item     
    Divide the $j$-th integer $(c_{\Pi[j]} - c_{\Pi[j-1]})$ into two integers 
    $(c - c_{\Pi[j-1]})$ and $(c_{\Pi[j]} - c)$ by $\dividePS(S_{4}, u, (c - c_{\Pi[j-1]}))$.    
\end{itemize}
Otherwise (i.e., $1 \leq i \leq n$ and $c_{i} = c$), 
the length $\ell_{i}$ of the $i$-th run is incremented by 1. 
In this case, 
we update the two sequences $S_{2}$ and $S_{3}$ as follows: 
(i) Increment the $i$-th integer in $S_{2}$ by 1; 
(ii) Increment the $i'$-th integer in $S_{3}$ by 1, where $i'$ is the position of $i$ in $\Pi$. 

\textbf{Supporting $\deleteRLE(L_{\RLE}, i)$ .}
If the length of the $i$-th run $(c_{i}, \ell_{i})$ is 1, 
then this run is deleted from $L_{\RLE}$. 
In this case, we update the four sequences $S_{1}$, $S_{2}$, $S_{3}$, $S_{4}$ as follows: 
\begin{itemize}
    \item Delete the $i$-th character from $S_{1}$.
    \item Delete the $i$-th integer from $S_{2}$.
    \item Delete the $i'$-th integer from $S_{3}$, 
    where $i'$ is the position of $i$ in $\Pi$. 
    Similar to the integer $j$, 
    $i'$ can be obtained by computing $\lexCount(S_{1}, S_{1}[i]) + \rank(S_{1}, i, S_{1}[i])$.     
    \item Merge the $i'$-th and $(i'+1)$-th integers in $S_{4}$ by $\mergePS(S_{4}, i')$. 
\end{itemize}
Otherwise (i.e., $\ell_{i} \geq 2$), 
the length $\ell_{i}$ of the $i$-th run is decremented by 1. 
In this case, 
we update the two sequences $S_{2}$ and $S_{3}$ as follows: 
(i) Decrement the $i$-th integer in $S_{2}$ by 1; 
(ii) Decrement the $i'$-th integer in $S_{3}$ by 1, where $i'$ is the position of $i$ in $\Pi$. 


\textbf{Supporting $\splitRLE(L_{\RLE}, i, t)$.}
We update the four sequences $S_{1}$, $S_{2}$, $S_{3}$, $S_{4}$ as follows: 
\begin{itemize}
    \item Insert $c$ into $S_{1}$ as the $(i+1)$-th character.
    \item Divide the $i$-th integer in $S_{2}$ into $t$ and $(\ell_{i} - t)$ by $\dividePS(S_{2}, i, t)$.
    \item Divide the $i'$-th integer in $S_{2}$ into $t$ and $(\ell_{i} - t)$ by $\dividePS(S_{2}, i, t)$, 
    where $i'$ is the position of $i$ in $\Pi$. 
    \item Insert $0$ into $S_{4}$ as the $(i'+1)$-th integer. 
\end{itemize}

\textbf{Supporting $\mergeRLE(L_{\RLE}, i)$.}
We update the four sequences $S_{1}$, $S_{2}$, $S_{3}$, $S_{4}$ as follows: 
\begin{itemize}
    \item Delete the $(i+1)$-th integer from $S_{1}$.
    \item Merge the $i$-th and $(i+1)$-th integers in $S_{2}$ by $\mergePS(S_{2}, i)$.
    \item Merge the $i'$-th and $(i'+1)$-th integers in $S_{3}$ by $\mergePS(S_{3}, i')$, where $i'$ is the position of $i$ in $\Pi$. 
    \item Delete the $(i'+1)$-th integer from $S_{4}$.
\end{itemize}
These six operations are performed using $\mathcal{O}(1)$ queries and operations on $S_{1}$, $S_{2}$, $S_{3}$, and $S_{4}$. 
Therefore, they take $\mathcal{O}(\log r / \log \log r)$ time in total. 
\end{proof}




\subsubsection{Details of \texorpdfstring{$\mathscr{D}_{\DI}(\SA_{s})$}{DDI} and \texorpdfstring{$\mathscr{D}_{\DI}(\SA_{e})$}{DDI}}\label{app:dynamic_DDI}
\begin{table}[htb]
\caption{Summary of the three queries supported by $\mathscr{D}_{\DI}(\SA_{a})$. Here, $\SA_{a}$ consists of $r$ distinct integers; $i$ is a position in $\SA_{a}$; $t \geq 0$ is an integer.}
\label{tab:queries_on_SA}

\begin{tabular}{|c|l|c|}
\hline
Query & Description & Time complexity \\
\hline
$\accessSA(\SA_{a}, i)$        & return $\SA_{a}[i]$       & \multirow{3}{*}{$\mathcal{O}(\log r)$} \\
$\orderSA(\SA_{a}, i)$  & return the position of the $i$-th smallest integer in $\SA_{a}$          &  \\
$\countSA(\SA_{a}, t)$  & return the number of integers in $\SA_{a}$ that are smaller than $t$          &  \\
\hline
\end{tabular}
\end{table}

\begin{table}[htb]
\caption{Summary of the four update operations supported by $\mathscr{D}_{\DI}(\SA_{a})$. Here, $\SA_{a}$ consists of $r$ distinct integers; $i$ is a position in $\SA_{a}$; $t \geq 0$ is an integer; $k \geq 0$ is an integer. 
The Decrement operation $\decrementSA(\SA_{a}, t, k)$ is performed if $\Pi_{a}$ is not changed by this operation. 
}
\label{tab:operations_on_SA}
\begin{tabular}{|c|l|c|}
\hline
Operation & Description & Time complexity \\
\hline
\multirow{2}{*}{$\insertSA(\SA_{a}, i, t)$}  & insert $t$ into $\SA_{a}$ as the $i$-th value if & \multirow{7}{*}{$\mathcal{O}(\log r)$} \\
  & $t \not \in \{ \SA_{a}[1], \SA_{a}[2], \ldots, \SA_{a}[r] \}$ &  \\

$\deleteSA(\SA_{a}, i)$         & delete the $i$-th integer from $\SA_{a}$      &  \\
\multirow{2}{*}{$\incrementSA(\SA_{a}, t, k)$}   & increment each integer $\SA_{a}[i]$ in $\SA_{a}$ by $k$      &  \\
   &  only if $\SA_{a}[i]$ is larger than $t$      &  \\

\multirow{2}{*}{$\decrementSA(\SA_{a}, t, k)$}          & decrement each integer $\SA_{a}[i]$ in $\SA_{a}$ by $k$     &  \\
     & only if $\SA_{a}[i]$ is larger than $t$.     &  \\
\hline
\end{tabular}
\end{table}

\begin{table}[htb]
\caption{Summary of the two queries supported by the dynamic data structure for permutations introduced by Salson et al.~\cite{DBLP:journals/jda/SalsonLLM10}, which is built on $\Pi_a$ of $r$ integers. 
Here, $i$ is a position in $\Pi_{a}$; $\pi \in \{ 1, 2, \ldots, r \}$ is an integer.}
\label{tab:queries_on_permutation}

\begin{tabular}{|c|l|c|}
\hline
Query & Description & Time complexity \\
\hline
$\paccess(\Pi_{a}, i)$        & return $\Pi_{a}[i]$       & \multirow{2}{*}{$\mathcal{O}(\log r)$} \\
$\pinvAccess(\Pi_a, \pi)$  & return the position of $\pi$ in $\Pi_{a}$          &  \\
\hline
\end{tabular}
\end{table}

\begin{table}[htb]
\caption{Summary of the two update operations supported by the dynamic data structure for permutations introduced by Salson et al.~\cite{DBLP:journals/jda/SalsonLLM10}, which is built on $\Pi_a$ of $r$ integers. 
Here, $i$ is a position in $\Pi_{a}$; $\pi \in \{ 1, 2, \ldots, r \}$ is an integer.}
\label{tab:operations_on_permutation}
\begin{tabular}{|c|l|c|}
\hline
Operation & Description & Time complexity \\
\hline
\multirow{2}{*}{$\incrementInsert(\Pi_{a}, i, \pi)$}  & increment each integer $\Pi_{a}[j]$ in $\Pi_{a}$ by 1 only if $\Pi_{a}[j] \geq \pi$;  & \multirow{4}{*}{$\mathcal{O}(\log r)$} \\
  & insert $\pi$ into $\Pi$ at positoin $i$  &  \\
\multirow{2}{*}{$\decrementDelete(\Pi_a, i)$}  & decrement each integer $\Pi_{a}[j]$ in $\Pi_{a}$ by 1 only if $\Pi_{a}[j] \geq \Pi_{a}[i]$;      &  \\
   &  remove the element at position $i$ from $\Pi_{a}$     &  \\
\hline
\end{tabular}
\end{table}

The dynamic data structures $\mathscr{D}_{\DI}(\SA_s)$ and $\mathscr{D}_{\DI}(\SA_e)$ support the following three queries,  
where $\SA_{a}$ denotes either $\SA_s$ or $\SA_e$:  
(i) access $\accessSA(\SA_{a}, i)$;  
(ii) order $\orderSA(\SA_{a}, i)$;  
(iii) count $\countSA(\SA_{a}, i)$.  
Descriptions of these queries are provided in Table~\ref{tab:queries_on_SA}.
In addition, $\mathscr{D}_{\DI}(\SA_{a})$ supports four update operations on $\SA_{a}$:  
(i) insertion $\insertSA(\SA_{a}, i, t)$;  
(ii) deletion $\deleteSA(\SA_{a}, i)$;  
(iii) increment $\incrementSA(\SA_{a}, t, k)$;  
(iv) decrement $\decrementSA(\SA_{a}, t, k)$.  
Details of these four operations are provided in Table~\ref{tab:operations_on_SA}.

$\mathscr{D}_{\DI}(\SA_{a})$ consists of two components, each of which is a dynamic data structure.
The first component is the dynamic data structure for permutations introduced by Salson et al.~\cite{DBLP:journals/jda/SalsonLLM10}, which supports \emph{access} and \emph{inverse access} queries as well as \emph{insertion} and \emph{deletion} operations on a permutation. Here, $\Pi_a$ denotes the permutation of the integers $\{1,2,\ldots,r\}$ that sorts $\SA_{a}$ in increasing order; that is,
$\SA_{a}[\Pi_a[1]] < \SA_{a}[\Pi_a[2]] < \cdots < \SA_{a}[\Pi_a[r]]$.
\emph{Access} and \emph{inverse access} queries on $\Pi_a$ are described in Table~\ref{tab:queries_on_permutation}.
Two update operations of increment-insertion and decrement-deletion on $\Pi$ are also described in Table~\ref{tab:operations_on_permutation}.

Salson et al.'s dynamic data structure for $\Pi_a$ requires $\mathcal{O}(r)$ words of space and supports access and inverse access queries, as well as two update operations, in $\mathcal{O}(\log r)$ time (Section 4 in \cite{DBLP:journals/jda/SalsonLLM10}).

The second component is the Bille et al.'s dynamic data structure (introduced in Appendix~\ref{app:RLE_dynamic_data_structure})  built on a sequence $S'$ of ($r+1$) integers $\SA_{a}[\Pi_{a}[1]]$, $(\SA_{a}[\Pi_{a}[2]] - \SA_{a}[\Pi_{a}[1]])$, $\ldots$, $(\SA_{a}[\Pi_{a}[r]] - \SA_{a}[\Pi_{a}[r-1]])$, $(n - \SA_{a}[\Pi_{a}[r]])$, with a user-defined parameter $\delta = 1$. 

The dynamic data structure $\mathscr{D}_{\DI}(\SA_{a})$ requires a total of $\mathcal{O}(r)$ words of space and supports access, order, and count queries on $\SA_{a}$ in $\mathcal{O}(\log r)$ time. 

The lemmas below provide time complexities for the queries and updates supported by $\mathscr{D}_{\DI}(\SA_s)$ and $\mathscr{D}_{\DI}(\SA_e)$.

\begin{lemma}\label{lem:DynDI_Lemma}
$\mathscr{D}_{\DI}(\SA_s)$ and $\mathscr{D}_{\DI}(\SA_e)$ support    
(i) \emph{access}, \emph{order}, and \emph{count} queries, and  
(ii) \emph{insertion}, \emph{deletion}, \emph{increment}, and \emph{decrement} operations in $\mathcal{O}(\log r)$ time.
\end{lemma}
\begin{proof}
The proof of Lemma~\ref{lem:DynDI_Lemma} is as follows. 

\textbf{Supporting $\accessSA(\SA_{s}, i)$.}
$\SA_{s}[i]$ is computed by $\sumPS(S', i')$, 
where $i'$ is the position of $i$ in $\Pi_{s}$. 
Here, $i'$ is computed by $\pinvAccess(\Pi_{s}, i)$.

\textbf{Supporting $\orderSA(\SA_{s}, i)$.}
The $i$-th smallest integer in $\SA_{s}$ is computed by $\sumPS(S', i)$. 

\textbf{Supporting $\countSA(\SA_{s}, i)$.}
Let $x \geq 0$ be the largest integer satisfying $\sumPS(S', x) < i$, 
where let $\sumPS(S', 0) = 0$ for simplicity. 
Then, $\countSA(\SA_{s}, i)$ returns $x$. 
This integer $x$ is computed as $\searchPS(S', i+1) - 1$.

\textbf{Supporting $\insertSA(\SA_{s}, i, t)$.}
Let $u$ be the number of integers in $\SA_{s}$ that are smaller than $t$ (i.e., $u = \countSA(\SA_{s}, t)$). 
Then, $\Pi_{s}$ and $S'$ are updated as follows: 
(i)  Divide the $(u+1)$-th integer in $S'$ into $(t - \SA_{s}[\Pi_{s}[u]])$ and $(\SA_{s}[\Pi_{s}[u+1]] - t)$ by $\dividePS(S', u+1, t - \SA_{s}[\Pi_{s}[u]])$; 
(ii) Perform $\incrementInsert(\Pi_{s}, u+1, i)$. 

\textbf{Supporting $\deleteSA(\SA_{s}, i)$.}
Let $u'$ be the number of integers in $\SA_{s}$ that are smaller than $\SA_{s}[i]$ (i.e., 
$u' = \countSA(\SA_{s}, \accessSA(\SA_{s}, i))$). 
Then, $\Pi_{s}$ and $S'$ are updated as follows: 
(i) Merge the $(u'+1)$-th and $(u'+2)$ integers in $S'$ by $\mergePS(S', u'+1)$; 
(ii) Perform $\decrementDelete(\Pi_s, u+1)$. 

\textbf{Supporting $\incrementSA(\SA_{s}, t, k)$.}
Let $v$ be the number of integers in $\SA_{s}$ that are smaller than $t+1$ (i.e., $v = \countSA(\SA_{s}, t+1)$). 
Then, the $(v+1)$-th integer in $S'$ is incremented by $k$. 
This update is performed as follows: 
(i) Insert $k$ into $S'$ as the $(v+2)$-th integer; 
(ii) Merge the $(v+1)$-th and $(v+2)$-th integers in $S'$ by $\mergePS(S', v+1)$

\textbf{Supporting $\decrementSA(\SA_{s}, t, k)$.}
In contrast, 
the $(v+1)$-th integer in $S'$ is decremented by $k$, 
where $S'[v+1] > k$.
This update is performed as follows: 
(i) Divide the $(v+1)$-th integer in $S'$ into $k$ and $(S'[v+1] - 1)$ by $\dividePS(S', v+1, k)$; 
(ii) Delete the $(v+1)$-th integer from $S'$. 

These three queries and four operations are performed using $\mathcal{O}(1)$ queries and operations on $\Pi_{s}$ and $S'$. 
Therefore, they take $\mathcal{O}(\log r)$ time in total. 
Similarly, the four operations on $\SA_{e}$ can be supported in the same time using the same approach. 
\end{proof}
\subsection{Character Insertion Algorithm}\label{subsec:char_insertion_algorithm}
\textbf{Problem setup.} We describe how to update all three data structures when character $c$ is inserted at position $i$ in string $T$, producing $T' = T[1..i-1]cT[i..n]$. Our algorithm simulates the BWT and suffix array update method of Salson et al. \cite{DBLP:journals/tcs/SalsonLLM09,DBLP:journals/jda/SalsonLLM10}, but operates directly on the compressed and sampled representations.

\textbf{Key insight.} The insertion of a single character affects the entire conceptual matrix structure, which is intrinsically linked to the RLBWT and the sampled SAs. Instead of reconstructing everything from scratch, we simulate the gradual transformation through a series of updates that maintain the compressed and sampled representations at each step.

\textbf{Correctness guarantee.} Throughout the update process, we maintain critical invariants: (1) $L_{\RLE}$ correctly represents the run-length encoding of the current BWT, (2) $\SA_s$ and $\SA_e$ contain suffix array values at run start and end positions respectively, and (3) all three structures remain synchronized. Correctness follows from the proven correctness of the underlying Salson et al. algorithm and our faithful preservation of these invariants at each step.

\textbf{Algorithmic framework.}
The algorithm simulates the transformation from conceptual matrix $\CM$ of $T$ to $\CM'$ of $T'$ through $(n+1)$ iterations, from $j = n+1$ down to $j = 1$. We maintain the following sequences.  
(A) Conceptual matrix: $\CM_{n+1} = \CM$, $\CM_n$, $\ldots$, $\CM_0 = \CM'$ (not stored explicitly), 
where $\CM'$ is the conceptual matrix of $T'$; 
(B) RLBWT: $L_{\RLE, n+1} = L_{\RLE}$, $L_{\RLE, n}$, $\ldots$, $L_{\RLE, 0} = L'_{\RLE}$, 
where $L'_{\RLE}$ is the RLBWT representing the BWT $L'$ of $T'$; 
(C) Sampled SAs: $\SA_{s, n+1} = \SA_s$, $\ldots$, $\SA_{s, 0} = \SA'_s$ and $\SA_{e, n+1} = \SA_e$, $\ldots$, $\SA_{e, 0} = \SA'_e$, 
where $\SA'_{s}$ and $\SA'_{e}$ are the sampled SAs built from the suffix array $\SA'$ of $T'$. 

Each iteration $j$ transforms the state from step $j$ to step $j-1$ using a three-step process that maintains consistency across all data structures. 

\emph{Purpose:} 
At iteration $j$, insert the $j$-th circular shift $T^{\prime[j]}$ of $T'$ into $\CM_{j}$ in lexicographic order, 
and remove a circular shift of $T$ corresponding to $T^{\prime[j]}$ from $\CM_{j}$, resulting in $\CM_{j-1}$. 
This insertion and deletion establish the correspondence between the original string $T$ and the modified string $T'$. 
Based on these insertions and deletions of circular shifts, RLBWT and sampled SAs are appropriately updated. 

\emph{Intuition:} When we insert character $c$ at position $i$, each circular shift $T^{[k]}$ of the original string $T$ maps to 
the $(k+1)$-th circular shift $T^{\prime[k+1]}$ of $T'$ if $k \geq i$; otherwise it maps to $T^{\prime[k]}$ of $T'$. 
If there is a circular shift $T^{[k]}$ of $T$ mapped to $T^{\prime[j]}$ (i.e., $j \neq i$), 
then $T^{[k]}$ is removed from $\CM_{j}$ at iteration $j$. 
This removal is necessary because the conceptual matrix must maintain exactly $n+1$ rows at each step: when we insert the new circular shift $T^{\prime[j]}$ of $T'$, we must simultaneously remove the corresponding circular shift $T^{[k]}$ of $T$ to preserve the matrix size and lexicographic ordering. 
The details of the three steps are as follows. 

\textbf{Step (I): Identification of circular shifts to be removed.}
In this step, we identify the circular shift of $T$ to be removed from $\CM_{j}$. 

\emph{Algorithm:} 
This step is performed only when $j \neq i$ (i.e., when some circular shift of $T$ is mapped to $T^{\prime[j]}$). 
The row index $x$ of the circular shift to be removed from $\CM_{j}$ is computed as follows. 

\begin{enumerate}
\item \textbf{Initial iteration} ($j = n + 1$): $x = 1$ (always remove the lexicographically smallest circular shift); 
\item \textbf{Boundary iteration} ($j = i - 1$): $x = \ISA[i-1] + \epsilon$, where $\epsilon = 1$ if $T[i-1] \geq c$ and otherwise $\epsilon = 0$. 
Here, $\ISA[i-1]$ and $T[i-1]$ are computed by two auxiliary operations $\compISA_{1}(L_{\RLE}, \SA_s, i)$ and $\compT(L_{\RLE}, \SA_s, i)$, respectively, which are described in detail later. 
\item \textbf{The other iterations}: $x = \lexCount(L_{\RLE, j+1}, L_{j+1}[x_{+1}])$ $+ \rank(L_{\RLE, j+1}, x_{+1}$, $L_{j+1}[x_{+1}])$, 
where $x_{+1}$ denotes the row index identified in the previous iteration.
\end{enumerate}
Here, $L_{\RLE, j}$ is temporarily restored to $L_{\RLE, j+1}$ in order to perform this step. This is because the formula for computing $x$ uses $L_{\RLE, j+1}$, whereas $L_{\RLE, j+1}$ was already updated to $L_{\RLE, j}$ in the previous iteration. 
This restoration can be carried out in $\mathcal{O}(1)$ operations on $L_{\RLE, j}$ based on the result of the previous iteration.

\textbf{Step (II): Deletion from data structures.} 
Remove the identified circular shift from $\CM_j$, and update all three data structures accordingly to maintain consistency.

\emph{Intuition:} 
Removing the identified circular shift from $\CM_{j}$ indicates that the $x$-th character is removed from the current BWT $L_{j}$, 
which is the BWT corresponding to $\CM_{j}$ (i.e., $L_{j}$ consists of the last characters in the rows of $\CM_{j}$). 
This removal of a character from the BWT affects runs in its RLBWT (merge or shorten runs), 
and the changed runs affect values in the sampled SAs (delete or change values) 
since these arrays are intrinsically linked to runs. 

\emph{Algorithm:} 
This step is performed only when $j \neq i$ (i.e., when a circular shift was identified in Step I). 

\textbf{RLBWT update:}
Let $(c_v, \ell_v)$ be the run containing the $x$-th character of the BWT 
in $L_{\RLE, j} = (c_1, \ell_1), (c_2, \ell_2), \ldots, (c_r, \ell_r)$. 
Then, $L_{\RLE, j}$ is updated by the following three operations:

\begin{enumerate}
\item \textbf{Run-length decrement}: Decrement the length of the $v$-th run (i.e., $(c_{v}, \ell_{v}) \rightarrow (c_{v}, \ell_{v}-1)$) if its length is at least two (i.e., $\ell_{v} \geq 2$).
\item \textbf{Run removal}: Remove the $v$-th run from RLBWT if it is a single character (i.e., $\ell_{v} = 1$).
\item \textbf{Run merging}: Merge the $(v-1)$-th and $(v+1)$-th runs into a new run $(c_{v-1}, \ell_{v-1} + \ell_{v+1})$ if the two runs are the same characters (i.e., $c_{v-1} = c_{v+1}$), 
and the $v$-th run is removed from RLBWT.
\end{enumerate}

\textbf{Sampled SA updates:}
$\SA_{s, j}$ is updated consistently with the RLBWT changes. 
This update is executed by the following two operations. 

\begin{enumerate}
    \item \textbf{Run-start removal}: Remove the $v$-th sampled position $\SA_{s}[v]$ from $\SA_{s}$ 
    if the $v$-th run is removed. 
    In addition, the $(v+1)$-th sampled position is also removed 
    if the $(v-1)$-th and $(v+1)$-th runs are merged.
    \item \textbf{Boundary update}: 
    Change the $v$-th sampled position to $\SA_{j}[x+1]$ if  
    the first character of the $v$-th run is removed (i.e., $x = t_{v}$ and $\ell_{v} \geq 2$ for the starting position $t_{v}$ of the $v$-th run in BWT). 
    Here, $\SA_{j}$ is the suffix array corresponding to $\CM_{j}$ 
    (i.e., For each $s \in \{ 1, 2, \ldots, |L_{j}| \}$, $\SA_{j}[s] := k$ if the $s$-th row of $\CM_{j}$ is a circular shift $T^{[k]}$ of $T$; otherwise, the $s$-th row is a circular shift $T^{\prime[k']}$ of $T'$, and $\SA_{j}[s] := k'$), 
    and $\SA_j[x+1]$ is computed by auxiliary operation $\compSA_{X, 1}(L_{\RLE, j+1}, \SA_{s, j+1}, \SA_{e, j+1}, x)$, which is described in detail later. 
\end{enumerate}

Similarly, $\SA_{e, j}$ is updated by the following two operations. 
\begin{enumerate}
    \item \textbf{Run-end removal}: Remove the $v$-th sampled position $\SA_{e}[v]$ from $\SA_{e}$ 
    if the $v$-th run is removed. 
    In addition, the $(v-1)$-th sampled position is also removed 
    if the $(v-1)$-th and $(v+1)$-th runs are merged.
    \item \textbf{Boundary update}: 
    Change the $v$-th sampled position to $\SA_{j}[x-1]$ if  
    the last character of the $v$-th run is removed (i.e., $x = t_{v} + \ell_{v} - 1$ and $\ell_{v} \geq 2$. 
    Here, $\SA_j[x-1]$ is computed by auxiliary operation $\compSA_{X, 2}(L_{\RLE, j+1}, \SA_{s, j+1}, \SA_{e, j+1}, x)$. 
\end{enumerate}

Computing required suffix array values like $\SA_{j}[x-1]$ and $\SA_{j}[x+1]$ is non-trivial since we only store sampled positions, not the full suffix array.
We employ a dynamic LF function that maps positions between consecutive suffix arrays during the transformation process. This function enables us to compute any required $\SA_j$ value because: (i) each $\SA_j$ can be derived from the original $\SA$ through a sequence of LF transformations, and (ii) the original $\SA$ can be reconstructed from our sampled structures $\SA_s$ and $\SA_e$ using standard techniques~\cite{10.1145/3375890}.
The computation of values like $\SA_{j}[x-1]$ and $\SA_{j}[x+1]$ reduces to $\mathcal{O}(\log n)$ time using queries supported by our three dynamic structures. See Section~\ref{subsubsec:dynamic_LF} for the mathematical details of the dynamic LF function.

\textbf{Step (III): Insertion into data structures.}
Insert the new circular shift $T'^{[j]}$ into $\CM_j$ in lexicographic order and update all data structures accordingly. 

\emph{Intuition:} 
The insertion of $T'^{[j]}$ into $\CM_{j}$ at an appropriate row index $y$ indicates that 
its last character $T'^{[j]}[n+1]$ is inserted into the BWT at position $y$. 
Similar to Step (II), this insertion into BWT affects RLBWT and sampled SAs. 

\emph{Algorithm:} 
Before updating data structures, 
we compute the last character $T'^{[j]}[n+1]$ of $T'^{[j]}$ and the row index $y$. 

\textbf{Computing $T'^{[j]}[n+1]$:}
The character to be inserted depends on the iteration type:
\begin{enumerate}
    \item \textbf{New character insertion} ($j = i+1$): $T^{\prime[j]}[n+1] = c$ (the inserted character)
    \item \textbf{Boundary character} ($j = i$): $T^{\prime[j]}[n+1] = T[i-1]$ (character at insertion boundary)  
    \item \textbf{Existing characters} ($j \in \{ n+1, n, \ldots, i+2, i-1, i-2, \ldots, 1  \}$): $T^{\prime[j]}[n+1] = c_{v}$ (character from the run identified in Step~(II))
\end{enumerate}

Let $y_{+1}$ be the row index for insertion computed in the previous iteration. 
Then, the row index $y$ is computed as follows. 
\begin{enumerate}
\item \textbf{Early iterations} ($j \in \{n+1, \ldots, i+1\}$): $y = x$ (insert at the same position where we removed); 
\item \textbf{Special iteration} ($j = i$): $y = \lexCount(L_{\RLE, j}, L_j[y_{+1}]) + \rank(L_{\RLE, j}, y_{+1}, L_j[y_{+1}]) + \epsilon$. 
Here, $\epsilon = 1$ if (A) $T[i-1] < L_j[y_{+1}]$ or (B) $T[i-1] = L_j[y_{+1}]$ and $\ISA[i] \leq y_{+1}$. Otherwise $\epsilon = 0$. Similar to $\ISA[i-1]$, 
$\ISA[i]$ is precomputed by an auxiliary operation $\compISA_{2}(L_{\RLE}, \SA_s, i)$, which is described in detail later. 
\item \textbf{Later iterations} ($j \in \{ i-1, i-2, \ldots, 1 \}$): 
$y = \lexCount(L_{\RLE, j}, L_j[y_{+1}]) + \rank(L_{\RLE, j}$, $y_{+1}, L_j[y_{+1}])$. 
\end{enumerate}
Here, similar to Step~(I),
$L_{\RLE, j}$ is temporarily restored to the state of the RLBWT before Step~(II) in order to use the formula for computing $y$,
since this formula requires $L_{\RLE, j}$ and $L_j$ to denote the RLBWT and BWT before they are updated in Step~(II). 

\begin{figure}[t]
\begin{center}
	\includegraphics[width=1.0\textwidth]{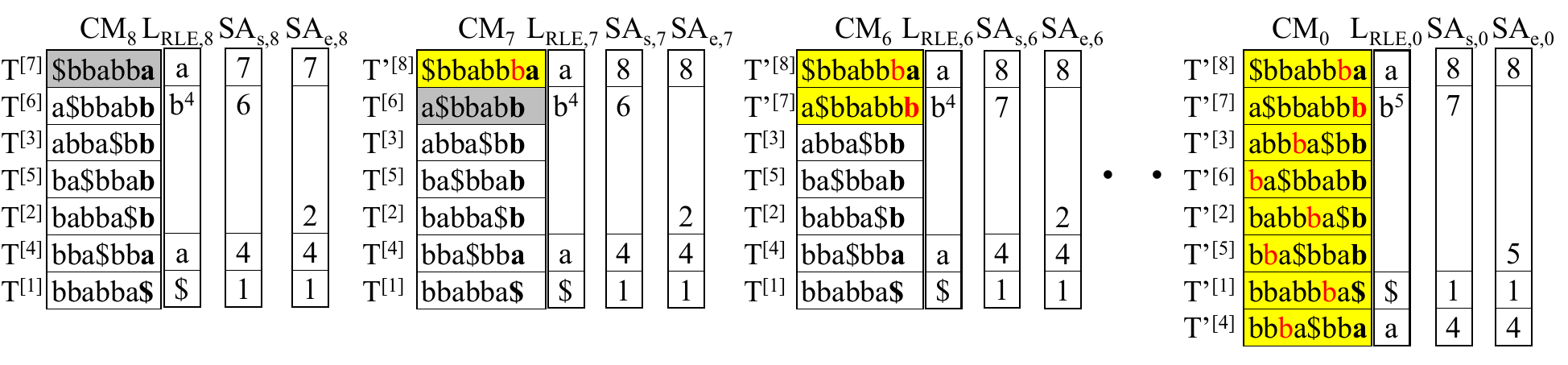}
\caption{
Illustration of the update process from $\CM$, $L_{\RLE}$, $\SA_{s}$, and $\SA_{e}$ 
to $\CM'$, $L'_{\RLE}$, $\SA'_{s}$, and $\SA'_{e}$, respectively. 
Here, $T = \texttt{bbabba\$}$ and $T' = \texttt{bbabbba\$}$, 
where $T'$ is obtained by inserting the character $b$ into $T$ at position 6. In each $\CM_j$, the circular shift to be removed is highlighted in gray, while the newly inserted circular shift is highlighted in yellow. The BWT corresponding to $\CM_{j}$ is shown with bold characters. See Appendix~\ref{sec:appendix} for the full version of this figure.}
\label{fig:update_CM}
\end{center}
\end{figure}

\textbf{RLBWT updates:}
Let $(c_{w}, \ell_{w})$ be the run containing position $y$ in $L_{\RLE, j}$ (or the last run if $y > |L_{j}|$). 
$L_{\RLE, j}$ is updated by the following three operations: 
\begin{enumerate}
    \item \textbf{Run-length increment}: 
    Increment the length of the $w$-th run if $T'^{[j]}[n+1]$ is inserted as a character of the run 
    (i.e., $T'^{[j]}[n+1] = c_{w}$). 
    Similarly, 
    if $T'^{[j]}[n+1]$ is inserted as the last character of the $(w-1)$-th run (i.e., $t_{w} = y$ and $T'^{[j]}[n+1] = c_{w-1}$), 
    then its length is incremented. 
    \item \textbf{Run split}: Split the $w$-th run into two runs $(c_w, y - t_w)$ and $(c_w, t_w + \ell_w - y)$ if $T'^{[j]}[n+1]$ is inserted into the $w$-th run (i.e., $t_{w} < y < t_{w} + \ell_{w} - 1$) and $T'^{[j]}[n+1] \neq c_{w}$. 
    \item \textbf{Run insertion}: 
    Insert a new run $(T'^{[j]}[n+1], 1)$ into $L_{\RLE, j}$ at an appropriate position if 
    $T'^{[j]}[n+1]$ is inserted after a run $(c_{u}, \ell_{u})$ and 
    $T'^{[j]}[n+1] \neq c_{u}$. 
    Here, $(c_{u}, \ell_{u})$ is (A) the $(w-1)$-th run, 
    (B) the $w$-th run, or (C) the new run $(c_w, y - t_w)$ created by the split operation. 
\end{enumerate}

\textbf{Sampled SA updates:}
$\SA_{s, j}$ is updated consistently with the RLBWT changes. 
These updates are performed by the following operations: 

\begin{enumerate}
    \item \textbf{Run-start insertion}: Insert $j$ into $\SA_{s, j}$ as a new sampled position if $T'^{[j]}[n+1]$ is inserted as a new run. 
    \item \textbf{Boundary update}: The sampled position of the $w$-th run is changed to $j$ (i.e., $\SA_{s, j}[w] \rightarrow j$) if $T'^{[j]}[n+1]$ is inserted as the first character of the $w$-th run.
    \item \textbf{Split handling}: 
    If the $w$-th run is split into two new runs, 
    then insert $\SA_{j-1}[y+1]$ as the sampled position of the second new run. 
    Here, $\SA_{j-1}[y+1]$ is computed by auxiliary operation $\compSA_{Y, 1}(L_{\RLE, j}, \SA_{s, j}, \SA_{e, j}, y)$, which is described in detail later. 
\end{enumerate}

Similarly, $\SA_{e, j}$ is updated by the following operations: 
\begin{enumerate}
    \item \textbf{Run-end insertion}: Insert $j$ into $\SA_{e, j}$ as a new sampled position 
    if $T'^{[j]}[n+1]$ is inserted as a new run. 
    \item \textbf{Boundary update}: The sampled position of the $w$-th run is changed to $j$ (i.e., $\SA_{e, j}[w] \rightarrow j$) if $T'^{[j]}[n+1]$ is inserted as the last character of the $w$-th run.
    \item \textbf{Split handling}: 
    If the $w$-th run is split into two new runs, 
    then insert $\SA_{j-1}[y-1]$ as the sampled position of the first new run. 
    Here, $\SA_{j-1}[y-1]$ is computed by auxiliary operation $\compSA_{Y, 2}(L_{\RLE, j}, \SA_{s, j}, \SA_{e, j}, y)$, which is described in detail later. 
\end{enumerate}

\textbf{Iteration completion.}
Through these three steps, each iteration transforms:
$(\CM_{j}, L_{\RLE, j}$, $\SA_{s, j}, \SA_{e, j}) \rightarrow (\CM_{j-1}, L_{\RLE, j-1}, \SA_{s, j-1}, \SA_{e, j-1})$. 
The algorithm uses only the queries and update operations supported by $\mathscr{D}_{\RLE}(L_{\RLE})$, $\mathscr{D}_{\DI}(\SA_s)$, and $\mathscr{D}_{\DI}(\SA_e)$. 

Figure~\ref{fig:update_CM} illustrates an example transformation.  
In this example, 
$L_{7} = L_{6} = \texttt{abbbba\$}$, 
$\SA_{7} = 8, 6, 3, 5, 2, 4, 1$, and $\SA_{6} = 8, 7, 3, 5, 2, 4, 1$. 
We obtain $x = 2$, $y = 2$, $x_{+1} = 1$, and $y_{+1} = 1$ for iteration $j = 7$. 
$x$ and $y$ are computed in Step~(I) and Step~(III), respectively, as follows: 
$x = \lexCount(L_{\RLE, 8}, L_{8}[x_{+1}]) + \rank(L_{\RLE, 8}, x_{+1}, L_{8}[x_{+1}])$ 
and 
$y = x$. 
Here, $\lexCount(L_{\RLE, 8}, L_{8}[x_{+1}]) = 1$ and $\rank(L_{\RLE, 8}, x_{+1}, L_{8}[x_{+1}]) = 1$. 
$L_{\RLE, 7}$ is unchanged, i.e., $L_{\RLE, 7} = L_{\RLE, 6}$. 
This is because $L_{7}[x] = b$, and this character is changed to $b$, resulting in $L_{6}$. 
Similarly, $\SA_{7}[x] = 6$ is changed to $7$, resulting in $\SA_{6}$. 
The $x$-th character of $L_{7}$ is the first character of the second run $(b, 4)$ in $L_{\RLE, 7}$, 
and hence, $\SA_{s, 7}[2] = 6$ is changed to $7$, resulting in $\SA_{s, 6}$. 
On the other hand, this character is not the last character of the second run. 
Therefore, $\SA_{e, 7}$ is unchanged, i.e., $\SA_{e, 7} = \SA_{e, 6}$. 

\textbf{Critical optimization: Iteration skipping.}
\emph{Key insight:} Most of the $(n+1)$ iterations do not actually change the data structures, so they can be skipped or replaced with $\mathcal{O}(1)$ operations supported by $\mathscr{D}_{\DI}(\SA_s)$ and $\mathscr{D}_{\DI}(\SA_e)$. 
The correctness of this iteration skipping is guaranteed by the following lemma. 
\begin{lemma}\label{lem:state_stability}
    Let $K \geq 1$ be 
    the smallest integer such that $x = y$ in some iteration $j \leq i - K$. If no such $K$ exists, define $K := i$. 
    Here, $x$ and $y$ are the row indexes of $\CM_{j}$ used in Step (II) and Step (III), respectively.     
    The following three statements hold: 
    \begin{itemize}
        \item[(i)] $L_{n+1} = L_{n} = \cdots = L_{i+1}$;
        \item[(ii)] For each $j \in \{ n+1, n, \ldots, i \}$, 
        $\SA_{j}[t] = \SA[t] + 1$ if $\SA[t] \geq j$; otherwise, $\SA_{j}[t] = \SA[t]$, 
        where $t$ is a position of $\SA_{j}$; 
        \item[(iii)] $L_{i-K-1} = L_{i-K-2} = \cdots = L_{0}$ and $\SA_{i-K-1} = \SA_{i-K-2} = \cdots = \SA_{0}$. 
    \end{itemize}
\end{lemma}
\begin{proof}
    Lemma~\ref{lem:state_stability}(i) and Lemma~\ref{lem:state_stability}(ii) follow from 
    the following two observations: 
    (A) $x = y$ always hold for all $j \in \{ n+1, n, \ldots, i+1 \}$, 
    i.e., $T^{[j-1]}$ in $\CM_{j}$ is replaced with $T^{\prime[j]}$ by the update procedure, 
    and (B) $T^{[j-1]}$ and $T^{\prime[j]}$ have the same last character (i.e., $T^{[j-1]}[n] = T^{\prime[j]}[n+1]$) 
    for all $j \in \{ n+1, n, \ldots, i+2 \}$. 

    Similarly, Lemma~\ref{lem:state_stability}(iii) follows from 
    the following two observations: 
    (a) $x = y$ always hold for all $j \in \{ i-K-1, i-K-2, \ldots, 1 \}$, 
    i.e., $T^{[j]}$ in $\CM_{j}$ is replaced with $T^{\prime[j]}$ by the update procedure 
    (See Section 3.2.2 in \cite{DBLP:journals/tcs/SalsonLLM09} for more details), 
    and (b) $T^{[j]}$ and $T^{\prime[j]}$ have the same last character. 
\end{proof}

\emph{Early iteration skipping:} 
The first $(n-i)$ iterations ($j \in \{n+1, n, \ldots, i+2\}$) can be replaced with $\mathcal{O}(1)$ operations. 
This is because (i) the RLBWT is unchanged, i.e., $L_{\RLE} = L_{\RLE, i+1}$; 
(ii) $\SA_{s} \rightarrow \SA_{s, i+1}$ via $\incrementSA(\SA_{s}, i, 1)$ (increment positions $\geq i+1$);
(iii) $\SA_{e} \rightarrow \SA_{e, i+1}$ via $\incrementSA(\SA_{e}, i, 1)$ (increment positions $\geq i+1$). 

\emph{Late iteration skipping:} The last $(i-K-1)$ iterations ($j \in \{i-K-1, i-K-2, \ldots, 1 \}$) can be skipped.  
This is because, once $x = y$ occurs, this equation always holds until the algorithm reaches a stable state where $L_{\RLE, i-K-1} = L'_{\RLE}$, $\SA_{s, i-K-1} = \SA'_{s}$, and $\SA_{e, i-K-1} = \SA'_{e}$, meaning no further changes are needed. 

Finally, the total iterations reduce from $(n+1)$ to $(2 + K)$, specifically $j \in \{ i+1, i, \ldots, i-K \}$. 
If $K$ is sufficiently small, this leads to dramatic performance improvements. 

For the example in Figure~\ref{fig:update_CM},  
we obtain $n = 7$, $i = 6$, and $K = 3$. 
Here, the full version of Figure~\ref{fig:update_CM} is found in Appendix~\ref{sec:appendix}. 
The first iteration is replaced with two operations $\incrementSA(\SA_{s}, 6, 1)$ 
and $\incrementSA(\SA_{s}, 7, 1)$. 
The last two iterations (i.e., $j \in \{ 2, 1 \}$) are skipped. 
Therefore, the total iterations reduce from $8$ to $5$, specifically $j \in \{ 7, 6, \ldots, 3 \}$.

\subsubsection{Details of Dynamic LF Function}\label{subsubsec:dynamic_LF}
\begin{figure}[t]
\begin{center}
	\includegraphics[width=0.5\textwidth]{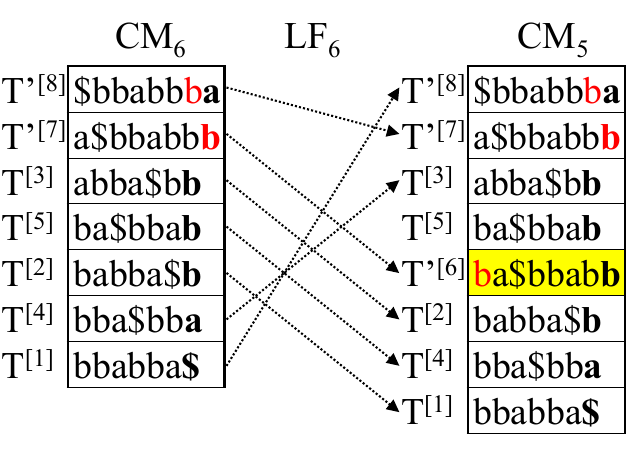}
\caption{
Illustration of the dynamic LF function for the example shown in Figure~\ref{fig:update_CM}. 
Here, $L_{6} = \texttt{abbbba\$}$, $L_{5} = \texttt{abbbbba\$}$, 
$\SA_{6} = [8, 7, 3, 5, 2, 4, 1]$, and $\SA_{5} = [8, 7, 3, 5, 6, 2, 4, 1]$. 
The arrow starting at each $t$-th row of $\CM_{5}$ represents $\LF_{5}(t) = t'$, 
where $t'$ is the index of the row pointed to by the arrow.}
\label{fig:dynamic_LF}
\end{center}
\end{figure}

\emph{LF} is a function that maps $t \in \{1, \ldots, n\}$ to $\ISA[\SA[t] - 1]$.  
That is, $\SA[\LF(t)] = \SA[t] - 1$ if $\SA[t] \ne 1$; otherwise, $\SA[\LF(t)] = n$.
The \emph{LF} function satisfies the LF formula~\cite{DBLP:conf/focs/FerraginaM00}:  
$\LF(t) = \lexCount(L_{\RLE}, L[t]) + \rank(L_{\RLE}, t, L[t])$.

We define a modified version of the \emph{LF} function at each iteration $j$, denoted $\LF_j$.  
The function $\LF_j$ maps a position $t$ in $\SA_j$ to a position $t'$ such that  
$\SA_{j-1}[t'] = \SA_j[t] - 1$ if $\SA_j[t] \ne 1$, and $\SA_{j-1}[t'] = n + 1$ otherwise. 
Figure~\ref{fig:dynamic_LF} depicts $\LF_{5}$ for the example shown in Figure~\ref{fig:update_CM}. 

The following lemma guarantees that the modified version $\LF_j$ satisfies the \emph{LF formula}.
\begin{lemma}[Dynamic LF formula]\label{lem:dynamic_LF_formula}
Let $j \in \{n+1, n, \ldots, 1\}$ and $t$ be a position in $\SA_j$.  
Define $\kappa(t) = 1$ if either  
(i) $T[i-1] < L_j[t]$, or  
(ii) $T[i-1] = L_j[t]$ and $\ISA[i] \leq t$;  
otherwise, $\kappa(t) = 0$.  
Then, the following equation holds: 
\begin{equation}
  \LF_{j}(t) =   \begin{cases}
    \lexCount(L_{\RLE, j}, L_{j}[t]) + \rank(L_{\RLE, j}, t, L_{j}[t]) & \text{if $j \neq i$} \\
    \lexCount(L_{\RLE, j}, L_{j}[t]) + \rank(L_{\RLE, j}, t, L_{j}[t]) + \kappa(t) & 
    \text{otherwise.}
  \end{cases}
\end{equation}
\end{lemma}
We give a proof sketch of Lemma~\ref{lem:dynamic_LF_formula}. 
We prove Lemma~\ref{lem:dynamic_LF_formula} by modifying the LF formula. 
The LF formula guarantees that LF function can be computed using 
rank and lex-count on the RLBWT representing the BWT $L$. 
For each iteration $j \in \{n+1, n, \ldots, 1\}$, 
$L_{j}$ is updated to $L_{j-1}$ by deleting a character from $L_{j}$ and inserting a new one. 
Based on these edit operations, 
we appropriately modify the LF formula, 
which results in the dynamic LF formula (i.e., Lemma~\ref{lem:dynamic_LF_formula}).

The detailed proof of Lemma~\ref{lem:dynamic_LF_formula} is as follows. 
We assume that $\SA_{j}[t] \neq 1$ for simplicity. 
Let $t' = \LF_{j}(t)$ and 
$W(t) = \lexCount(L_{\RLE, j}, L_{j}[t]) + \rank(L_{\RLE, j}, t, L_{j}[t])$. 
We prove Lemma~\ref{lem:dynamic_LF_formula} by induction on $j$ 
and consider the following four cases: 
\begin{enumerate}
\item \textbf{Early iterations}: $j \in \{n+1, \ldots, i+1\}$; 
\item \textbf{Boundary iteration}: $j = i$;
\item \textbf{Special iteration}: $j = i-1$;
\item \textbf{Later iterations}: $j \leq i-2$.
\end{enumerate}

\paragraph{Proof of Lemma~\ref{lem:dynamic_LF_formula} for Early Iterations ($j \in \{n+1, \ldots, i+1\}$)}
We prove $\LF_{j}(t) = W(t)$ (i.e., Lemma~\ref{lem:dynamic_LF_formula}). 
$L_{j} = L$, $\SA_{j}[t] = \SA[t] + \epsilon_{1}$, and $\SA_{j-1}[t'] = \SA[t'] + \epsilon_{2}$ follows from Lemma~\ref{lem:state_stability}. 
Here, 
$\epsilon_{1} = 1$ if $\SA[t] \geq j+1$; 
otherwise, $\epsilon_{1} = 0$. 
Similarly, 
$\epsilon_{2} = 1$ if $\SA[t'] \geq j$; 
otherwise, $\epsilon_{2} = 0$. 


$\SA[t'] = \SA[t] - 1$ always holds. 
This is because 
$\SA_{j-1}[t'] = \SA[t']$, $\SA_{j}[t] = \SA[t]$, and $\SA[t'] = \SA[t] - 1$ 
if $\SA[t'] < j$.  
Otherwise (i.e., $\SA[t'] \geq j$), 
$\SA_{j-1}[t'] = \SA[t'] + 1$, $\SA_{j}[t] = \SA[t] + 1$, and $\SA[t'] = \SA[t] - 1$. 

$\LF_{j}(t) = \LF(t)$ holds 
because $\LF(t) = t'$ follows from $\SA[t'] = \SA[t] - 1$. 
$\LF(t) = W(t)$ follows from LF formula. 
Therefore, we obtain $\LF_{j}(t) = W(t)$. 

\paragraph{Proof of Lemma~\ref{lem:dynamic_LF_formula} for Boundary Iteration ($j = i$)}
Let $x_{+1}$ and $y_{+1}$ be the row indexes in $\CM_{j+1}$ identified in Step~(I) and Step~(III) of the update procedure, respectively. 
Similarly, 
let $x$ and $y$ be the row indexes in $\CM_{j}$ identified in Step~(I) and Step~(III) of the update procedure, respectively. 
Here, $x$ does not exist in this iteration $j$. 
$x_{+1} = y_{+1}$ and $y = W(y_{+1}) + \epsilon$ follow from 
the formula for computing $y_{+1}$ used in Step~(III). 
Here, $\epsilon = 1$ if $T^{\prime [j]}[n+1] \leq L_{j}[y_{+1}]$, and $\epsilon = 0$ otherwise. 

We prove $\LF_{j}(t) = W(t) + \kappa(t)$ (i.e., Lemma~\ref{lem:dynamic_LF_formula}). 
We consider two cases: (A) $t = y_{+1}$ and (B) $t \neq y_{+1}$. 

\textbf{Case (A): $t = y_{+1}$.}
We prove $\epsilon = \kappa(y_{+1})$.
$y_{+1} = \ISA[i]$ 
because (i) $y_{+1} = x_{+1}$, 
and $x_{+1}$ is the row index of the circular shift $T^{[i]}$ in $\CM_{j+1}$; 
(ii) the $x_{+1}$-th row in the conceptual matrix of $T$ is not changed 
during the update from $\CM_{n+1}$ to $\CM_{i+1}$.
$T^{\prime [j]}[n+1] = T[i-1]$ follows from $j = i$. 
If $T^{\prime [j]}[n+1] \leq L_{j}[y_{+1}]$ (i.e., $T[i-1] \leq L_{j}[y_{+1}]$), 
then $\epsilon = 1$. 
$\kappa(y_{+1}) = 1$ follows from $T[i-1] \leq L_{j}[y_{+1}]$ and $\ISA[i] = y_{+1}$. 
Otherwise (i.e., $T^{\prime [j]}[n+1] > L_{j}[y_{+1}]$), 
$\epsilon = 0$ and $\kappa(y_{+1}) = 0$. 
Therefore, we obtain $\epsilon = \kappa(y_{+1})$. 

We prove $\LF_{j}(t) = W(t) + \kappa(t)$. 
In this case, $\LF_{j}(t) = y$ holds. 
Therefore, 
$\LF_{j}(t) = W(t) + \kappa(t)$ 
follows from $t = y_{+1}$, $y = W(y_{+1}) + \epsilon$, and $\epsilon = \kappa(y_{+1})$. 

\textbf{Case (B): $t \neq y_{+1}$.}
In this case, $\SA_{j}$ contains $(\SA_{j}[t] - 1)$ at a position $u$. 
The following lemma states the relationship between $u$ and $L_{\RLE, j}$. 
\begin{lemma}\label{lem:dynamic_LF_for_boundary_iteration_case_B}
The following equation holds.
\begin{equation*}
 u = (\lexCount(L_{\RLE, j}, L_{j}[t]) - \epsilon_{A}) + (\rank(L_{\RLE, j}, t, L_{j}[t]) - \epsilon_{B}).
\end{equation*}
Here, $\epsilon_{A}$ and $\epsilon_{B}$ are defined as follows:
\begin{align*}
    \epsilon_{A} &= 
  \begin{cases}
    1 & \text{if $L[\ISA[i]] \geq L_{j}[t]$ and $c < L_{j}[t]$} \\
    -1 & \text{if $L[\ISA[i]] < L_{j}[t]$ and $c \geq L_{j}[t]$} \\
    0 & \text{otherwise.} 
  \end{cases} \\
    \epsilon_{B} &= 
  \begin{cases}
    1 & \text{if $\ISA[i] \leq t$, $L[\ISA[i]] \neq L_{j}[t]$, and $c = L_{j}[t]$} \\
    -1 & \text{if $\ISA[i] \leq t$, $L[\ISA[i]] = L_{j}[t]$, and $c \neq L_{j}[t]$} \\
    0 & \text{otherwise.}
  \end{cases}
\end{align*}
\end{lemma}
\begin{proof}
We can prove $\LF(t) = u$ using the same approach used in the proof for early iterations. 
Using LF formula, 
we obtain $u = \lexCount(L_{\RLE}, L[t]) + \rank(L_{\RLE}, t, L[t])$. 
Lemma~\ref{lem:state_stability} indicates that 
$L_{j}$ can be obtained by changing the $x_{+1}$-th character to $c$ in $L$. 
Since $x_{+1} = \ISA[i]$, 
the following relationship holds between two lex-count queries $\lexCount(L_{\RLE, j}, L_{j}[t])$ and $\lexCount(L_{\RLE}, L_{j}[t])$:
\begin{equation*}
 \lexCount(L_{\RLE, j}, L_{j}[t]) = \lexCount(L_{\RLE}, L_{j}[t]) + \epsilon_{A}.
\end{equation*}


Similarly, 
the following relationship holds between two rank queries $\rank(L_{\RLE, j}, t, L_{j}[t])$ and $\rank(L_{\RLE}, t, L_{j}[t])$:
\begin{equation*}
 \rank(L_{\RLE, j}, t, L_{j}[t]) = \rank(L_{\RLE}, t, L_{j}[t]) + \epsilon_{B}.
\end{equation*}

We modify the LF formula using these relationships. 
Then, we obtain Lemma~\ref{lem:dynamic_LF_for_boundary_iteration_case_B}. 
\end{proof}

We prove $\LF_{j}(t) = W(t) + \kappa(t)$ using Lemma~\ref{lem:dynamic_LF_for_boundary_iteration_case_B}. 
If $u < y$, 
then $\SA_{j-1}$ contains the value $(\SA_{j}[t] - 1)$ at the same position $u$. 
In this case, $\LF_{j}(t) = u$, 
$\epsilon_{A} + \epsilon_{B} = 0$, and $\kappa(t) = 0$ hold. 
Otherwise (i.e., $u \geq y$), 
$\SA_{j-1}$ contains the value $(\SA_{j}[j] - 1)$ at position $(u+1)$. 
In this case, $\LF_{j}(t) = u+1$, 
$\epsilon_{A} + \epsilon_{B} = 1$, and $\kappa(t) = 1$ hold. 
Therefore, $\LF_{j}(t) = W(t) + \kappa(t)$ follows from Lemma~\ref{lem:dynamic_LF_for_boundary_iteration_case_B}.

\paragraph{Proof of Lemma~\ref{lem:dynamic_LF_formula} for Special Iteration ($j = i-1$)}
To prove $\LF_{j}(t) = W(t)$, 
we introduce a string $L_{\tmp}$ defined as $L_{\tmp} = L_{j+1}[1..\ISA[i]-1] T[i-1] L_{j+1}[\ISA[i]..n]$. 
Similarly, let $\SA_{\tmp} = \SA_{j+1}[1..\ISA[i]-1] (i) \SA_{j+1}[\ISA[i]..n]$. 
Let $W_{\tmp}(t) = \lexCount(L_{\tmp}, L_{\tmp}[t]) + \rank(L_{\tmp}, t, L_{\tmp}[t])$. 
The following lemma states properties of $W_{\tmp}(t)$. 

\begin{lemma}\label{lem:optional_lemma}
    Let $g$ be a position in $\SA_{j+1}$. 
    Let $\epsilon = 1$ if $g \geq \ISA[i]$; 
    otherwise $\epsilon = 0$. 
    Then, $\LF_{j+1}(g) = W_{\tmp}(g+\epsilon)$ holds.
\end{lemma}
\begin{proof}
    $\LF_{j+1}(g) = W_{+1}(g) + \kappa(g)$ follows from the inductive assumption. 
    We prove $W_{+1}(g) + \kappa(g) = W_{\tmp}(g + \epsilon)$.
    There exist four cases: 
    (i) $T[i-1] < L_{j+1}[g]$; 
    (ii) $T[i-1] = L_{j+1}[g]$ and $\ISA[i] \leq g$;
    (iii) $T[i-1] = L_{j+1}[g]$ and $\ISA[i] > g$;
    (iv) $T[i-1] > L_{j+1}[g]$.     

For case (i), 
$\kappa(g) = 1$, $\lexCount(L_{\tmp}, L_{\tmp}[g+\epsilon]) = \lexCount(L_{\RLE, j+1}, L_{j+1}[g]) + 1$, 
and $\rank(L_{\tmp}, g, L_{\tmp}[g+\epsilon]) = \rank(L_{j+1}, g, L_{j+1}[g])$. 
For case (ii), 
$\kappa(g) = 1$, $\lexCount(L_{\tmp}, L_{\tmp}[g+\epsilon]) = \lexCount(L_{\RLE, j+1}, L_{j+1}[g])$, 
and $\rank(L_{\tmp}, g, L_{\tmp}[g+\epsilon]) = \rank(L_{j+1}, g, L_{j+1}[g]) + 1$. 
For case (iii), 
$\kappa(g) = 0$, $\lexCount(L_{\tmp}, L_{\tmp}[g+\epsilon]) = \lexCount(L_{\RLE, j+1}, L_{j+1}[g])$, 
and $\rank(L_{\tmp}, g, L_{\tmp}[g+\epsilon]) = \rank(L_{j+1}, g, L_{j+1}[g])$. 
For case (iv), 
$\kappa(g) = 0$, $\lexCount(L_{\tmp}, L_{\tmp}[g+\epsilon]) = \lexCount(L_{\RLE, j+1}, L_{j+1}[g])$, 
and $\rank(L_{\tmp}, g, L_{\tmp}[g+\epsilon]) = \rank(L_{\RLE, j+1}, g, L_{j+1}[g])$. 
Therefore, $W_{+1}(g) + \kappa(g) = W_{\tmp}(g + \epsilon)$ always holds. 
Finally, Lemma~\ref{lem:optional_lemma} follows from 
$\LF_{j+1}(g) = W_{+1}(g) + \kappa(g)$ and $W_{+1}(g) + \kappa(g) = W_{\tmp}(g + \epsilon)$. 
\end{proof}

Let $\alpha$ be the position of $\SA_{j}[t]$ in $\SA_{\tmp}$.     
The following lemma states the relationship between $W(t)$ and $W_{\tmp}(\alpha)$. 

\begin{lemma}\label{lem:lexrank_0_relationship_C}
    The following equation holds. 
    \begin{equation}\label{eq:W_eq1}
        W(t) =  
    \begin{cases}
        W_{\tmp}(\alpha) - 1 & \text{if $x < W_{\tmp}(\alpha) \leq y$} \\
        W_{\tmp}(\alpha) + 1 & \text{if $y \leq W_{\tmp}(\alpha) < x$} \\
        W_{\tmp}(\alpha) & \text{otherwise.}
    \end{cases}
    \end{equation}
\end{lemma}
\begin{proof}
    From the definitions of $L_{j}$ and $L_{\tmp}$, 
    $L_{j}$ can be obtained by moving the character at position $\ISA[i]$ in $L_{\tmp}$ 
    to position $y_{+1}$. 
    Using this fact and the definitions of lex-count and rank queries, 
    we can prove the following equation: 
    \begin{equation*}
        W(t) =  
    \begin{cases}
        W_{\tmp}(\alpha) - 1 & \text{if $W_{\tmp}(\ISA[i]) < W_{\tmp}(\alpha) \leq W(y_{+1})$} \\
        W_{\tmp}(\alpha) + 1 & \text{if $W(y_{+1}) \leq W_{\tmp}(\alpha) < W_{\tmp}(\ISA[i])$} \\
        W_{\tmp}(\alpha) & \text{otherwise.}
    \end{cases}
    \end{equation*}
    This equation indicates that Lemma~\ref{lem:lexrank_0_relationship_C} holds 
    if $x = W_{\tmp}(\ISA[i])$ and $y = W(y_{+1})$. 
    Here, 
    $y = W(y_{+1})$ follows from the formula for computing $y$ stated in Step~(III). 

    We prove $x = W_{\tmp}(\ISA[i])$. 
    Lemma~\ref{lem:optional_lemma} shows that  
    $W_{\tmp}$ maps an integer $z$ in $\{ 1, 2, \ldots, n+1 \}$ 
    into $\{ 1, 2, \ldots, n+1 \} \setminus \{ x  \}$ except for the case $z = \ISA[i]$.
    This is because $x$ is the row index of $T^{[i-1]}$ in $\CM_{j}$ (i.e., $\SA_{j}[x] = i-1$), 
    and $\SA_{j+1}$ does not contain $i$. 
    On the other hand, 
    $W_{\tmp}$ is a bijective function that maps $\{ 1, 2, \ldots, n+1 \}$ to $\{ 1, 2, \ldots, n+1 \}$. 
    Therefore, $x = W_{\tmp}(\ISA[i])$ holds. 
\end{proof}

We prove $\LF_{j}(t) = W(t)$ (i.e., Lemma~\ref{lem:dynamic_LF_formula}) using Lemma~\ref{lem:optional_lemma} and Lemma~\ref{lem:lexrank_0_relationship_C}. 
If $\SA_{j}[t] = i$, 
then $t = y_{+1}$ and $\LF_{j}(y_{+1}) = y$ follow from the update algorithm for $L$ and $\SA$, 
$y = W(y_{+1})$ follows from the formula for computing $y$. 
Therefore, 
\begin{equation*}
    \begin{split}
        \LF_{j}(t) &= \LF_{j}(y_{+1}) \\
         &= y \\
          &= W(y_{+1}) \\
          &= W(t).
    \end{split}
\end{equation*}

Otherwise (i.e., $\SA_{j}[t] \neq i$), 
$\SA_{j}$ contains $(\SA_{j}[t] - 1)$ at a position $u$. 
In this case, $u \neq x$ because 
$x$ is the row index of $T^{[i-1]}$ in $\CM_{j}$ (i.e., $\SA_{j}[x] = i -1$). 
$\SA_{j}$ is updated to $\SA_{j-1}$ by moving the $x$-th integer in $\SA_{j}$ to the $y$-th integer. 
Therefore, $\LF_{j}(t)$ can be computed as follows: 
\begin{equation}\label{eq:u_eq1}
        \LF_{j}(t) =  
    \begin{cases}
        u - 1 & \text{if $x < u \leq y$} \\
        u + 1 & \text{if $y \leq u < x$} \\
        y & \text{if $u = x$} \\
        u & \text{otherwise.}
    \end{cases}
\end{equation}

Lemma~\ref{lem:optional_lemma} shows that $u = W_{\tmp}(\alpha)$. 
Since $u \neq x$, 
at least one of the following three cases occurs: 
(i) $x < W_{\tmp}(\alpha) \leq y$, (ii) $y \leq W_{\tmp}(\alpha) < x$, 
(iii) either $W_{\tmp}(\alpha) < \min \{ x, y \}$ or $W_{\tmp}(\alpha) > \max \{ x, y \}$ holds. 

For case (i), 
$W(t) = W_{\tmp}(\alpha) - 1$ follows from Lemma~\ref{lem:lexrank_0_relationship_C}. 
Therefore, 
\begin{equation*}
    \begin{split}
        \LF_{j}(t) &= u-1  \\
        &= W_{\tmp}(\alpha) - 1 \\
        &= W(t).
    \end{split}
\end{equation*}

For case (ii), 
$W(t) = W_{\tmp}(\alpha) + 1$ follows from Lemma~\ref{lem:lexrank_0_relationship_C}. 
Therefore, 
\begin{equation*}
    \begin{split}
        \LF_{j}(t) &= u+1  \\
        &= W_{\tmp}(\alpha) + 1 \\
        &= W(t).
    \end{split}
\end{equation*}

For case (iii), 
$W(t) = W_{\tmp}(\alpha)$ follows from Lemma~\ref{lem:lexrank_0_relationship_C}. 
Therefore, 
\begin{equation*}
    \begin{split}
        \LF_{j}(t) &= u  \\
        &= W_{\tmp}(\alpha) \\
        &= W(t).
    \end{split}
\end{equation*}

\paragraph{Proof of Lemma~\ref{lem:dynamic_LF_formula} for Later Iterations ($j \in \{ i-2, \ldots, 1 \}$)}
For these later iterations, 
$\SA_{j}$ is updated to $\SA_{j-1}$ by moving the $x$-th integer in $\SA_{j}$ to the $y$-th integer. 
Similarly, 
$L_{j}$ is updated to $L_{j-1}$ by moving the $x$-th character in $L_{j}$ to the $y$-th character. 
Using these facts and the inductive assumption, 
we obtain the following two equations: 
\begin{equation}\label{eq:W_eq2}
        W(t) =  
    \begin{cases}
        W_{+1}(\alpha) - 1 & \text{if $x < W_{+1}(\alpha) \leq y$} \\
        W_{+1}(\alpha) + 1 & \text{if $y \leq W_{+1}(\alpha) < x$} \\
        W_{+1}(\alpha) & \text{otherwise.}
    \end{cases}
\end{equation}
\begin{equation}\label{eq:u_eq2}
        \LF_{j}(t) =  
    \begin{cases}
        u - 1 & \text{if $x < u \leq y$} \\
        u + 1 & \text{if $y \leq u < x$} \\
        y & \text{if $u = x$} \\
        u & \text{otherwise.}
    \end{cases}
\end{equation}
Here, 
$\alpha$ is the position of $\SA_{j}[t]$ in $\SA_{j+1}$; 
$u$ is the position of $(\SA_{j}[t] - 1)$ in $\SA_{j}$; 
$W_{+1}(\alpha) = \lexCount(L_{\RLE, j+1}, L_{j+1}[\alpha]) + \rank(L_{\RLE, j+1}, \alpha, L_{j+1}[\alpha])$. 

Equation~\ref{eq:W_eq2} and Equation~\ref{eq:u_eq2} correspond to 
Equation~\ref{eq:W_eq1} and Equation~\ref{eq:u_eq1}, respectively. 
Thus, we can prove $\LF_{j}(t) = W(t)$ using the same approach as in the proof of Lemma~\ref{lem:dynamic_LF_formula} for the special iteration. 
Finally, we obtain Lemma~\ref{lem:dynamic_LF_formula}. 

\subsubsection{Supporting \texorpdfstring{$\compSA_{X, 1}$}{computeSAX}, \texorpdfstring{$\compSA_{X, 2}$}{computeSAX}, \texorpdfstring{$\compSA_{Y, 1}$}{computeSAY}, and \texorpdfstring{$\compSA_{Y, 2}$}{computeSAY}}\label{subsec:support_SAXY}
\begin{table}[htbp]
\centering
\caption{The four operations used in the update algorithm for $\SA_{s}$ and $\SA_{e}$.}
\label{tab:XY_operations}
\begin{tabular}{l|l|l}
\hline
Operation & Description & Time complexity \\
\hline
$\compSA_{X, 1}(L_{\RLE, j}, \SA_{s, j}, \SA_{e, j}, x_{-1})$ & return $\SA_{j-1}[x_{-1}+1]$ & $\mathcal{O}(\log n)$ \\
$\compSA_{X, 2}(L_{\RLE, j}, \SA_{s, j}, \SA_{e, j}, x_{-1})$ & return $\SA_{j-1}[x_{-1}-1]$ & $\mathcal{O}(\log n)$ \\
$\compSA_{Y, 1}(L_{\RLE, j}, \SA_{s, j}, \SA_{e, j}, y)$ & return $\SA_{j-1}[y+1]$ & $\mathcal{O}(\log n)$ \\
$\compSA_{Y, 2}(L_{\RLE, j}, \SA_{s, j}, \SA_{e, j}, y)$ & return $\SA_{j-1}[y-1]$ & $\mathcal{O}(\log n)$ \\
\hline
\end{tabular}
\end{table}

This section describes the four operations $\compSA_{X, 1}$, $\compSA_{X, 2}$, $\compSA_{Y, 1}$, and $\compSA_{Y, 2}$ used in the update algorithms for $\SA_s$ and $\SA_e$. 
These four operations are summarized in Table~\ref{tab:XY_operations} and computed using the dynamic LF function introduced in Section~\ref{subsubsec:dynamic_LF}.

\paragraph{Supporting $\compSA_{X, 1}(L_{\RLE, j}, \SA_{s, j}, \SA_{e, j}, x_{-1})$}
The operation $\compSA_{X, 1}(L_{\RLE, j}, \SA_{s, j}, \SA_{e, j}, x_{-1})$ computes $\SA_{j-1}[x_{-1}+1]$  
from $L_{\RLE, j}$, $\SA_{s, j}$, $\SA_{e, j}$, and the row index $x_{-1}$ of a circular shift in $\CM_{j-1}$,  
identified in Step~(I) of the update algorithm. 
Here, $j \in \{i+2, i, i-1, \ldots, i-K+1 \}$ holds 
since most of the $(n+1)$ iterations are skipped. 
Let $y$ be the insertion position in $\CM_j$ determined in Step~(III). 
The value $\SA_{j-1}[x_{-1}+1]$ is computed based on the following cases:  
(i) $j = i+2$;  
(ii) $j = i$ and $y = x_{-1} + 1$;  
(iii) $j = i$ and $y \ne x_{-1} + 1$;  
(iv) $j < i$. 

\textbf{Case (i): $j = i+2$.}
The following equation follows from Lemma~\ref{lem:state_stability}: 
\begin{align*}
    \SA_{j-1}[x_{-1}+1] &= 
  \begin{cases}
    \SA[\ISA[i]+1] + 1 & \text{if $\SA[\ISA[i]+1] \geq i+1$} \\
    \SA[\ISA[i]+1] & \text{otherwise.}
  \end{cases}
\end{align*}
This equation implies that $\SA_{j-1}[x_{-1}+1]$ can be computed using $\SA[\ISA[i+1]+1]$. 
$\SA[\ISA[i]+1]$ is computed by $\phi^{-1}(i)$, which is introduced in Section~\ref{sec:review_r_index}. 
This function can be precomputed in $\mathcal{O}(\log r)$ time using 
$\mathscr{D}_{\RLE}(L_{\RLE})$, $\mathscr{D}_{\DI}(\SA_{s})$, and $\mathscr{D}_{\DI}(\SA_{e})$. 

\textbf{Case (ii): $j = i$ and $y = x_{-1} + 1$.}
$i$ is inserted at position $x_{-1} + 1$ in $\SA_j$, so $\SA_{j-1}[x_{-1} + 1] = i$. 

\textbf{Case (iii): $j = i$ and $y \neq x_{-1} + 1$.}
Similar to case (i), 
$\SA_{j-1}[x_{-1}+1]$ is computed as follows: 
\begin{align*}
    \SA_{j-1}[x_{-1}+1] &= 
  \begin{cases}
    \SA[\ISA[i-1]+1] + 1 & \text{if $\SA[\ISA[i-1]+1] \geq i$} \\
    \SA[\ISA[i-1]+1] & \text{otherwise.}
  \end{cases}
\end{align*}
Here, $\SA[\ISA[i-1]+1]$ is precomputed by $\phi^{-1}(i-1)$. 

\textbf{Case (iv): $j < i$.}
We compute $\SA_{j-1}[x_{-1} + 1]$ using the dynamic \emph{LF} function introduced in Section~\ref{subsubsec:dynamic_LF}. 
The value $\SA_{j-1}[x_{-1}+1]$ is computed via $\LF_j$ by finding $h$ such that $\LF_j(h) = x_{-1}+1$ if $x_{-1} \neq |\SA_{j-1}|$, and $\LF_j(h) = 1$ otherwise. 
Such $h$ always exists since $\LF_j$ is bijective (except when $j = i$).
We then compute $\SA_{j-1}[x_{-1}+1]$ using the computed $h$ and $\LF_{j}$ as follows:
$\SA_{j-1}[x_{-1}+1] = \SA_{j}[h]-1$ if $\SA_{j}[h] \neq 1$; otherwise, $\SA_{j-1}[x_{-1}+1] = |\SA_{j-1}|$. 

$\SA_{j}[h]$ is computed from $\SA_{s, j}$ and the value $\SA_{j}[x+1]$ computed in Step (II) of the update algorithm, 
where $x$ is the row index in $\CM_{j}$ identified in Step (I). 
The following lemma shows that $\SA_{j}[h]$ is a value in $\SA_{s,j}$ or equals $\SA_{j}[x+1]$.

\begin{lemma}\label{lem:sah_lemma}
For any $j \in \{n+1, n, \ldots, i+2, i-1, \ldots, 1\}$,  
let $x$ be the row index in $\CM_j$ identified in Step~(I), and let $(c_v, \ell_v)$ be the run containing the $x$-th character in $L_j$,  
where $L_{\RLE, j} = (c_1, \ell_1), (c_2, \ell_2), \ldots, (c_r, \ell_r)$.  
Let $c_{\mathsf{succ}}$ be the smallest character greater than $c_v$ in $\{c_d \mid 1 \leq d \leq r\}$;  
if none exists, set $c_{\mathsf{succ}} = \$$.  
Then, the following holds:

\begin{equation}\label{eq:support_SAXY}
  \SA_j[h] = 
  \begin{cases}
    \SA_j[x+1] & \text{if } x \ne t_v + \ell_v - 1, \\
    \SA_{s,j}[\min \mathcal{U}] & \text{if } x = t_v + \ell_v - 1 \text{ and } \mathcal{U} \ne \emptyset, \\
    \SA_{s,j}[\min \mathcal{U}'] & \text{otherwise.}
  \end{cases}
\end{equation}
Here, $t_v$ is the starting position of the $v$-th run in $L_j$;  
$\mathcal{U} = \{d \mid v + 1 \leq d \leq r, c_d = c_v\}$;  
$\mathcal{U}' = \{d \mid 1 \leq d \leq r, c_d = c_{\mathsf{succ}}\}$.
\end{lemma}
\begin{proof}
For simplicity, 
let $W(t) = \lexCount(L_{\RLE, j}, L_{j}[t]) + \rank(L_{\RLE, j}, t, L_{j}[t])$ 
for a position $t$ in $\SA_{j}$. 
Let $x_{-1}$ be the row index of a circular shift in $\CM_{j-1}$ identified in Step~(I) of the update algorithm. 
$x_{-1} = W(x)$ follows from the formula for computing $x_{-1}$ stated in Step~(I). 
If $x_{-1} \neq |\SA_{j-1}|$, 
then we can show the following equation holds using LF formula, rank, and lex-count queries on $L_{j}$: 
\begin{equation}\label{eq:support_SAXY:1}
  W(x)+1 = 
  \begin{cases}
    W(x+1) & \text{if $x \neq t_{v} + \ell_{v} - 1$} \\
    W(t_{\alpha}) & \text{if $x = t_{v} + \ell_{v} - 1$ and $\mathcal{U} \neq \emptyset$} \\
    W(t_{\alpha'}) & \text{otherwise.}
  \end{cases}
\end{equation}
Here, let $\alpha = \min \mathcal{U}$ and $\alpha' = \min \mathcal{U}'$; 
$t_{\alpha}$ and $t_{\alpha'}$ are the starting positions of the $\alpha$-th and $\alpha'$-th runs in $L_{j}$, respectively. 
Using Equation~\ref{eq:support_SAXY:1}, 
Equation~\ref{eq:support_SAXY} is proved according to the following three cases: 
(i) If $x \neq t_{v} + \ell_{v} - 1$, 
then LF formula ensures that $h$ is a position in $\SA_{j}$ satisfying $W(h) = W(x+1)$ (i.e., $h = x+1$), 
and hence, $\SA_{j}[h] = \SA_{j}[x+1]$.
(ii) If $x = t_{v} + \ell_{v} - 1$ and $\mathcal{U} \neq \emptyset$, 
then LF formula ensures that $h$ is a position in $\SA_{j}$ satisfying $W(h) = W(t_{\alpha})$ (i.e., $h = t_{\alpha}$). 
$\SA_{j}[t_{\alpha}]$ is contained in $\SA_{s, j}$ as the $\alpha$-th value, 
and hence, $\SA_{j}[h] = \SA_{s, j}[\alpha]$ (i.e., $\SA_{j}[h] = \SA_{s, j}[\min \mathcal{U}]$).
(iii) Otherwise,  
LF formula ensures that $h$ is a position in $\SA_{j}$ satisfying $W(h) = W(t_{\alpha'})$ (i.e., $h = t_{\alpha'}$). 
$\SA_{j}[t_{\alpha'}]$ is contained in $\SA_{s, j}$ as the $\alpha'$-th value, 
and hence, $\SA_{j}[h] = \SA_{s, j}[\alpha']$ (i.e., $\SA_{j}[h] = \SA_{s, j}[\min \mathcal{U}']$).

Otherwise (i.e., $x_{-1} = |\SA_{j-1}|$), 
we can prove Equation~\ref{eq:support_SAXY} in a similar way.
\end{proof}

Lemma~\ref{lem:sah_lemma} shows that $\SA_{j+1}[h]$ is computed from 
$v$, $(c_{v}, \ell_{v})$, $t_{v}$, $\SA_{j}[x+1]$, 
$\SA_{s,j}[\min \mathcal{U}]$, and $\SA_{s,j}[\min \mathcal{U}']$. 
These values are computed as follows: 
\begin{itemize}
    \item \textbf{Computing $v$.} $v$ is computed by $\runIndex(L_{\RLE, j}, x)$
    \item \textbf{Computing $(c_{v}, \ell_{v})$ and $t_{v}$.} 
    The run $(c_{v}, \ell_{v})$ and its starting position $t_{v}$ are computed by $\runAccess(L_{\RLE, j}, v)$
    \item \textbf{Computing $\SA_{j}[x+1]$.} This value is already computed in the previous iteration. 
    \item \textbf{Computing $\SA_{s,j}[\min \mathcal{U}]$.} 
    $\min \mathcal{U} = \select(S_{1}, b + 1, c_{v})$ holds, where $S_{1} = c_{1}, c_{2}, \ldots, c_{r}$ is the sequence introduced in Section~\ref{subsec:dynamic_r_index}, and $b$ is the number of $c_{v}$ in $S_{1}[1..v]$ (i.e., $b = \rank(S_{1}, v, c_{v})$). 
    We compute $\min \mathcal{U}$ by $\select(S_{1}, b + 1, c_{v})$. 
    After that, $\SA_{s, j}[\min \mathcal{U}]$ is computed by $\accessSA(\SA_{s, j}, \min \mathcal{U})$. 
    \item \textbf{Computing $\SA_{s,j}[\min \mathcal{U}']$.} 
    Similarly, $\SA_{s, j}[\min \mathcal{U}']$ is computed in the following steps: 
    (i) compute $c_{\mathsf{succ}}$ by $\lexSearch(L_{\RLE, j}, \lexCount(L_{\RLE, j}, c_{v}+1))$; 
    (ii) compute $\min \mathcal{U}'$ by $\select(S_{1}, 1, c_{\mathsf{succ}})$. 
\end{itemize}

Finally, $\SA_{j-1}[x_{-1}+1]$ can be computed in $\mathcal{O}(\log r')$ time using 
$\mathscr{D}_{\RLE}(L_{\RLE, j})$, $\mathscr{D}_{\DI}(\SA_{s, j})$, and $\mathscr{D}_{\DI}(\SA_{e, j})$, 
where $r'$ is the number of runs in $L_{\RLE, j}$. 
$r' \leq n+1$ always holds, and hence, 
the time complexity of $\compSA_{X, 1}$ is $\mathcal{O}(\log n)$. 

\paragraph{Supporting $\compSA_{X, 2}(L_{\RLE, j}, \SA_{s, j}, \SA_{e, j}, x_{-1})$}
The operation $\compSA_{X, 2}(L_{\RLE, j}, \SA_{s, j}, \SA_{e, j}, x_{-1})$ computes the value $\SA_{j-1}[x_{-1}-1]$ from given $L_{\RLE, j}$, $\SA_{s, j}$, $\SA_{e, j}$, and row index $x_{-1}$ for each $j \in \{i+2, i, i-1, \ldots, i-K+1 \}$.
Similar to $\compSA_{X, 1}$, 
the value $\SA_{j-1}[x_{-1}-1]$ is computed based on the following cases:  
(i) $j = i+2$;  
(ii) $j = i$ and $y = x_{-1} - 1$;  
(iii) $j = i$ and $y \ne x_{-1} - 1$;  
(iv) $j < i$. 

\textbf{Case (i): $j = i+2$.}
$\SA_{j-1}[x_{-1}-1]$ is computed as follows: 
\begin{align*}
    \SA_{j-1}[x_{-1}+1] &= 
  \begin{cases}
    \SA[\ISA[i]-1] + 1 & \text{if $\SA[\ISA[i]-1] \geq i+1$} \\
    \SA[\ISA[i]-1] & \text{otherwise.}
  \end{cases}
\end{align*}
Here, $\SA[\ISA[i]-1]$ is computed via a function $\phi(i)$. 
This function returns $\phi(\SA[t]) = \SA[t-1]$ for all $t \in \{ 1, 2, \ldots, n \}$. 
Similar to $\phi^{-1}$, 
this function can be precomputed in $\mathcal{O}(\log r)$ time using 
$\mathscr{D}_{\RLE}(L_{\RLE})$, $\mathscr{D}_{\DI}(\SA_{s})$, $\mathscr{D}_{\DI}(\SA_{e})$, 
and the following lemma. 
\begin{lemma}[\cite{10.1145/3375890}]\label{lem:phi}
For an integer $j \in \{ 1, 2, \ldots, n \}$, 
let $k$ be the position in $\SA_{s}$ such that $\SA_{s}[k]$ is the largest value in $\SA_{s}$ less than $(j+1)$ (i.e., $k$ is the position of $u$ in $\SA_{s}$, where $u = \max \{ \SA_{s}[u'] \mid 1 \leq u' \leq r \text{ s.t. }  \SA_{s}[u'] < j+1 \}$). 
For simplicity, we define $\SA_{e}[0] = \SA_{e}[r]$. 
Then, $\phi(j) = (j - \SA_{s}[k]) + \SA_{e}[k-1]$. 
\end{lemma}

\textbf{Case (ii): $j = i$ and $y = x_{-1} - 1$.} 
$i$ is inserted at position $x_{-1} -1$ in $\SA_j$, so $\SA_{j-1}[x_{-1} - 1] = i$. 

\textbf{Case (iii): $j = i$ and $y \neq x_{-1} - 1$.} 
Similar to case (i), 
$\SA_{j-1}[x_{-1}-1]$ is computed as follows: 
\begin{align*}
    \SA_{j-1}[x_{-1}-1] &= 
  \begin{cases}
    \SA[\ISA[i-1]-1] + 1 & \text{if $\SA[\ISA[i-1]-1] \geq i$} \\
    \SA[\ISA[i-1]-1] & \text{otherwise.}
  \end{cases}
\end{align*}
Here, $\SA[\ISA[i-1]-1]$ is precomputed by $\phi(i-1)$. 

\textbf{Case (iv): $j < i$.} 
The value $\SA_{j-1}[x_{-1}-1]$ is computed using $\LF_j$ as follows. 
We compute a position $h'$ in $\SA_{j}$ such that 
$\LF_{j}(h') = x_{-1}-1$ if $x_{-1} \neq 1$; 
otherwise, $\LF_{j}(h') = |\SA_{j-1}|$. 
We then compute $\SA_{j-1}[x_{-1}-1]$ using the computed $h$ and $\LF_{j}$ as follows:
$\SA_{j-1}[x_{-1}-1] = \SA_{j}[h']-1$ if $\SA_{j}[h'] \neq 1$; otherwise, $\SA_{j-1}[x_{-1}-1] = n+1$. 

$\SA_{j}[h']$ is computed from $\SA_{e, j}$ and the value $\SA_{j}[x-1]$ computed in Step~(II) of the update algorithm. 
The following lemma ensures that $\SA_{j}[h']$ is either a value of $\SA_{e, j}$ 
or equal to $\SA_{j}[x-1]$. 

\begin{lemma}\label{lem:sah_prime_lemma}
For an integer $j \in \{ n+1, n, \ldots, i+2, i-1, \ldots, 1 \}$, 
let $(c_{v}, \ell_{v})$ be the run in $L_{\RLE, j}$ that contains the $x$-th character in $L_{j}$, 
where $L_{\RLE, j} = (c_{1}, \ell_{1}), (c_{2}, \ell_{2}), \ldots, (c_{r}, \ell_{r})$. 
Let $c_{\mathsf{pred}}$ be the largest character larger than $c_{v}$ 
in set $\{ c_{d} \mid 1 \leq d \leq r  \}$. 
If such character does not exist, then $c_{\mathsf{pred}}$ is defined as the largest character in set $\{ c_{d} \mid 1 \leq d \leq r  \}$. 
Then, the following equation holds: 

\begin{equation*}
  \SA_{j}[h'] = 
  \begin{cases}
    \SA_{j}[x-1] & \text{if $x \neq t_{v}$} \\
    \SA_{e, j}[\max \mathcal{V}] & 
    \text{if $x = t_{v}$ and $\mathcal{V} \neq \emptyset$} \\
    \SA_{e, j}[\max \mathcal{V}'] & \text{otherwise.}
  \end{cases}
\end{equation*}
Here, $t_{v}$ is the starting position of the $v$-th run in $L_{j}$; 
$\mathcal{V} = \{ 1 \leq d \leq v-1  \mid c_{d} = c_{v} \}$; 
$\mathcal{V}' = \{ 1 \leq d \leq r \mid c_{d} = c_{\mathsf{pred}} \}$.
\end{lemma}
\begin{proof}
    Lemma~\ref{lem:sah_prime_lemma} can be proved using a similar technique to that of Lemma~\ref{lem:sah_lemma}. 
\end{proof}

Similar to $\compSA_{X, 1}$, 
$\SA_{j}[h']$ can be computed in $\mathcal{O}(\log n)$ time using Lemma~\ref{lem:sah_prime_lemma}. 
Therefore, the time complexity of $\compSA_{X, 2}$ is $\mathcal{O}(\log n)$. 

\paragraph{Supporting $\compSA_{Y, 1}(L_{\RLE, j}, \SA_{s, j}, \SA_{e, j}, y)$}
The operation $\compSA_{Y, 1}(L_{\RLE, j}, \SA_{s, j}, \SA_{e, j}, y)$ computes the value $\SA_{j-1}[y+1]$ from given $L_{\RLE, j}$, $\SA_{s, j}$, $\SA_{e, j}$, and the appropriate row index $y$ in $\CM_{j}$ for insertion in Step~(III) of the update algorithm, where $j \in \{ i+1, i, \ldots, i-K \}$. 
$\SA_{j-1}[y+1]$ is computed according to the four cases: 
(i) $j = i+1$; (ii) $j = i$ and $y = \ISA[i-1]$; (iii) $j = i$ and $y \neq \ISA[i-1]$; 
(iv) $j \leq i-1$. 

\textbf{Case (i): $j = i+1$.}
$\SA_{j-1}[y+1] = \SA_{j}[x+1]$ because 
$y = x$ follows from the formula for computing $y$ in Step (III). 
This value $\SA_{j}[x+1]$ has been computed in Step~(II) of the update algorithm. 

\textbf{Case (ii): $j = i$ and $y = \ISA[i-1]$.}
$i$ is inserted at position $\ISA[i-1]$ in $\SA_j$, so $\SA_{j-1}[y+1] = \SA_{j}[y]$ 
(i.e., $\SA_{j-1}[y+1] = \SA_{i}[\ISA[i-1]]$). 
Using Lemma~\ref{lem:state_stability}, 
$\SA_{i}[\ISA[i-1]]$ is computed as follows: 
\begin{align*}
    \SA_{i}[\ISA[i-1]] &= 
  \begin{cases}
    \SA[\ISA[i-1]] + 1 & \text{if $\SA[\ISA[i-1]] \geq i$} \\
    \SA[\ISA[i-1]] & \text{otherwise.}
  \end{cases}
\end{align*}
Here, $\SA[\ISA[i-1]] = i-1$ if $i \neq 1$; 
otherwise $\SA[\ISA[i-1]] = n$. 

\textbf{Case (iii): $j = i$ and $y \neq \ISA[i-1]$.}
Similar to the case (iv) for $\compSA_{X, 1}$, 
the value $\SA_{j-1}[y+1]$ is computed using $\LF_j$ as follows. 
We compute a position $g$ in $\SA_{j}$ such that 
$\LF_{j}(g) = y+1$ if $y \neq |\SA_{j-1}|$; 
otherwise, $\LF_{j}(g) = 1$. 
We then compute $\SA_{j-1}[y+1]$ using the computed $g$ and $\LF_{j}$ as follows:
$\SA_{j-1}[y+1] = \SA_{j}[g]-1$ if $\SA_{j}[g] \neq 1$; otherwise, $\SA_{j-1}[y+1] = n+1$. 

$\SA_{j}[g]$ is computed from $\SA_{s, j}$ and the value $\SA_{j}[y_{+1}+1]$ computed in Step~(III) of the previous iteration, 
where $y_{+1}$ is the appropriate row index in $\CM_{j+1}$ for insertion. 
The following lemma ensures that $\SA_{j}[g]$ is either a value of $\SA_{s, j}$ or equal to $\SA_{j}[y_{+1}+1]$. 

\begin{lemma}\label{lem:sag1_lemma}
For an integer $j \in \{ n, n-1, \ldots, 1 \}$, 
let $(c_{v}, \ell_{v})$ be the run in $L_{\RLE, j}$ that contains the $y_{+1}$-th character in $L_{j}$, 
where $L_{\RLE, j} = (c_{1}, \ell_{1}), (c_{2}, \ell_{2}), \ldots, (c_{r}, \ell_{r})$. 
Let $c_{\mathsf{succ}}$ be the largest character larger than $c_{v}$ 
in set $\{ c_{d} \mid 1 \leq d \leq r  \}$. 
If such character does not exist, then $c_{\mathsf{succ}}$ is defined as $\$$. 
Unless $j = i$ and $y \neq \ISA[i-1]$, 
the following equation holds.  

\begin{equation*}
  \SA_{j}[g] = 
  \begin{cases}
    \SA_{j}[y_{+1}+1] & \text{if $y_{+1} \neq t_{v} + \ell_{v} + 1$} \\
    \SA_{s, j}[\min \mathcal{U}] & 
    \text{if $y_{+1} = t_{v} + \ell_{v} + 1$ and $\mathcal{U} \neq \emptyset$} \\
    \SA_{s, j}[\min \mathcal{U}'] & \text{otherwise.}
  \end{cases}
\end{equation*}
Here, $t_{v}$ is the starting position of the $v$-th run in $L_{j}$; 
$\mathcal{U} = \{ v+1 \leq d \leq r  \mid c_{d} = c_{v} \}$; 
$\mathcal{U}' = \{ 1 \leq d \leq r \mid c_{d} = c_{\mathsf{succ}} \}$.
\end{lemma}
\begin{proof}
    Lemma~\ref{lem:sag1_lemma} can be proved using a similar technique to that of Lemma~\ref{lem:sah_lemma}. 
\end{proof}

Similar to $\compSA_{X, 1}$, 
$\SA_{j}[g]$ can be computed in $\mathcal{O}(\log n)$ time using Lemma~\ref{lem:sag1_lemma}. 

\textbf{Case (iv): $j \leq i - 1$.}
Similar to case (iii), 
we compute $\SA_{j-1}[y+1]$ using Lemma~\ref{lem:sag1_lemma}. 
Finally, the time complexity of $\compSA_{Y, 1}$ is $\mathcal{O}(\log n)$. 

\paragraph{Supporting $\compSA_{Y, 2}(L_{\RLE, j}, \SA_{s, j}, \SA_{e, j}, y)$}
The operation $\compSA_{Y, 2}(L_{\RLE, j}, \SA_{s, j}, \SA_{e, j}, y)$ computes the value $\SA_{j-1}[y-1]$ from given $L_{\RLE, j}$, $\SA_{s, j}$, $\SA_{e, j}$, and the row index $y$ in $\CM_{j}$ for insertion in Step~(III) of the update algorithm. 
Similar to $\compSA_{Y, 1}$, 
$\SA_{j-1}[y-1]$ is computed according to the four cases: 
(i) $j = i+1$; (ii) $j = i$ and $y = \ISA[i-1]$; (iii) $j = i$ and $y \neq \ISA[i-1]$; 
(iv) $j \leq i-1$. 

\textbf{Case (i): $j = i+1$.}
Similar to case (i) for $\compSA_{Y, 1}$, 
$\SA_{j-1}[y-1] = \SA_{j}[x-1]$ holds, 
where $\SA_{j}[x-1]$ has been computed in Step~(II) of the update algorithm. 

\textbf{Case (ii): $j = i$ and $y = \ISA[i-1] + 1$.}
$i$ is inserted at position $\ISA[i-1] + 1$ in $\SA_j$, so $\SA_{j-1}[y-1] = \SA_{j}[\ISA[i-1]]$ 
(i.e., $\SA_{j-1}[y-1] = \SA_{i}[\ISA[i-1]]$). 
Using Lemma~\ref{lem:state_stability}, 
$\SA_{i}[\ISA[i-1]]$ is computed as follows: 
\begin{align*}
    \SA_{i}[\ISA[i-1]] &= 
  \begin{cases}
    i-1 & \text{if $i \neq 1$} \\
    n & \text{otherwise.}
  \end{cases}
\end{align*}

\textbf{Case (iii): $j = i$ and $y = \ISA[i-1] + 1$.}
Similar to the case (iv) for $\compSA_{X, 1}$, 
the value $\SA_{j-1}[y-1]$ is computed using $\LF_j$ as follows. 
We compute a position $g'$ in $\SA_{j}$ such that 
$\LF_{j}(g') = y-1$ if $y \neq 1$; 
otherwise, $\LF_{j}(g') = |\SA_{j-1}|$. 
We then compute $\SA_{j-1}[y-1]$ using the computed $g'$ and $\LF_{j}$ as follows:
$\SA_{j-1}[y-1] = \SA_{j}[g']-1$ if $\SA_{j}[g'] \neq 1$; otherwise, $\SA_{j-1}[y+1] = n+1$. 

$\SA_{j}[g']$ is computed from $\SA_{e, j}$ and the value $\SA_{j}[y_{+1}+1]$ computed in Step~(III) of the previous iteration, 
where $y_{+1}$ is the appropriate row index in $\CM_{j+1}$ for insertion. 
The following lemma ensures that $\SA_{j}[g']$ is either a value of $\SA_{e, j}$ or equal to $\SA_{j}[y_{+1}-1]$. 

\begin{lemma}\label{lem:sag2_lemma}
For an integer $j \in \{ n, n-1, \ldots, 1 \}$, 
let $(c_{v}, \ell_{v})$ be the run in $L_{\RLE, j}$ that contains the $y_{+1}$-th character in $L_{j}$, 
where $L_{\RLE, j} = (c_{1}, \ell_{1}), (c_{2}, \ell_{2}), \ldots, (c_{r}, \ell_{r})$. 
Let $c_{\mathsf{pred}}$ be the largest character larger than $c_{v}$ 
in set $\{ c_{d} \mid 1 \leq d \leq r  \}$. 
If such character does not exist, then $c_{\mathsf{pred}}$ is defined as the largest character in set $\{ c_{d} \mid 1 \leq d \leq r  \}$. 
Unless $j = i$ and $y \neq \ISA[i-1] + 1$, 
the following equation holds.  
\begin{equation*}
  \SA_{j}[g'] = 
  \begin{cases}
    \SA_{j}[y_{+1}-1] & \text{if $y_{+1} \neq t_{v}$} \\
    \SA_{e, j}[\max \mathcal{V}] & 
    \text{if $y_{+1} = t_{v}$ and $\mathcal{V} \neq \emptyset$} \\
    \SA_{e, j}[\max \mathcal{V}'] & \text{otherwise.}
  \end{cases}
\end{equation*}
Here, $t_{v}$ is the starting position of the $v$-th run in $L_{j}$; 
$\mathcal{V} = \{ 1 \leq d \leq v-1  \mid c_{d} = c_{v} \}$; 
$\mathcal{V}' = \{ 1 \leq d \leq r \mid c_{d} = c_{\mathsf{pred}} \}$.
\end{lemma}
\begin{proof}
    Lemma~\ref{lem:sag2_lemma} can be proved using a similar technique to that of Lemma~\ref{lem:sah_lemma}. 
\end{proof}
Similar to $\compSA_{X, 1}$, 
$\SA_{j}[g']$ can be computed in $\mathcal{O}(\log n)$ time using Lemma~\ref{lem:sag2_lemma}. 

\textbf{Case (iv): $j \leq i - 1$.}
Similar to case (iii), 
we compute $\SA_{j-1}[y-1]$ using Lemma~\ref{lem:sag2_lemma}. 
Finally, the time complexity of $\compSA_{Y, 2}$ is $\mathcal{O}(\log n)$. 

\subsubsection{Supporting \texorpdfstring{$\compISA_{1}$}{computeISA}, 
\texorpdfstring{$\compISA_{2}$}{computeISA}, and \texorpdfstring{$\compT$}{computeT}}\label{subsec:support_ISA}
\begin{table}[htbp]
\centering
\caption{The three operations used in the update algorithm for $L_{\RLE}$, $\SA_{s}$, and $\SA_{e}$. Here, $i$ is the position of the character $c$ inserted into $T$; $\lambda$ is the position in $\SA_{s}$ such that $\SA_{s}[\lambda]$ is the largest value in $\SA_{s}$ less than $(i+1)$.}
\label{tab:five_operations}
\begin{tabular}{l|l|l}
\hline
Operation & Description & Time complexity \\
\hline
$\compISA_{1}(L_{\RLE}, \SA_{s}, i)$ & return $\ISA[i-1]$ & $\mathcal{O}((1 + i - \SA_s[\lambda]) \log r)$  \\
$\compISA_{2}(L_{\RLE}, \SA_{s}, i)$ & return $\ISA[i]$ & $\mathcal{O}((1 + i - \SA_s[\lambda]) \log r)$  \\
$\compT(L_{\RLE}, \SA_{s}, i)$ & return $T[i-1]$ & $\mathcal{O}((1 + i - \SA_s[\lambda]) \log r)$ \\
\hline
\end{tabular}
\end{table}
This section describes the three operations $\compISA_{1}$, $\compISA_{2}$, and $\compT$, 
where $\compISA_{1}$ and $\compT$ are used in Step~(I); 
$\compISA_{2}$ is used in Step~(III). 
These three operations are summarized in Table~\ref{tab:five_operations}.

\paragraph{Supporting $\compISA_{1}(L_{\RLE}, \SA_{s}, i)$}
The operation $\compISA_{1}(L_{\RLE}, \SA_s, i)$ computes $\ISA[i - 1]$  
using $L_{\RLE}$, $\SA_s$, and the insertion position $i$ of the character $c$ into $T$. It relies on the LF function and its inverse $\LF^{-1}$, where  
$\LF^{-1}(\LF(j)) = j$ for all $j \in \{1, 2, \ldots, n\}$. 
Let $\lambda$ be the position in $\SA_{s}$ such that $\SA_{s}[\lambda]$ is the largest value in $\SA_{s}$ less than $(i+1)$. 
Then, $\ISA[i - 1]$ can be computed using $\LF$ and $\LF^{-1}$ based on the following two observations derived from the definitions of $\LF$ and $\LF^{-1}$:
(i) $\ISA[i]$ can be obtained by applying $\LF^{-1}$ to $t_{\lambda}$ recursively $(i - \SA_{s}[\lambda])$ times, 
where $t_{\lambda}$ is the starting position of the $\lambda$-th run in $L_{\RLE}$; 
(ii) $\ISA[i-1] = \LF(\ISA[i])$. 

Thus, $\ISA[i-1]$ can be computed in the following four steps: 
\begin{enumerate}
    \item compute the position $\lambda$ in $\SA_{s}$ by $\orderSA(\SA_{s}, \countSA(\SA_{s}, i+1))$;
    \item compute the starting position $t_{\lambda}$ of the $\lambda$-th run by $\runAccess(L_{\RLE}, \lambda)$; 
    \item compute $\ISA[i]$ by applying $\LF^{-1}$ to $t_{\lambda}$ recursively $(i - \SA_{s}[\lambda])$ times;
    \item compute $\ISA[i-1]$ by applying $\LF$ to $\ISA[i]$. 
\end{enumerate}

The first and second steps take $\mathcal{O}(\log r)$ time. 
The third step uses $\LF^{-1}$ $(i - \SA_{s}[\lambda])$ times. 
$\LF^{-1}$ can be computed using select, lex-search, and lex-count queries on $L_{\RLE}$, based on the LF formula 
(see Section~\ref{subsubsec:dynamic_LF}) as follows: $\LF^{-1}(j) = \select(L_{\RLE}, j', c')$, 
where $c' = \lexSearch(L_{\RLE}, j)$ and $j' = j - \lexCount(L_{\RLE}, c')$. 
$\mathscr{D}_{\RLE}(L_{\RLE})$ supports these queries in $\mathcal{O}(\log r / \log \log r)$, 
and hence, we can support $\LF^{-1}$ in the same time. 

The fourth step uses $\LF$, which can be computed in $\mathcal{O}(\log r / \log \log r)$ time using LF formula. 
Therefore, $\compISA_{1}$ takes $\mathcal{O}((1 + i - \SA_s[\lambda]) \log r)$ time in total. 

\paragraph{Supporting $\compISA_{2}(L_{\RLE}, \SA_{s}, i)$}
This operation computes $\ISA[i]$ from $L_{\RLE}$, $\SA_s$, and $i$. 
We use the algorithm of $\compISA_{1}$ for computing $\ISA[i]$ 
because this algorithm computes $\ISA[i]$ in the third step. 
Therefore, $\compISA_{2}$ takes $\mathcal{O}((1 + i - \SA_s[\lambda]) \log r)$ time.

\paragraph{Supporting $\compT(L_{\RLE}, \SA_s, i)$}
This operation computes $T[i-1]$ from $L_{\RLE}$, $\SA_s$, and $i$. 
By the definitions of RLBWT and ISA, 
$T[i-1] = c_{u}$ holds, where $(c_{u}, \ell_{u})$ is the run in $L_{\RLE}$ that contains the character at position $\ISA[i]$. 
Thus, $T[i-1]$ can be computed in the following three steps: 
\begin{enumerate}
    \item compute $\ISA[i]$ by $\compISA_{2}(L_{\RLE}, \SA_{s}, i)$;
    \item compute the index $u$ of the $u$-th run by $\runIndex(L_{\RLE}, \ISA[i])$;
    \item compute $c_{u}$ by $\runAccess(L_{\RLE}, u)$;
    \item return $c_{u}$ as $T[i-1]$. 
\end{enumerate}

The bottleneck of this algorithm is the first step, which takes $\mathcal{O}((1 + i - \SA_s[\lambda]) \log r)$ time. 
Therefore, $\compT$ runs in the same time. 

\subsubsection{Overall time complexity.}
$L_{\RLE}$, $\SA_{s}$, and $\SA_{e}$ are updated by executing $\mathcal{O}(K)$ iterations. 
Each iteration $j \in \{ i+1, i, \ldots, i-K \}$ can be executed in $\mathcal{O}(1)$ queries and update operations supported by 
$\mathscr{D}_{\RLE}(L_{\RLE, j})$, $\mathscr{D}_{\DI}(\SA_{s, j})$, and $\mathscr{D}_{\DI}(\SA_{e, j})$. 
Thus, the iteration $j$ can be executed in $\mathcal{O}(\log r')$, 
where $r'$ is the number of runs in $L_{\RLE, j}$, and $r' \leq n+1$ holds. 
Here, this time complexity includes the four auxiliary operations 
$\compSA_{X, 1}$, $\compSA_{X, 2}$, $\compSA_{Y, 1}$, and $\compSA_{Y, 2}$, 
which are described in Section~\ref{subsec:support_SAXY}. 

In addition, we need to compute $\ISA[i-1]$, $\ISA[i]$, and $T[i-1]$ before executing the update procedure. 
These values are computed in $\mathcal{O}((1 + i - \SA_s[\lambda]) \log r)$ by the three auxiliary operations $\compISA_{1}$, $\compISA_{2}$, and $\compT$, which are described in Section~\ref{subsec:support_ISA}. 
Here, $\lambda$ is introduced in Section~\ref{subsec:support_ISA}. 
Therefore, this update algorithm requires $\mathcal{O}((K + i - \SA_{s}[\lambda])\log n)$ time in total. 

\subsection{Performance Analysis}\label{subsec:performance_analysis}
This section describes the performance of the dynamic r-index. 
The following theorem summarizes the results of this section.

\begin{theorem}\label{theo:query_time_summary}
    The following two properties hold for the dynamic r-index: 
    (i) It can support count and locate queries in $\mathcal{O}(m \log r / \log \log r)$ time and $\mathcal{O}(m (\log r / \log \log r) + \occ \log r)$ time, respectively, 
    where $m$ is the length of a given pattern with $\occ$ occurrences. 
    (ii) It can be updated in $\mathcal{O}((1 + L_{\max}) \log n)$ time for character insertion.
    The average time complexity is $\mathcal{O}((1 + L_{\avg}) \log n)$. 
\end{theorem}

\subsubsection{Worst-case Time Complexity for Count and Locate Queries}
The dynamic r-index can execute the original algorithms of the static r-index for count and locate queries. 
Thus, count and locate queries take $\mathcal{O}(m \log r / \log \log r)$ time and $\mathcal{O}(m (\log r / \log \log r) + \occ \log r)$ time, respectively, for a pattern $P$ of length $m$ with $\occ$ occurrences. 
The detailed analysis of these queries is as follows. 

\paragraph{Supporting Count Query}
The count query is supported using the backward search, which uses $\mathcal{O}(m)$ rank and lex-count queries on $L_{\RLE}$. 
The dynamic r-index can support rank and lex-count queries on $L_{\RLE}$ in $\mathcal{O}(\log r / \log \log r)$. 
Therefore, the count query can be supported in $\mathcal{O}(m \log r / \log \log r)$ time. 

\paragraph{Supporting Locate Query}
The locate query is supported using $\SA[sp]$ and $\phi^{-1}$ after executing the backward search, 
where $[sp, ep]$ is the sa-interval of $P$. 
$\SA[sp]$ and $\phi^{-1}(i)$ are computed as follows.

\textbf{Computing $\SA[sp]$.}
Consider the $m$ sa-intervals $[sp_{1}, ep_{2}]$, $[sp_{1}, ep_{2}]$, $\ldots$, $[sp_{m}, ep_{m}]$ computed by the backward search, where $m$ is the length of $P$. 
Similar to the r-index, 
$\SA[sp]$ is computed in the following three phases: 

In the first phase, 
we compute the integer $i'$ in Corollary~\ref{cor:toe} by 
verifying whether 
$P[m - i] = c_{v}$ or not for each integer $i \in \{ 1, 2, \ldots, m-1 \}$, 
where $(c_{v}, \ell_{v})$ is the run that contains the $sp_{i}$-th character in $L$. 
Here, $v$ is computed by $\runIndex(L, sp_{i})$. 
Therefore, this phase takes $\mathcal{O}(m \log r / \log \log r)$ time.

The second phase is executed if the integer $i'$ does not exist. 
In this case, 
we compute $\SA[sp]$ using $\SA_{s}[u]$, 
where $u \in \{ 1, 2, \ldots, r \}$ is the smallest integer 
satisfying $c_{u} = P[m]$. 
Here, Corollary~\ref{cor:toe} ensures that $\SA[sp] = \SA_{s}[u] - m$. 
$u$ is computed by $\select(S_{1}, 1, P[m])$, 
where $S_{1}$ is the sequence introduced in Section~\ref{subsec:dynamic_r_index}. 
We can compute $\SA_{s}[u]$ in $\mathcal{O}(\log r)$ time using $\mathscr{D}_{\DI}(\SA_{s})$. 
Therefore, this phase takes $\mathcal{O}(\log r)$ time. 

The third phase is executed if the integer $i'$ exists.
In this case, 
we compute $\SA[sp]$ using $\SA_{s}[v']$, 
where $v' \in \{ 1, 2, \ldots, r \}$ is the smallest integer 
satisfying $v' \geq v+1$ and $P[m - i] = c_{v'}$. 
Here, Corollary~\ref{cor:toe} ensures that $\SA[sp] = \SA_{s}[v'] - (m - i')$. 
Similar to $u$, 
$v'$ can be computed by $\select(S_{1}, 1 + d, P[m - i])$, 
where $d = \rank(S_{1}, v, P[m - i])$. 
We can compute $\SA_{s}[v']$ in $\mathcal{O}(\log r)$ time using $\mathscr{D}_{\DI}(\SA_{s})$. 
Therefore, this phase takes $\mathcal{O}(\log r)$ time. 
Finally, $\SA[sp]$ can be computed in $\mathcal{O}(\log r + m \log r / \log \log r)$ time.

\textbf{Supporting $\phi^{-1}(i)$.}
Lemma~\ref{lem:phi_inv} shows that 
$\phi^{-1}(i)$ can be computed using $\SA_{e}[j]$ and $\SA_{s}[j+1]$.
$j$ can be computed by $\orderSA(\SA_{e}, \countSA(\SA_{e}, i+1))$. 
$\SA_{e}[j]$ and $\SA_{s}[j+1]$ are computed by two access queries on $\SA_{e}$ and $\SA_{s}$. 
These queries can be computed in $\mathcal{O}(\log r)$ time. 
Therefore, we can support $\phi^{-1}(i)$ in $\mathcal{O}(\log r)$ time. 

The locate query is answered using the function $\phi^{-1}$ $\mathcal{O}(m)$ times after computing $\SA[sp]$. 
Therefore, this query can be supported in $\mathcal{O}(m (\log r / \log \log r) + \occ \log r)$ time, 
where $\occ$ is the number of occurrences of $P$ in $T$.

\subsubsection{Worst-case Time Complexity for Character Insertion}
When inserting a character $c$ into $T$ at position $i$, 
the insertion takes $\mathcal{O}((K + i - \SA_s[\lambda]) \log n)$ time, 
where $K$ iterations are needed for the main update algorithm, 
and $(i - \SA_s[\lambda])$ steps are required for auxiliary operations. 
Both terms are bounded by values in the LCP array, leading to 
$\mathcal{O}((1 + L_{\max}) \log n)$ worst-case time, 
where $L_{\max}$ denotes the maximum value in $\LCP$. 
This is because (i) Salson et al. \cite{DBLP:journals/jda/LeonardMS12} showed that $K$ is at most a constant factor of a value in the $\LCP$ array; (ii) there exists a position $u$ in $T$ such that 
substring $T[\SA_{s}[\lambda]..i-1]$ is 
a common prefix between two different suffixes $T[u..n]$ and $T[\SA_{s}[\lambda]..n]$, 
which indicates that $(i - \SA_s[\lambda])$ is also at most a constant factor of a value in the $\LCP$ array. 

We provide a formal proof of this worst-case time complexity using 
the reversed string $T^{R}$ of $T$, 
where $T^{R}$ is defined as $T[n], T[n-1], \ldots, T[1]$. 
Let $\SA^{R}$ and $\LCP^{R}$ be the suffix and LCP arrays of $T^{R}$, respectively. 
The following lemma ensures that
$(K + i - \SA_{s}[\lambda])$ is at most a constant factor of a value in $\LCP^{R}$.
\begin{lemma}\label{lem:bound_K_LCP_R}
    There exists a constant $\alpha > 0$ satisfying the following two conditions: 
    \begin{enumerate}[(i)]
        \item $i - \SA_{s}[\lambda] \leq \alpha (1 + \max \{ \LCP^{R}[p], \LCP^{R}[p+1] \})$;
        \item $K - 1 \leq \alpha(1 + \max \{ \LCP^{R}[p], \LCP^{R}[p+1] \})$.        
    \end{enumerate}
Here, $p$ is the position of $(n-i+2)$ in $\SA^{R}$. 
\end{lemma}
\begin{proof}
    See Sections~\ref{paragraph:bound_K_LCP_R1} and \ref{paragraph:bound_K_LCP_R2}. 
\end{proof}

Lemma~\ref{lem:bound_K_LCP_R} indicates that 
the time complexity for insertion is $\mathcal{O}(\max \{ \LCP^{R}[p], \LCP^{R}[p+1] \} \log n)$. 
Ohlebusch et al. proved that the values in $\LCP^{R}$ are a permutation of the values in $\LCP$ (Lemma 6.1 in \cite{DBLP:journals/jda/OhlebuschBA14}). 
Therefore, there exists a position $p'$ in $\LCP$ such that 
the time complexity for insertion is $\mathcal{O}(\LCP[p'] \log n)$ time. 
$\LCP[p'] \leq L_{\max}$, 
and hence 
the time complexity for insertion can be bounded by $\mathcal{O}(L_{\max} \log n)$, 
but it is $\mathcal{O}(n \log n)$ in the worst case because $0 \leq L_{\max} \leq n$ holds. 

\paragraph{Proof of Lemma~\ref{lem:bound_K_LCP_R}(i)}\label{paragraph:bound_K_LCP_R1}
\begin{figure}[t]
\begin{center}
	\includegraphics[width=0.3\textwidth]{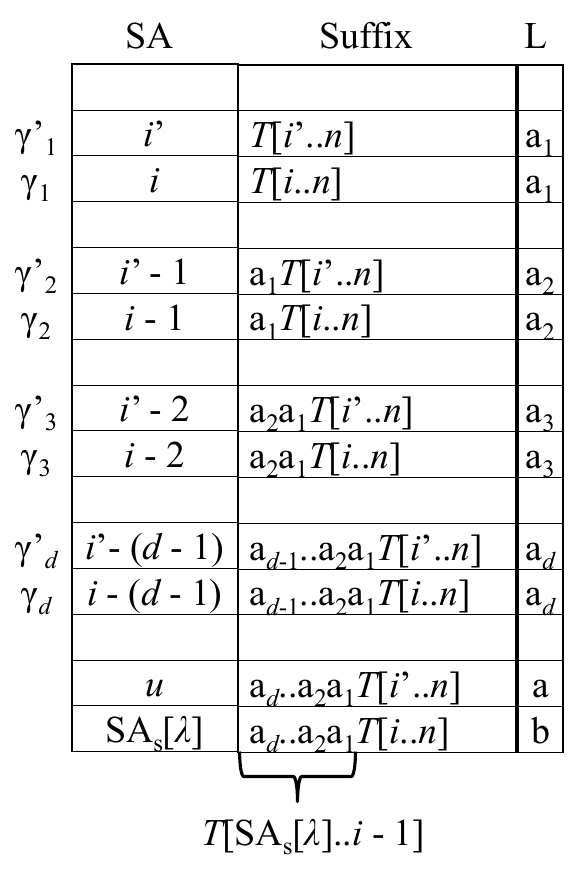}
\caption{Illustration of the relationship between the two suffixes $T[u..n]$ and $T[\SA_{s}[\lambda]..n]$ stated 
in Lemma~\ref{lem:lambda_prop}.}
 \label{fig:lambda_prop}
\end{center}
\end{figure}

$1 \leq \SA_{s}[\lambda] \leq i$ holds because $\SA_{s}$ always contains $1$. 
Therefore, $0 \leq i - \SA_{s}[\lambda] \leq 2$ holds if $i \leq 2$. 
Otherwise, 
we use the following two lemmas to prove Lemma~\ref{lem:bound_K_LCP_R}(i). 

\begin{lemma}\label{lem:lambda_prop}
For $i \neq 1$, 
there exists a position $u$ in $T$ such that 
substring $T[\SA_{s}[\lambda]..i-1]$ is 
a common prefix between two different suffixes $T[u..n]$ and $T[\SA_{s}[\lambda]..n]$.    
\end{lemma}
\begin{proof}
Lemma~\ref{lem:lambda_prop} can be proved by modifying the proof of Lemma 3.5 in \cite{10.1145/3375890}. 
For simplicity, let $d = |T[\SA_{s}[\lambda]..i-1]|$ (i.e., $d = i - \SA_{s}[\lambda]$). 
For each integer $k \in \{ 1, 2, \ldots, d \}$, 
let $\gamma_{k}$ be the position of $(i-k+1)$ in $\SA$. 
For simplicity, 
let $i' = \SA[\gamma_{1}-1]$. 
Similar to the position $\gamma_{k}$, 
let $\gamma'_{k}$ be the position of $(i'-k+1)$ in $\SA$. 
Here, $\gamma'_{1} = \gamma_{1} - 1$ holds. 
By the definition of $\lambda$,
each of the $d$ consecutive integers $(i-d-1)$, $(i-d)$, $\ldots$, $i$ is not 
the starting position of a run in $L$, resulting in 
$L[\gamma_{1}] = L[\gamma_{1}-1]$, 
$L[\gamma_{2}] = L[\gamma_{2}-1]$, $\ldots$,
$L[\gamma_{d}] = L[\gamma_{d}-1]$. 
In this case, 
LF formula ensures that $\LF(\gamma_{k}) = 1 + \LF(\gamma_{k}-1)$ hold, 
and hence, we obtain $\gamma'_{k} = \gamma_{k} - 1$. 
By the definition of BWT, 
$T[i-d..i-1] = L[\gamma_{d}], L[\gamma_{d-1}], \ldots, L[\gamma_{1}]$ 
and 
$T[i-d..i'-1] = L[\gamma'_{d}], L[\gamma'_{d-1}], \ldots, L[\gamma'_{1}]$ hold. 
The two substrings $T[i-d..i-1]$ and $T[i-d..i'-1]$ are the same string (i.e., $T[i-d..i-1] = T[i-d..i'-1]$). 
This is because 
(i) 
$T[i-d..i-1] = L[\gamma_{d}], L[\gamma_{d-1}], \ldots, L[\gamma_{1}]$ 
and 
$T[i-d..i'-1] = L[\gamma'_{d}], L[\gamma'_{d-1}], \ldots, L[\gamma'_{1}]$ follows from the definition of BWT, 
(ii) $L[\gamma'_{d}], L[\gamma'_{d-1}], \ldots, L[\gamma'_{1}] = L[\gamma_{d}-1], L[\gamma_{d-1}-1], \ldots, L[\gamma_{1}-1]$, 
and (iii) $L[\gamma_{1}] = L[\gamma_{1}-1]$, $L[\gamma_{2}] = L[\gamma_{2}-1]$, $\ldots$, $L[\gamma_{d}] = L[\gamma_{d}-1]$. 
$i-d = \SA_{s}[\lambda]$. 
Therefore, we obtain Lemma~\ref{lem:lambda_prop}. 
Figure~\ref{fig:lambda_prop} illustrates the relationship between 
the two suffixes $T[u..n]$ and $T[\SA_{s}[\lambda]..n]$ explained in this proof. 
\end{proof}

\begin{lemma}\label{lem:lcp_prop}
Let $S$ be the longest common prefix between a suffix $T[u..n]$ and any other suffix of $T$, 
and let $p$ be the position of $u$ in the suffix array of $T$.
Then, $|S| \leq \max \{ \LCP[p], \LCP[p+1] \}$. 
\end{lemma}
\begin{proof}
    Lemma~\ref{lem:lcp_prop} follows from the definition of the suffix array and LCP array.
\end{proof}

\begin{figure}[t]
\begin{center}
	\includegraphics[width=0.9\textwidth]{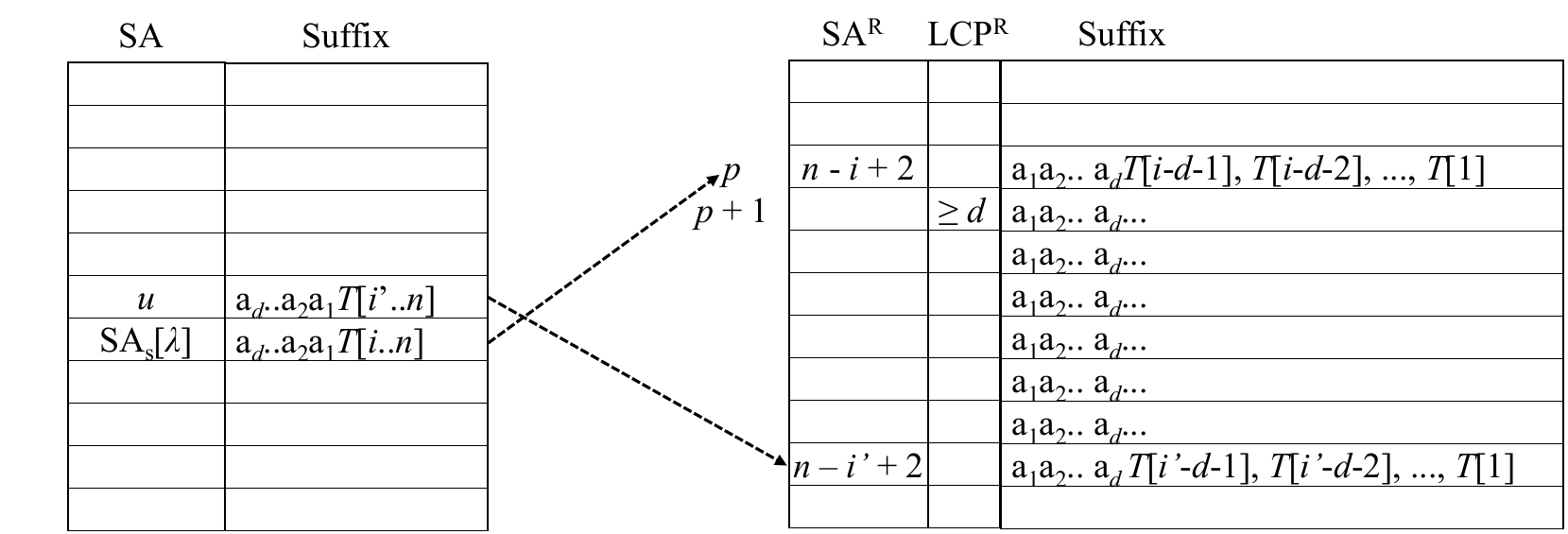}
\caption{Illustration of the relationship between the two suffixes $T^{R}[(n-i+2)..n]$ and $T^{R}[(n-i'+2)..n]$ on $\SA^{R}$. Here, $T[u..n]$ and $T[\SA_{s}[\lambda]..n]$ are the suffixes stated in Lemma~\ref{lem:lambda_prop}; 
$i'$ and $d$ are the integers introduced in the proof of Lemma~\ref{lem:lambda_prop}.}
 \label{fig:bound_K_LCP_R1}
\end{center}
\end{figure}

We prove Lemma~\ref{lem:bound_K_LCP_R}(i) for $i \geq 3$. 
By Lemma~\ref{lem:lambda_prop}, 
there exists a suffix $T[u..n]$ of $T$ such that 
$T[u..i'-1] = T[\SA_{s}[\lambda]..i-1]$ holds, 
where $i' = u + d$ and $d = i - \SA_{s}[\lambda]$. 
On the other hand, 
$T^{R}$ has two suffixes $T^{R}[(n-i+2)..n]$ and $T^{R}[(n-i'+2)..n]$. 
Figure~\ref{fig:bound_K_LCP_R1} illustrates these suffixes on $\SA^{R}$. 
Let $S^{R}$ be the string obtained by reversing $T[\SA_{s}[\lambda]..i-1]$. 
Then, $S^{R}$ is a common prefix of the two suffixes $T^{R}[(n-i+2)..n]$ and $T^{R}[(n-i'+2)..n]$, 
and hence $|S^{R}| \leq \max \{ \LCP^{R}[p], \LCP^{R}[p+1] \}$ holds by Lemma~\ref{lem:lcp_prop}. 
$|S^{R}| = d$, i.e., $|S^{R}| = i - \SA_{s}[\lambda]$. 
Therefore, Lemma~\ref{lem:bound_K_LCP_R}(i) follows from 
$i - \SA_{s}[\lambda] = |S^{R}|$ and $|S^{R}| \leq \max \{ \LCP^{R}[p], \LCP^{R}[p+1] \}$.

\paragraph{Proof of Lemma~\ref{lem:bound_K_LCP_R}(ii)}\label{paragraph:bound_K_LCP_R2}
\begin{figure}[t]
\begin{center}
	\includegraphics[width=0.9\textwidth]{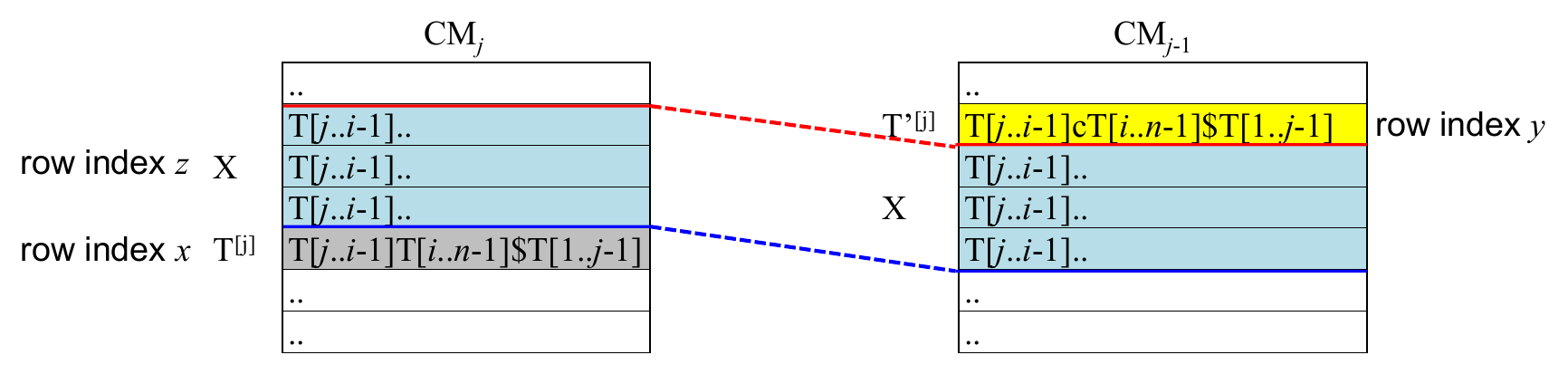}
\caption{Illustration of the relationship among the three circular shifts $T^{[j]}$, $T^{\prime[j]}$, and $X$ used in Lemma~\ref{lem:bound_K_LCP_R}(ii).}
 \label{fig:bound_K_LCP_R2}
\end{center}
\end{figure}

If $i \leq 2$ or $K = 1$ holds, 
then $K - 1 \leq 2$ holds because $K \leq i$ follows from the definition of $K$. 
Otherwise (i.e., $i \geq 3$ and $K \geq 2$), 
let $x$ and $y$ be the row indexes identified in $\CM_{j}$ in Step~(I) and Step~(III) of the update algorithm, respectively, where $j = i-K+1$. 
Since $j \leq i-1$, 
the $x$-th row in $\CM_{j}$ is the circular shift $T^{[j]}$ of $T$, and it is removed from $\CM_{j}$ in Step~(II). After that, the circular shift $T^{\prime [j]}$ of $T'$ is inserted into $\CM_{j}$ as the $y$-th row in Step~(III).  
$x \neq y$ follows from the definition of $K$, 
and hence, 
there exists a row index $z$ of $\CM_{j}$ satisfying either $x < z \leq y$ and $y \leq z < x$.  
Let $X$ be the circular shift at the $z$-th row, 
which is a circular shift of $T$ or $T'$. 
Then, either $T^{[j]} \prec X \prec T^{\prime [j]}$ or 
$T^{\prime [j]} \prec X \prec T^{[j]}$ holds 
since the circular shifts in $\CM_{j}$ are sorted in lexicographic order. 
The substring $T[j..i-1]$ is a common prefix between the two circular shifts $T^{[j]}$ and $T^{\prime [j]}$, 
and hence, this substring is a common prefix among the three circular shifts 
$T^{[j]}$, $X$, and $T^{\prime [j]}$. 
Figure~\ref{fig:bound_K_LCP_R2} illustrates these three circular shifts in $\CM_{j}$ and $\CM_{j-1}$.

Let $S^{R}$ be the string obtained by reversing $T[j..i-1]$. 
Then, the length of $S^{R}$ is $(K-1)$ because $|T[j..i-1]| = K-1$.
If $X$ is a circular shift of $T$, 
then $S^{R}$ must be a common prefix between the $(n-i+2)$-th suffix $T^{R}[n-i+2..n]$ and any other suffix of $T^{R}$, 
and hence  
$|S^{R}| \leq \max \{ \LCP^{R}[p], \LCP^{R}[p+1] \}$ follows from Lemma~\ref{lem:lcp_prop}. 
Therefore, $K - 1 \leq \max \{ \LCP^{R}[p], \LCP^{R}[p+1] \}$ follows from 
$K - 1 \leq |S^{R}|$ and $|S^{R}| \leq \max \{ \LCP^{R}[p], \LCP^{R}[p+1] \}$. 

Otherwise (i.e., $X$ is a circular shift of $T'$), 
the substring $T[j..i-1]$ is a common suffix between the prefix $T'[1..i-1]$ and any other prefix of $T'$. 
Let $u \geq 1$ be the smallest integer such that 
$T[u..i-1]$ is a common suffix between the prefix $T[1..i-1]$ and any other prefix of $T$. 
Then, $T[j..i-1]$ is a suffix of $T[u..i-1]$, resulting in $|T[j..i-1]| \leq |T[u..i-1]|$. 
Similarly, 
let $u' \geq 1$ be the smallest integer such that 
$T'[u'..i-1]$ is a common suffix between the prefix $T'[1..i-1]$ and any other prefix of $T'$. 
The following lemma states the relationship between the two substrings $T[u..i-1]$ and $T'[u'..i-1]$. 

\begin{figure}[t]
\begin{center}
	\includegraphics[width=0.9\textwidth]{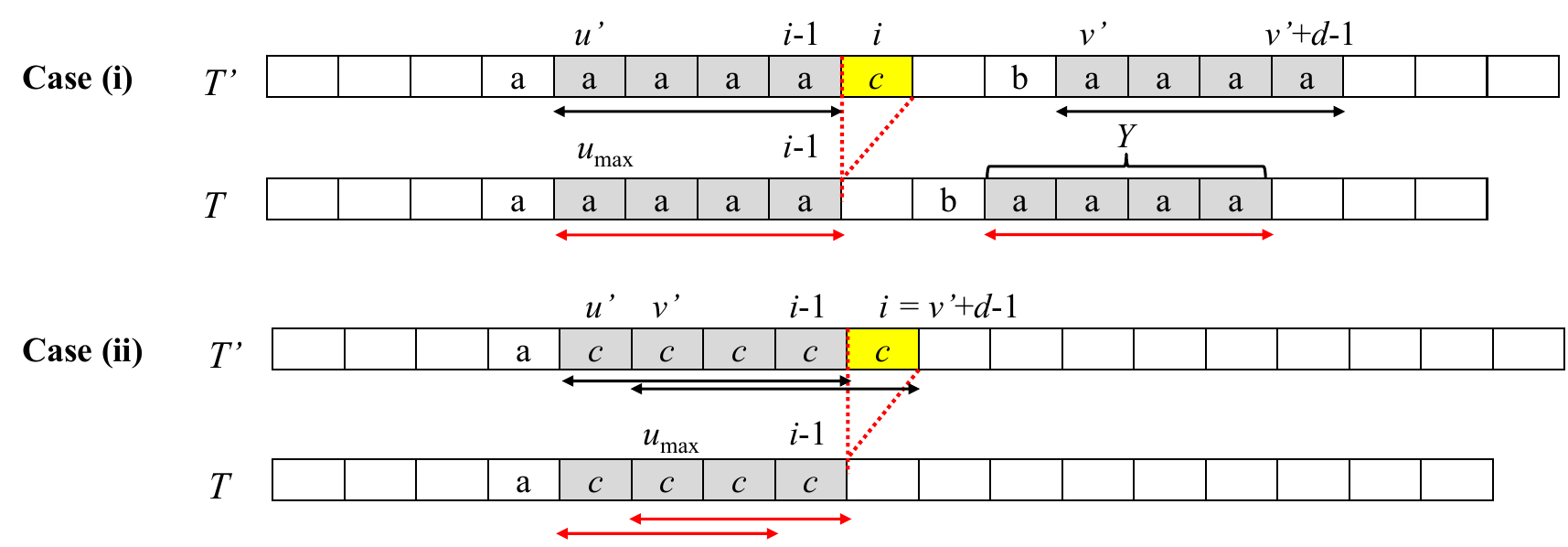}
    \includegraphics[width=0.9\textwidth]{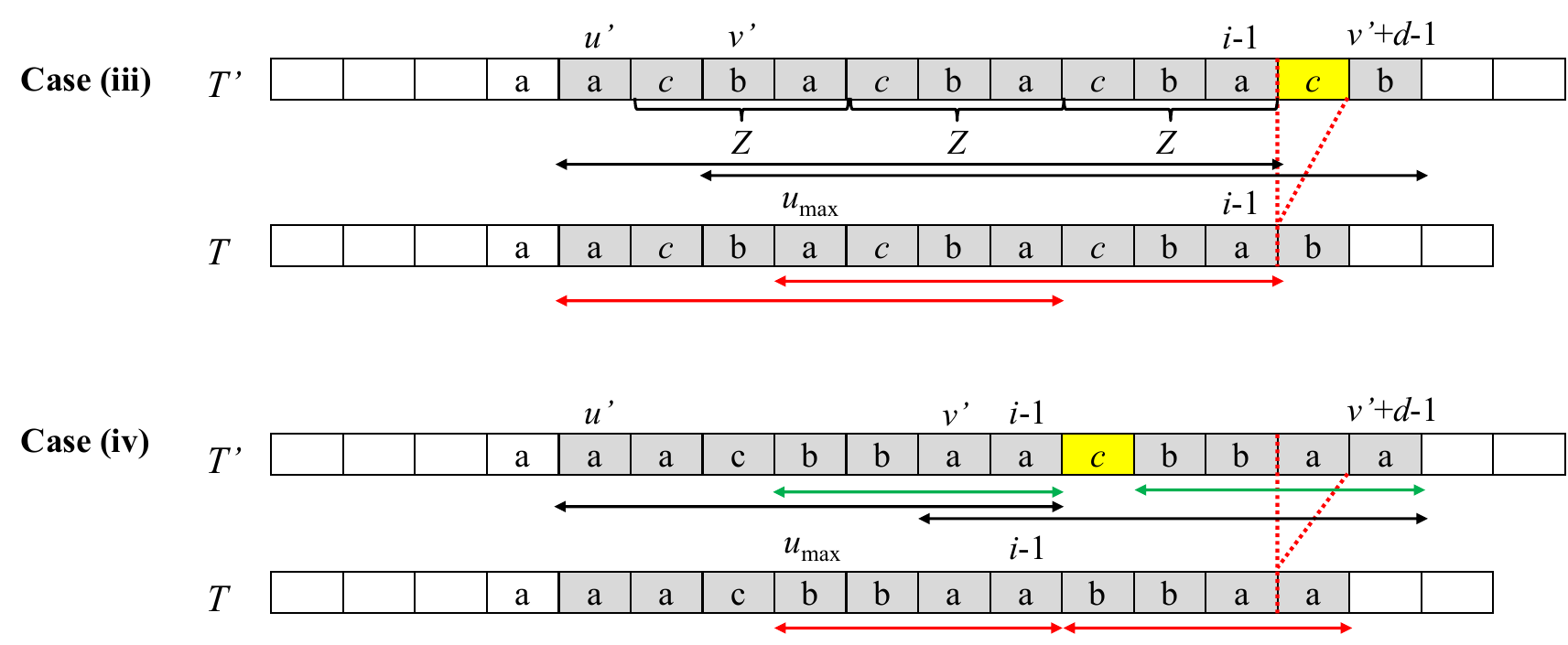}
\caption{Illustration of the relationship between the two integers $u$ and $u'$ stated in the proof of Lemma~\ref{lem:bound_K_char_insertion:sub2}. Here, $u_{\max}$ is the upper bound of $u$ obtained from the relationship the two substrings $T'[u'..i-1]$ and $T'[v'..v' + d -1]$. The black, red, and green arrows represent the three strings 
$T'[u'..i-1]$, $T[u_{\max}..i-1]$, and $T[i+1..v+d-1]$, respectively.}
 \label{fig:bound_K_char_insertion}
\end{center}
\end{figure}

\begin{lemma}\label{lem:bound_K_char_insertion:sub2}
    $|T'[u'..i-1]| \leq 2|T[u..i-1]| + 2$.
\end{lemma}
\begin{proof}

    Let $d$ be the length of $T'[u'..i-1]$ for simplicity (i.e., $d = i-u'$).
    Since $T'[u'..i-1]$ is a common suffix between two prefixes of $T'$, 
    there exists a substring $T'[v'..v' + d-1]$ satisfying $T'[v'..v' + d-1] = T'[u'..i-1]$ and $v' \neq u'$. 
    We consider four cases: 
    (i) $T'[v'..v' + d-1]$ does not contain the $i$-th character of $T'$ (i.e., $i \not \in [v', v' + d-1]$); 
    (ii) $T'[v'..v' + d-1]$ contains the $i$-th character of $T'$ as the last character (i.e., $i = v' + d-1$); 
    (iii) $T'[v'..v' + d-2]$ contains the $i$-th character of $T'$, 
    and $|[u', v'-1]| \leq \frac{d}{2}$;
    (iv) $T'[v'..v' + d-2]$ contains the $i$-th character of $T'$, 
    and $|[u', v'-1]| > \frac{d}{2}$.

    Figure~\ref{fig:bound_K_char_insertion} illustrates the relationship between $u$ and $u'$. 
    For case (i), 
    $T$ contains $T'[v'..v' + d-1]$ as a substring $Y$. 
    Since $T[1..i-1] = T'[1..i-1]$, 
    $Y$ is also a common suffix between the prefix $T[1..i-1]$ and any other prefix of $T$. 
    Therefore, we obtain $|Y| \leq |T[u..i-1]|$ (i.e., $|T'[u'..i-1]| \leq |T[u..i-1]|$).

    For case (ii), 
    the two substrings $T'[u'..i-1]$ and $T'[v'..i]$ are the same string. 
    This fact means that $T'[u'..i-1]$ is a repetition of the character $c$, 
    where $T'[i] = c$. 
    Similarly, $T[u'..i-1]$ is also a repetition of the character $c$ 
    because $T[u'..i-1] = T'[u'..i-1]$. 
    This fact means that 
    $T[t..i-1]$ is a common suffix between the two prefixes $T[1..i-1]$ and $T[1..i-2]$ of $T$ 
    for any $t \in [u'+1, i-1]$, resulting $u \leq u'+1$. 
    Therefore, $|T'[u'..i-1]| \leq |T[u..i-1]| + 1$.
    
    For case (iii), 
    there exists a string $Z$ of length at most $\frac{d}{2}$ 
    such that $T[u'..i-1]$ is a repetition of $Z$, 
    i.e., $T[u'..i-1] = Z'ZZ \cdots Z$ holds, 
    where $Z'$ is a suffix of $Z$ or an empty string. 
    In this case, $T[u'..i-1 - |Z|] = T[u' + |Z|..i-1]$ holds. 
    This fact indicates that $|T'[u'..i-1]| \leq |T[u..i-1]| + |Z|$. 
    $|T[u..i-1]| \geq \frac{d}{2}$ 
    because $1 \leq |Z| \leq \frac{d}{2}$ and $|T'[u'..i-1]| = d$.
    Therefore, $|T'[u'..i-1]| \leq 2|T[u..i-1]|$.

    For case (iv), 
    let $d' = |T'[i+1..v' + d-1]|$. 
    Then, 
    (A) the two substrings $T'[i - d' ..i-1]$ and $T'[i+1..v' + d-1]$ are the same string, 
    and (B) $d' = |[u', v'-1]| - 1$. 
    In this case, $T'[i - d' ..i-1]$ and $T'[i+1..v' + d-1]$ 
    are the same string, 
    and hence $|T[u..i-1]| \geq d'$. 
    $d < 2 d' + 2$ 
    follows from $d' = |[u', v'-1]| - 1$ 
    and $|[u', v'-1]| > \frac{d}{2}$. 

    Therefore, 
\begin{equation*}
    \begin{split}
        |T'[u'..i-1]| &= d \\
        &< 2 d' + 2 \\        
        &\leq 2 |T[u..i-1]| + 2.
    \end{split}
\end{equation*}

\end{proof}
We consider the string $T^{\prime R}$ obtained by reversing $T'$. 
Let $\SA^{\prime R}$ and $\LCP^{\prime R}$ be the suffix and LCP arrays of $T^{\prime R}$, respectively. 
In this case, 
$S^{R}$ must be a common prefix between the $(n-i+3)$-th suffix $T^{\prime R}[n-i+3..n+1]$ and any other suffix of $T^{\prime R}$, 
and hence, $|S^{R}| \leq \max \{ \LCP^{\prime R}[p'], \LCP^{\prime R}[p'+1] \}$ follows from Lemma~\ref{lem:lcp_prop}, 
where $p'$ is the position of $(n-i+3)$ in $\SA^{\prime R}$. 
$\max \{ \LCP^{\prime R}[p'], \LCP^{\prime R}[p'+1] \} = |T'[u'..i-1]|$ 
follows from the definitions of $u'$ and $p'$. 
Similarly, 
$\max \{ \LCP^{R}[p], \LCP^{R}[p+1] \} = |T[u..i-1]|$ 
follows from the definitions of $u$ and $p$. 
We obtain $|T'[u'..i-1]| \leq 2|T[u..i-1]|+2$ by Lemma~\ref{lem:bound_K_char_insertion:sub2}. 
Therefore, 
\begin{equation*}
    \begin{split}
        K - 1 &\leq |S^{R}| \\
        &\leq \max \{ \LCP^{\prime R}[p'], \LCP^{\prime R}[p'+1] \} \\
        &= |T'[u'..i-1]| \\
        &\leq 2 + 2 |T[u..i-1]| \\
        &= 2 + 2 \max \{ \LCP^{R}[p], \LCP^{R}[p+1] \}.        
    \end{split}
\end{equation*}

\subsubsection{Average-case Time Complexity for Character Insertion}
We present an upper bound on the average-case time complexity for insertion using the average $L_{\avg}$ of the values in the LCP array of $T$. 
We generalize the two values $K$ and $\lambda$, which were introduced in Section~\ref{subsec:char_insertion_algorithm},  
to functions $K(i, c)$ and $\lambda(i, c)$, where $c \in \Sigma$ is inserted into $T$ at position $i \in \{ 1, 2, \ldots, n \}$. 
Then, the average-case time complexity is $\mathcal{O}((\sum_{i = 1}^{n} (K(i, c) + i - \SA_{s}[\lambda(i, c)])) \log n / n)$. 
The following lemma ensures that the average-case time complexity can be bounded using $L_{\avg}$. 

\begin{lemma}\label{lem:average_bound}
    $(\sum_{i = 1}^{n} (K(i, c) + i - \SA_{s}[\lambda(i, c)])) / n = \mathcal{O}(1 + L_{\avg})$ 
    for a character $c \in \Sigma$. 
\end{lemma}
\begin{proof}
Let $p_{i}$ is the position of $(n - i + 1)$ in the suffix array of $T^{R}$ for each integer $i \in \{ 1, 2, \ldots, n \}$. 
    Then, Lemma~\ref{lem:bound_K_LCP_R} shows that $K(i, c) + i - \SA_{s}[\lambda(i, c)] \leq \alpha (1 + \max \{ \LCP^{R}[p_{i-1}]$, $\LCP^{R}[p_{i-1}+1] \})$.  
    Here, $\alpha > 0$ is a constant. 
    We obtain $\sum_{i = 1}^{n} (K(i, c) + i - \SA_{s}[\lambda(i, c)]) \leq \alpha(n + 2 \sum_{i=1}^{n} \LCP^{R}[i])$ because 
    $p_{1}, p_{2}, \ldots, p_{n}$ is a permutation of $n$ integers $1, 2, \ldots, n$. 
    Since the values in $\LCP^{R}$ are a permutation of the values in $\LCP$ (Lemma 6.1 in \cite{DBLP:journals/jda/OhlebuschBA14}),     
    the sum of values in $\LCP$ is equal to the sum of values in $\LCP^{R}$. 
    Therefore, 
\begin{equation*}
    \begin{split}
        (\sum_{i = 1}^{n} (K(i, c) + i - \SA_{s}[\lambda(i, c)])) / n &\leq \alpha (n + 2 \sum_{i=1}^{n} \LCP^{R}[i])/n \\
        &= \alpha (n + 2 \sum_{i=1}^{n} \LCP[i])/n \\
        &= \alpha + 2\alpha L_{\avg} \\
        &= \mathcal{O}(L_{\avg})
    \end{split}
\end{equation*}

\end{proof}

Finally, the average-case time complexity for insertion is $\mathcal{O}((1 + L_{\avg}) \log n)$, 
and we obtain Theorem~\ref{theo:query_time_summary}(ii).
This average running time can be enormously small if $L_{\avg} \ll n$. 



\section{String Insertion Algorithm}\label{sec:insertion}
\textbf{Problem setup.} We describe how to update the dynamic r-index when a string $P$ of length $m$ 
is inserted at position $i$ in string $T$, producing $T' = T[1..i-1]PT[i..n]$. 
This string insertion can be achieved by extending the character insertion algorithm described in Section~\ref{subsec:char_insertion_algorithm}, with only minor modifications. 

\textbf{Correctness guarantee.} 
Similar to the character insertion, 
we maintain critical invariants: (1) $L_{\RLE}$ correctly represents the run-length encoding of the current BWT, (2) $\SA_s$ and $\SA_e$ contain suffix array values at run start and end positions respectively, and (3) all three structures remain synchronized. 
Correctness follows from the fact that 
Salson et al. algorithm can be extended to string insertion \cite{DBLP:journals/tcs/SalsonLLM09,DBLP:journals/jda/SalsonLLM10}. 

The following theorem summarizes the results of this section. 
\begin{theorem}\label{theo:string_insertion_theorem}
    The dynamic r-index can be updated in $\mathcal{O}((m + L_{\max}) \log (n+m))$ time for an insertion of a string $P$ of length $m > 0$ into $T$. 
    The average time complexity is $\mathcal{O}( (m + L_{\avg}) \log (n+m))$.
\end{theorem}

\textbf{Algorithmic framework.}
Similar to the character insertion algorithm, 
we simulate the transformation from conceptual matrix $\CM$ of $T$ to $\CM'$ of $T'$ through $(n+m)$ iterations, from $j = n+m$ down to $j = 1$. We maintain the following sequences.  
(A) Conceptual matrix: $\CM_{n+m} = \CM$, $\CM_{n+m-1}$, $\ldots$, $\CM_0 = \CM'$ (not stored explicitly); 
(B) RLBWT: $L_{\RLE, n+m} = L_{\RLE}$, $L_{\RLE, n+m-1}$, $\ldots$, $L_{\RLE, 0} = L'_{\RLE}$; 
(C) Sampled SAs: $\SA_{s, n+m} = \SA_s$, $\ldots$, $\SA_{s, 0} = \SA'_s$ and $\SA_{e, n+m} = \SA_e$, $\ldots$, $\SA_{e, 0} = \SA'_e$. 

Each iteration $j$ transforms the state from step $j$ to step $j-1$ using a three-step process that maintains consistency across all data structures. 
At iteration $j$, insert the $j$-th circular shift $T^{\prime[j]}$ of $T'$ into $\CM_{j}$ in lexicographic order, 
and remove a circular shift of $T$ corresponding to $T^{\prime[j]}$ from $\CM_{j}$, resulting in $\CM_{j-1}$. 
Here, each circular shift $T^{[k]}$ of the original string $T$ corresponds to  
the $(k+m)$-th circular shift $T^{\prime[k+m]}$ of $T'$ if $k \geq i$; otherwise it corresponds to $T^{\prime[k]}$ of $T'$. 
The details of the three steps are as follows.

\textbf{Step (I): Identification of circular shifts to be removed.}
In this step, we identify the circular shift of $T$ to be removed from $\CM_{j}$. 

\emph{Algorithm:} 
This step is performed only when $j \leq i-1$ or $j \geq i + m$ (i.e., when some circular shift of $T$ is mapped to $T^{\prime[j]}$). 
The row index $x$ of the circular shift to be removed from $\CM_{j}$ is computed using the formula used in Step~(I) for character insertion, but this formula is modified as follows. 

\begin{enumerate}
\item \textbf{Initial iteration} ($j = n + m$): $x = 1$ (always remove the lexicographically smallest circular shift); 
\item \textbf{Boundary iteration} ($j = i - 1$): $x = p$, where $p$ is the row index of the circular shift $T^{[i-1]}$ in $\CM_{j}$, and this row index is computed by auxiliary operation $\compISA_{\ins, 1}(L_{\RLE, j}, \SA_{s, j}, \SA_{e, j})$, which is described in detail later; 
\item \textbf{The other iterations}: $x = \lexCount(L_{\RLE, j+1}, L_{j+1}[x_{+1}])$ $+ \rank(L_{\RLE, j+1}, x_{+1}$, $L_{j+1}[x_{+1}])$, 
where $x_{+1}$ denotes the row index of $\CM_{j+1}$ identified in the previous iteration.
\end{enumerate}

\textbf{Step (II): Deletion from data structures.} 
Remove the identified circular shift from $\CM_j$, and update all three data structures accordingly to maintain consistency. 
This step is identical to Step (II) of the update procedure for character insertion, 
although this step is performed only when $j \leq i-1$ or $j \geq i + m$.

\textbf{Step (III): Insertion into data structures.}
Insert the new circular shift $T'^{[j]}$ into $\CM_j$ at an appropriate row index $y$ and update all data structures accordingly. 
This step is identical to Step (III) of the update procedure for character insertion, 
but the formula for computing $y$ is modified as follows. 

\begin{enumerate}
\item \textbf{Early iterations} ($j \in \{n+1, \ldots, i+m\}$): $y = x$ (insert at the same position where we removed); 
\item \textbf{Special iteration} ($j \in \{ i+m-1, i+m-2, \ldots, i \}$): 
$y = \lexCount(L_{\RLE, j}, L_j[y_{+1}]) + \rank(L_{\RLE, j}, y_{+1}, L_j[y_{+1}]) + \epsilon$. 
Here, $y_{+1}$ is the row index for insertion computed in the previous iteration. 
Let $q$ be the row index of the circular shift $T^{\prime [i+m]}$ in $\CM_{j}$. 
Then, $\epsilon$ is defined as follows: 
$\epsilon = 1$ if (A) $T[i-1] < L_j[y_{+1}]$ or (B) $T[i-1] = L_j[y_{+1}]$ and $q \leq y_{+1}$. 
Otherwise $\epsilon = 0$. 
Similar to the row index $p$, 
this row index $q$ is computed by auxiliary operation $\compISA_{\ins, 2}(L_{\RLE, j}, \SA_{s, j}, \SA_{e, j})$, which is described in detail later; 
\item \textbf{Later iterations} ($j \in \{ i-1, i-2, \ldots, 1 \}$): 
$y = \lexCount(L_{\RLE, j}, L_j[y_{+1}]) + \rank(L_{\RLE, j}$, $y_{+1}, L_j[y_{+1}])$. 
\end{enumerate}

\textbf{Iteration completion.}
Through these three steps, each iteration transforms:
$(\CM_{j}, L_{\RLE, j}$, $\SA_{s, j}, \SA_{e, j}) \rightarrow (\CM_{j-1}, L_{\RLE, j-1}, \SA_{s, j-1}, \SA_{e, j-1})$. 
The algorithm uses only the queries and update operations supported by $\mathscr{D}_{\RLE}(L_{\RLE})$, $\mathscr{D}_{\DI}(\SA_s)$, and $\mathscr{D}_{\DI}(\SA_e)$. 

\textbf{Critical optimization: Iteration skipping.}
Similar to character insertion, 
the first $(n-i)$ iterations ($j \in \{n+m, n+m-1, \ldots, i+m+1 \}$) can be replaced with $\mathcal{O}(1)$ operations, 
and the last $(i-K_{\ins}-1)$ iterations ($j \in \{i-K_{\ins}-1, i-K_{\ins}-2, \ldots, 1 \}$) can be skipped. 
Here, $K_{\ins}$ is the integer $K$ defined using the string insertion algorithm. 
Formally, let $K_{\ins}$ be the smallest integer such that $x = y$ in some iteration $j \leq i - K_{\ins}$. If no such $K_{\ins}$ exists, define $K_{\ins} := i$. 
This is because this optimization is based on Lemma~\ref{lem:state_stability}, 
and this lemma can be extended to string insertion as follows: 
\begin{lemma}\label{lem:state_stability_for_string_insertion}
    The following three statements hold for string insertion: 
    \begin{itemize}
        \item[(i)] $L_{n+m} = L_{n+m-1} = \cdots = L_{i+m}$;
        \item[(ii)] For each $j \in \{ n+m, n+m-1, \ldots, i + m - 1 \}$, 
        $\SA_{j}[t] = \SA[t] + m$ if $\SA[t] \geq j$; otherwise, $\SA_{j}[t] = \SA[t]$, 
        where $t$ is a position of $\SA_{j}$; 
        \item[(iii)] $L_{i-K-1} = L_{i-K-2} = \cdots = L_{0}$ and $\SA_{i-K-1} = \SA_{i-K-2} = \cdots = \SA_{0}$. 
    \end{itemize}
\end{lemma}
\begin{proof}
    Lemma~\ref{lem:state_stability_for_string_insertion} can be proved by appropriately 
    modifying the proof of Lemma~\ref{lem:state_stability}. 
\end{proof}

Therefore, the total iterations reduce from $(n+m)$ to $(1 + m + K_{\ins})$, specifically $j \in \{ i+m, i+m-1, \ldots, i-K_{\ins} \}$. 

\subsection{Supporting \texorpdfstring{$\compISA_{\ins, 1}$}{computeISA1} and 
\texorpdfstring{$\compISA_{\ins, 2}$}{computeISA2}}\label{subsec:isa_12_for_string_insertion}
This section describes the two operations $\compISA_{\ins, 1}$ and $\compISA_{\ins, 2}$, 
which are used in Step~(I) and Step~(III), respectively. 

\textbf{Supporting $\compISA_{\ins, 1}$.}
The operation $\compISA_{\ins, 1}(L_{\RLE, i-1}, \SA_{s, i-1}, \SA_{e, i-1})$ computes 
the row index $p$ of the circular shift $T^{[i-1]}$ in $\CM_{i-1}$. 
Let $p_{j}$ be the row index of the circular shift $T^{[i-1]}$ in $\CM_{j}$ for each $j \in \{ n+m, n+m-1, \ldots, i-1 \}$. 
Then, the update procedure indicates that 
the following relationship among holds $p_{n+m}$, $p_{n+m-1}$, $\ldots$, $p_{i-1}$. 
\begin{enumerate}
    \item $p_{n+m} = \ISA[i-1]$.
    \item $p_{n+m} = p_{n+m-1} = \cdots = p_{i+m-1}$.
    \item $p_{j-1} = p_{j} + \epsilon$ for each $j \in \{ i+m-1, i+m-2, \ldots, i \}$. 
    Here, $\epsilon = 1$ if a new row is inserted into $\CM_{j}$ in Step~(III) before the $p_{j}$-th row; 
    otherwise, $\epsilon = 0$. 
\end{enumerate}
$p_{i-1} = p$ holds, 
and $\ISA[i-1]$ can be precomputed by the operation $\compISA_{1}(L_{\RLE}, \SA_{s}, i)$ introduced in Section~\ref{subsec:support_ISA}. 
Therefore, $p$ can be computed using the relationship described above before updating $\CM_{i-1}$. 
The time complexity of $\compISA_{\ins, 1}$ is $\mathcal{O}((m + i - \SA_s[\lambda]) \log r)$ time, 
where $\lambda$ is the integer introduced in Section~\ref{subsec:char_insertion_algorithm}. 

\textbf{Supporting $\compISA_{\ins, 2}$.}
The operation $\compISA_{\ins, 2}(L_{\RLE, j}, \SA_{s, j}, \SA_{e, j})$ computes 
the row index $q$ of the circular shift $T^{\prime [i+m]}$ in $\CM_{j}$ 
for each $j \in \{ i+m-1, i+m-2, \ldots, i \}$. 
This row index $q$ can be computed using a similar approach to that of the operation $\compISA_{\ins, 1}$. 
Therefore, $m$ operations 
$\compISA_{\ins, 2}(L_{\RLE, i+m-1}, \SA_{s, i+m-1}, \SA_{e, i+m-1})$,
$\compISA_{\ins, 2}(L_{\RLE, i+m-2}, \SA_{s, i+m-2}, \SA_{e, i+m-2})$, $\ldots$, 
$\compISA_{\ins, 2}(L_{\RLE, i}, \SA_{s, i}, \SA_{e, i})$ can be computed in $\mathcal{O}((m + i - \SA_s[\lambda]) \log r)$ time. 

\subsection{Modifying \texorpdfstring{$\compSA_{X, 1}$}{computeSAX}, \texorpdfstring{$\compSA_{X, 2}$}{computeSAX}, \texorpdfstring{$\compSA_{Y, 1}$}{computeSAY}, and \texorpdfstring{$\compSA_{Y, 2}$}{computeSAY}}
Similar to character insertion, 
the string insertion algorithm uses the four auxiliary operations $\compSA_{X, 1}$, $\compSA_{X, 2}$, $\compSA_{Y, 1}$, and $\compSA_{Y, 2}$, which were introduced in Section~\ref{subsec:support_SAXY}, in Step~(II) and Step~(III). 
For correctly updating the dynamic r-index, 
we need to modify the algorithms of these four operations. 
This is because these operations compute suffix array values using the dynamic LF formula (i.e., Lemma~\ref{lem:dynamic_LF_formula}) and Lemma~\ref{lem:state_stability}, and 
these lemmas are extended to string insertion. 

The following lemma is the dynamic LF formula for string insertion. 
\begin{lemma}\label{lem:dynamic_LF_formula_for_string_insertion}
Let $j \in \{n+m, n+m-1, \ldots, 1\}$ and $t$ be a position in $\SA_j$.  
Define $\kappa(t) = 1$ if either  
(i) $T[i-1] < L_j[t]$, or  
(ii) $T[i-1] = L_j[t]$ and $q \leq t$;  
otherwise, $\kappa(t) = 0$. 
Here, $q$ is the row index of the circular shift $T^{\prime [i+m]}$ in $\CM_{j}$, which is used in Step~(III). 
Then, the following equation holds: 
\begin{equation*}
  \LF_{j}(t) =   \begin{cases}
    \lexCount(L_{\RLE, j}, L_{j}[t]) + \rank(L_{\RLE, j}, t, L_{j}[t]) + \kappa(t) & \text{if $i \leq j \leq i+m-1$} \\
    \lexCount(L_{\RLE, j}, L_{j}[t]) + \rank(L_{\RLE, j}, t, L_{j}[t]) & \text{otherwise.}
  \end{cases}
\end{equation*}
\end{lemma}
\begin{proof}
    Lemma~\ref{lem:dynamic_LF_formula_for_string_insertion} can be proved by appropriately 
    modifying the proof of Lemma~\ref{lem:dynamic_LF_formula}. 
\end{proof}

We use Lemma~\ref{lem:dynamic_LF_formula_for_string_insertion} and Lemma~\ref{lem:state_stability_for_string_insertion} instead of 
Lemma~\ref{lem:dynamic_LF_formula} and Lemma~\ref{lem:state_stability}, respectively. 
Then, the algorithms for computing suffix array values are modified as follows. 

\textbf{Modifying $\compSA_{X, 1}$.}
The operation $\compSA_{X, 1}(L_{\RLE, j}, \SA_{s, j}, \SA_{e, j}, x_{-1})$ computes $\SA_{j-1}[x_{-1}+1]$  
from $L_{\RLE, j}$, $\SA_{s, j}$, $\SA_{e, j}$, and the row index $x_{-1}$ of a circular shift in $\CM_{j-1}$,  
identified in Step~(I) of the update algorithm. 
Here, $j \in \{i+m+1, i, i-1, \ldots, i-K_{\ins}+1 \}$ holds 
since most of the $(n+1)$ iterations are skipped. 

Let $y$ be the insertion position in $\CM_j$ determined in Step~(III). 
Then, the value $\SA_{j-1}[x_{-1}+1]$ is computed based on the following cases:  
(i) $j = i+m+1$;  
(ii) $j = i$;  
(iii) $j < i$. 

\textbf{Case (i): $j = i+m+1$.}
$\SA_{j-1}[x_{-1}+1]$ is computed as follows: 
\begin{align*}
    \SA_{j-1}[x_{-1}+1] &= 
  \begin{cases}
    \SA[\ISA[i+1]+1] + m & \text{if $\SA[\ISA[i+1]+1] \geq i+m$} \\
    \SA[\ISA[i+1]+1] & \text{otherwise.}
  \end{cases}
\end{align*}
Here, this equation follows from Lemma~\ref{lem:state_stability_for_string_insertion}, 
and $\SA[\ISA[i+1]+1]$ is precomputed by $\phi^{-1}(i+1)$. 

\textbf{Case (ii): $j = i$.}
We leverage the $(m+1)$ positions $p_{i+m-1}, p_{i+m-2}, \ldots, p_{i-1}$ used to perform the operation $\compISA_{\ins, 1}$. 
We consider $(m+1)$ suffix array values $\SA_{i+m-1}[p_{i+m-1}+1]$, $\SA_{i+m-2}[p_{i+m-2}+1]$, $\ldots$, $\SA_{i-1}[p_{i-1}+1]$. 
Then, we can prove the following relationship among these values using the relationship among $p_{i+m-1}, p_{i+m-2}, \ldots, p_{i-1}$, which is explained in Section~\ref{subsec:isa_12_for_string_insertion}. 
\begin{enumerate}
    \item $\SA_{i+m-1}[p_{i+m-1}+1] = \phi^{-1}(i-1) + m$ if $\phi^{-1}(i-1) \geq i$; 
    otherwise $\SA_{i+m-1}[p_{i+m-1}+1] = \phi^{-1}(i-1)$.
    \item For each $j' \in \{ i+m-1, i+m-2, \ldots, i \}$, 
    $\SA_{j'-1}[p_{j'-1}+1] = j'$ if a new row is inserted into $\CM_{j'}$ as the $(p_{j'}+1)$-th row in Step~(III); 
    otherwise, $\SA_{j'-1}[p_{j'-1}+1] = \SA_{j'}[p_{j'}+1]$. 
    \item $\SA_{i-1}[p_{i-1}+1] = \SA_{j-1}[x_{-1}+1]$ 
    because $j = i$ and $x_{-1} = p_{i-1}$. 
\end{enumerate}
This relationship indicates that $\SA_{j-1}[x_{-1}+1]$ can be computed using $\phi^{-1}(i-1)$ 
and the the $(m+1)$ positions $p_{i+m-1}, p_{i+m-2}, \ldots, p_{i-1}$. 
Therefore, $\SA_{j-1}[x_{-1}+1]$ can be computed in $\mathcal{O}(\log r)$ time in this case. 

\textbf{Case (iii): $j < i$.}
We follow the same approach as the fourth case in the original algorithm of $\compSA_{X,1}$, 
i.e., we compute $\SA_{j-1}[x_{-1}+1]$ in $\mathcal{O}(\log n)$ time using $\SA_{j}[h]$ and Lemma~\ref{lem:sah_lemma}, 
where $h$ is the position introduced in Section~\ref{subsec:support_SAXY}. 
Lemma~\ref{lem:sah_lemma} still holds for string insertion because 
this lemma is proven using the dynamic LF formula, and this formula is appropriately modified by Lemma~\ref{lem:dynamic_LF_formula_for_string_insertion}. 

Finally, 
the time complexity of $\compSA_{X, 1}$ can be bounded by $\mathcal{O}(\log n)$. 

\textbf{Modifying other operations.}
The algorithms of $\compSA_{X, 2}$, $\compSA_{Y, 1}$, and $\compSA_{Y, 2}$ 
can be modified using a similar approach to $\compSA_{X, 1}$. 
Similar to $\compSA_{X, 1}$, 
the time complexity of these modified algorithms can be bounded by $\mathcal{O}(\log n)$. 

\subsection{Overall Time Complexity}
Similar to character insertion, 
$L_{\RLE}$, $\SA_{s}$, and $\SA_{e}$ are updated by executing $\mathcal{O}(m + K_{\ins})$ iterations. 
Each iteration $j \in \{ i+m, i+m-1, \ldots, i-K_{\ins} \}$ can be executed in $\mathcal{O}(1)$ queries and update operations supported by 
$\mathscr{D}_{\RLE}(L_{\RLE, j})$, $\mathscr{D}_{\DI}(\SA_{s, j})$, and $\mathscr{D}_{\DI}(\SA_{e, j})$. 
Thus, the iteration $j$ can be executed in $\mathcal{O}(\log r')$, 
where $r'$ is the number of runs in $L_{\RLE, j}$, and $r' \leq n+m$ holds. 
In addition, we use two operations $\compISA_{\ins, 1}$ and $\compISA_{\ins, 2}$, 
which take $\mathcal{O}((m + i - \SA_s[\lambda]) \log r)$ time. 
Therefore, this update algorithm requires $\mathcal{O}((m + K_{\ins} + i - \SA_{s}[\lambda])\log (n+m))$ time in total. 

\subsection{Worst-case Time Complexity}
When inserting a string $P$ of length $m$ into $T$ at position $i$, 
the insertion takes $\mathcal{O}((m + K_{\ins} + i - \SA_s[\lambda]) \log (n+m))$ time, 
where $(1 + m + K_{\ins})$ iterations are needed for the main update algorithm, 
and $(m + i - \SA_s[\lambda])$ steps are required for auxiliary operations. 
Similar to character insertion, 
the time complexity for insertion can be bounded using $L_{\max}$. 
This is because the following lemma holds instead of Lemma~\ref{lem:bound_K_LCP_R}. 

\begin{lemma}\label{lem:bound_K_string_insertion}
    There exists a constant $\alpha$ satisfying two conditions: 
    \begin{enumerate}[(i)]
        \item $i - \SA_{s}[\lambda] \leq \alpha (1 + \max \{ \LCP^{R}[p], \LCP^{R}[p+1] \})$;
        \item $K_{\ins} - 1 \leq \alpha (m + \max \{ \LCP^{R}[p], \LCP^{R}[p+1] \})$.    
   \end{enumerate}
   Here, $p$ is the position of $(n - i + 1)$ in the suffix array $\SA^{R}$ of $T^{R}$. 
\end{lemma}
\begin{proof}
Lemma~\ref{lem:bound_K_string_insertion} can be proved using the same approach 
as that used in the proof of Lemma~\ref{lem:bound_K_LCP_R}.
\end{proof}

Lemma~\ref{lem:bound_K_string_insertion} indicates that 
the time complexity for insertion is $\mathcal{O}(\max \{m, \LCP^{R}[p], \LCP^{R}[p+1] \} \log (n+m))$. 
Similar to character insertion, 
there exists a position $p'$ in $\LCP$ such that 
the time complexity for insertion is $\mathcal{O}((m + \LCP[p']) \log (n+m))$ time, 
where $\LCP[p'] \leq L_{\max}$. 
Therefore, the time complexity for insertion can be bounded by $\mathcal{O}((m + L_{\max}) \log (n+m))$.

\subsection{Average-case Time Complexity}
We present an upper bound on the average-case time complexity for insertion using the average $L_{\avg}$ of the values in the LCP array of $T$. 
Similar to character insertion, 
we generalize the two values $K_{\ins}$ and $\lambda$ to 
functions $K_{\ins}(i, P)$ and $\lambda_{\ins}(i, P)$, where $P \in \Sigma$ is inserted into $T$ at position $i \in \{ 1, 2, \ldots, n \}$. 

When a given string $P$ is inserted into $T$, 
the average-case time complexity for insertion is 
$\mathcal{O}(\log (n+m) \sum_{i = 1}^{n} (m + K_{\ins}(i, P) + i - \SA_{s}[\lambda_{\ins}(i, P)]) / n)$. 
Here, 
$\sum_{i = 1}^{n} (m + K_{\ins}(i, P) + i - \SA_{s}[\lambda_{\ins}(i, P)]) \leq \alpha \sum_{i=1}^{n} (m + \max \{ \LCP^{R}[p_{i-1}], \LCP^{R}[p_{i-1}+1] \})$ follows from Lemma~\ref{lem:bound_K_string_insertion}, 
where $\alpha$ is a constant, and $p_{i-1}$ is the position of $(n - i + 2)$ in the suffix array of $T^{R}$. 
Similar to Theorem~\ref{theo:query_time_summary}(ii), 
$\alpha \sum_{i=1}^{n} (m + \max \{ \LCP^{R}[p_{i-1}], \LCP^{R}[p_{i-1}+1] \})$ 
is at most $2 \alpha \sum_{i=1}^{n} (m + \LCP[i])$. 
Therefore, 
we obtain the average-case time complexity for insertion by the following equation:
\begin{equation*}
    \begin{split}
        \mathcal{O}(\log (n+m) \sum_{i = 1}^{n} (m + K_{\ins}(i, P) + i - \SA_{s}[\lambda_{\ins}(i, P)]) / n) &=  \mathcal{O}(\frac{nm + \sum_{i=1}^{n} \LCP[i]}{n} \log (n+m)) \\
        &= \mathcal{O}((m + L_{\avg}) \log (n+m)).
    \end{split}
\end{equation*}

Finally, we obtain Theorem~\ref{theo:string_insertion_theorem}.

\section{String Deletion Algorithm}\label{sec:deletion}
\textbf{Problem setup.} We describe how to update the dynamic r-index when 
a substring $T[i..i+m-1]$ of length $m$ is deleted from $T$ at a position $i \in \{ 1, 2, \ldots, n - m \}$, 
producing $T'=T[1..(i-1)]T[(i+m)..n]$. 
Similar to the string insertion, 
this string deletion can be achieved by extending the character insertion algorithm described in Section~\ref{subsec:char_insertion_algorithm}, with only minor modifications. 

\textbf{Correctness guarantee.} 
Similar to the character insertion, 
we maintain critical invariants: (1) $L_{\RLE}$ correctly represents the run-length encoding of the current BWT, (2) $\SA_s$ and $\SA_e$ contain suffix array values at run start and end positions respectively, and (3) all three structures remain synchronized. 
Correctness follows from the fact that 
Salson et al. algorithm can be extended to string deletion \cite{DBLP:journals/tcs/SalsonLLM09,DBLP:journals/jda/SalsonLLM10}. 

The following theorem summarizes the results of this section. 
\begin{theorem}\label{theo:string_deletion_theorem}
    The dynamic r-index can be updated in $\mathcal{O}((m + L_{\max}) \log n)$ time for a deletion of 
    a substring $T[i..i+m-1]$ from $T$. 
    The average time complexity is $\mathcal{O}((m + L_{\avg}) \log n)$.
\end{theorem}

\textbf{Algorithmic framework.}
Similar to the string insertion algorithm, 
we simulate the transformation from conceptual matrix $\CM$ of $T$ to $\CM'$ of $T'$ through $n$ iterations, from $j = n$ down to $j = 1$. We maintain the following sequences.  
(A) Conceptual matrix: $\CM_{n} = \CM$, $\CM_{n-1}$, $\ldots$, $\CM_0 = \CM'$ (not stored explicitly); 
(B) RLBWT: $L_{\RLE, n} = L_{\RLE}$, $L_{\RLE, n-1}$, $\ldots$, $L_{\RLE, 0} = L'_{\RLE}$; 
(C) Sampled SAs: $\SA_{s, n} = \SA_s$, $\ldots$, $\SA_{s, 0} = \SA'_s$ and $\SA_{e, n-1} = \SA_e$, $\ldots$, $\SA_{e, 0} = \SA'_e$. 

Each iteration $j$ transforms the state from step $j$ to step $j-1$ using a three-step process that maintains consistency across all data structures. 
At iteration $j$, remove the $j$-th circular shift $T^{[j]}$ of $T$ from $\CM_{j}$, 
and insert the circular shift of $T'$ corresponding to $T^{[j]}$ into $\CM_{j}$ in lexicographic order, 
resulting in $\CM_{j-1}$. 
Here, each circular shift $T^{\prime[k]}$ of $T'$ corresponds to  
the $(k+m)$-th circular shift $T^{[k+m]}$ of $T$ if $k \geq i$; otherwise it corresponds to $T^{[k]}$ of $T$. 
The details of the three steps are as follows. 

\textbf{Step (I): Identification of circular shifts to be removed.}
In this step, we identify the circular shift of $T$ to be removed from $\CM_{j}$. 

\emph{Algorithm:} 
The row index $x$ of the circular shift $T^{[j]}$ in $\CM_{j}$ is computed using the formula used in Step~(I) for character insertion, but this formula is modified as follows. 

\begin{enumerate}
\item \textbf{Initial iteration} ($j = n$): $x = 1$ (always remove the lexicographically smallest circular shift); 
\item \textbf{Special iteration} ($j \in \{ i+m-2, i+m-3, \ldots, i-1 \}$): 
$x = \lexCount(L_{\RLE, j+1}, L_{j+1}[x_{+1}]) + \rank(L_{\RLE, j+1}, x_{+1}, L_{j+1}[x_{+1}]) - \epsilon$. 
Here, $x_{+1}$ denotes the row index of $\CM_{j+1}$ identified in the previous iteration, 
and let $p$ be the row index of the $i$-th circular shift $T^{[i]}$ of $T$ in $\CM_{j+1}$, 
Then, $\epsilon = 1$ if either (A) $L_{j+1}[p] < L_{j+1}[x_{+1}]$ or 
(B) $L_{j+1}[p] = L_{j+1}[x_{+1}]$ and $p \leq x_{+1}$. 
Otherwise, $\epsilon = 0$. 
The row index $p$ is computed by auxiliary operation $\compISA_{\del, 1}(L_{\RLE, j+1}, \SA_{s, j+1}, \SA_{e, j+1})$, which is described in detail later;
\item \textbf{The other iterations}: $x = \lexCount(L_{\RLE, j+1}, L_{j+1}[x_{+1}]) + \rank(L_{\RLE, j+1}, x_{+1}, L_{j+1}[x_{+1}])$
\end{enumerate}

\textbf{Step (II): Deletion from data structures.} 
Remove the identified circular shift $T^{[j]}$ from $\CM_j$, and update all three data structures accordingly to maintain consistency. 
This step is identical to Step (II) of the update procedure for character insertion, 
although this step is not skipped. 

\textbf{Step (III): Insertion into data structures.}
Identify the circular shift of $T'$ corresponding to $T^{[j]}$, 
insert the identified circular shift into $\CM_j$ at an appropriate row index $y$, 
and update all data structures accordingly. 
This step is identical to Step (III) of the update procedure for character insertion, 
but this step is performed only when $j \leq i-1$ or $j \geq i+m$ (i.e., when some circular shift of $T'$ corresponds to $T^{[j]}$), 
and the formula for computing $y$ is modified as follows. 

\begin{enumerate}
\item \textbf{Early iterations} ($j \in \{n, n-1, \ldots, i+m\}$): $y = x$ (insert at the same position where we removed); 
\item \textbf{Boundary iteration} ($j = i - 1$): $y = \lexCount(L_{\RLE, j}, L_{j}[q]) + \rank(L_{\RLE, j}, q, L_{j}[q])$, 
where $q$ is the row index of the $i$-th circular shift $T^{\prime[i]}$ of $T'$ in $\CM_{j}$, 
and this row index is computed by auxiliary operation $\compISA_{\del, 2}(L_{\RLE, j}, \SA_{s, j}, \SA_{e, j})$, 
which is described in detail later;
\item \textbf{Later iterations} ($j \in \{ i-2, i-3, \ldots, 1 \}$): 
$y = \lexCount(L_{\RLE, j}, L_{j}[y_{+1}]) + \rank(L_{\RLE, j}$, $y_{+1}, L_{j}[y_{+1}])$, 
where $y_{+1}$ is the row index in $\CM_{j+1}$ identified in Step~(III). 
\end{enumerate}

\textbf{Iteration completion.}
Through these three steps, each iteration transforms:
$(\CM_{j}, L_{\RLE, j}$, $\SA_{s, j}, \SA_{e, j}) \rightarrow (\CM_{j-1}, L_{\RLE, j-1}, \SA_{s, j-1}, \SA_{e, j-1})$. 
The algorithm uses only the queries and update operations supported by $\mathscr{D}_{\RLE}(L_{\RLE})$, $\mathscr{D}_{\DI}(\SA_s)$, and $\mathscr{D}_{\DI}(\SA_e)$. 

\textbf{Value duplication.} 
Let $\SA_{a}$ be $\SA_{s}$ or $\SA_{e}$. 
Then, $\SA_{a}$ may contain the same value twice by performing Step~(III). 
This is because the third step inserts a circular shift $T^{\prime[k]}$ of $T'$ into $\CM$ 
before removing the $k$-th circular shift $T^{[k]}$ of $T$ for each iteration $j \geq i+m$. 
To avoid this duplication of values, 
we increment all the values in $\SA_{a}$ by $\Delta := n+1$ (i.e., $\incrementSA(\SA_{s}, 0, \Delta)$) 
before performing the update procedure. 
Then, we can distinguish between the two circular shifts $T^{[k]}$ and $T^{\prime[k]}$ in $\SA_{a}$. 
Specifically, the value $k$ in $\SA_{a}$ represents $T^{[k-(n+1)]}$ if $k \geq \Delta$;
otherwise, it represents $T^{\prime[k]}$. 
Similarly, the suffix array $\SA_{j}$ corresponding to each $\CM_{j}$ is modified as follows: 
for each $t \in \{ 1, 2, \ldots, |\SA_{j}| \}$, 
$\SA_{j}[t] := k + \Delta$ if the $t$-th row of $\CM_{j}$ represents a circular shift $T^{[k]}$ of $T$; 
otherwise, the row represents a circular shift $T^{\prime[k]}$ of $T'$ and $\SA_{j}[t] := k$.
This modification does not pose a critical problem for the update procedure, 
since all circular shifts of $T$ are removed from $\CM$.

\textbf{Critical optimization: Iteration skipping.}
Similar to character insertion, 
the first $(n-i-m)$ iterations ($j \in \{n, n-1, \ldots, i+m+1 \}$) can be replaced with $\mathcal{O}(1)$ operations. 
Similarly, the last $(i-K_{\del}-1)$ iterations ($j \in \{i-K_{\del}-1, i-K_{\del}-2, \ldots, 1 \}$) can be replaced with $\mathcal{O}(1)$ operations. 
Here, $K_{\del}$ is the integer $K$ defined using the string deletion algorithm. 
Formally, let $K_{\del}$ be the smallest integer such that $x = y$ in some iteration $j \leq i - K_{\del}$. If no such $K_{\del}$ exists, define $K_{\del} := i$. 
This is because this optimization is based on Lemma~\ref{lem:state_stability}, 
and this lemma can be modified as follows: 
\begin{lemma}\label{lem:state_stability_for_string_deletion}
    The following four statements hold for string deletion: 
    \begin{itemize}
        \item[(i)] $L_{n} = L_{n-1} = \cdots = L_{i+m}$;
        \item[(ii)] For each $j \in \{ n, n-1, \ldots, i + m - 1 \}$, 
        $\SA_{j}[t] = \SA_{n}[t] - \Delta - m$ if $\SA_{n}[t] \geq \Delta + j$; otherwise, $\SA_{j}[t] = \SA_{n}[t]$, 
        where $t$ is a position of $\SA_{j}$; 
        \item[(iii)] $L_{i-K_{\del}-1} = L_{i-K_{\del}-2} = \cdots = L_{0}$;
        \item[(iv)] For each $j \in \{ i-K_{\del}-1, n-2, \ldots, 0 \}$, 
        $\SA_{j}[t] = \SA_{i-K_{\del}-1}[t] - \Delta$ if $\SA_{i-K_{\del}-1}[t] \geq \Delta + j$; 
        otherwise, $\SA_{j}[t] = \SA_{i-K_{\del}-1}[t]$, where $t$ is a position of $\SA_{j}$. 
    \end{itemize}
\end{lemma}
\begin{proof}
    Lemma~\ref{lem:state_stability_for_string_deletion} can be proved by appropriately 
    modifying the proof of Lemma~\ref{lem:state_stability}. 
\end{proof}

Therefore, the total iterations reduce from $n$ to $(1 + m + K_{\del})$, specifically $j \in \{ i+m, i+m-1, \ldots, i-K_{\del} \}$. 

\subsection{Supporting \texorpdfstring{$\compISA_{\del, 1}$}{computeISA1} and 
\texorpdfstring{$\compISA_{\del, 2}$}{computeISA2}}
This section describes the two operations $\compISA_{\del, 1}$ and $\compISA_{\del, 2}$, 
which are used in Step~(I) and Step~(III), respectively. 
The operation $\compISA_{\del, 1}(L_{\RLE, j+1}, \SA_{s, j+1}, \SA_{e, j+1})$ computes 
the row index $p$ of the $i$-th circular shift $T^{[i]}$ of $T$ in $\CM_{j+1}$, 
where $j \in \{ i+m-2, i+m-3, \ldots, i-1 \}$. 
The operation $\compISA_{\del, 2}(L_{\RLE, j}, \SA_{s, j}, \SA_{e, j})$ computes 
the row index $q$ of the $i$-th circular shift $T^{\prime[i]}$ of $T'$ in $\CM_{j}$, 
where $j = i-1$. 
These row indexes $p$ and $q$ can be computed using the same method as 
the two operations $\compISA_{\ins, 1}$ and $\compISA_{\ins, 2}$ for string insertion. 
Therefore, the time complexity of $\compISA_{\del, 1}$ and $\compISA_{\del, 2}$ 
can be bounded by $\mathcal{O}((m + i - \SA_s[\lambda]) \log r)$ time, 
where $\lambda$ is the integer introduced in Section~\ref{subsec:char_insertion_algorithm}. 

\subsection{Modifying \texorpdfstring{$\compSA_{X, 1}$}{computeSAX}, \texorpdfstring{$\compSA_{X, 2}$}{computeSAX}, \texorpdfstring{$\compSA_{Y, 1}$}{computeSAY}, and \texorpdfstring{$\compSA_{Y, 2}$}{computeSAY}}
Similar to string insertion, 
the four auxiliary operations $\compSA_{X, 1}$, $\compSA_{X, 2}$, $\compSA_{Y, 1}$, and $\compSA_{Y, 2}$ are modified 
because (A) the algorithms of these operations depend on the dynamic LF function and Lemma~\ref{lem:state_stability}, 
and (B) these function and lemma are modified to string deletion. 

For each $j = \{ n+m, n+m-1, \ldots, 1 \}$, 
The dynamic LF function $\LF_{j}(t)$ is modified as follows: 

\begin{enumerate}
\item \textbf{Standard case} ($\SA_{j}[t] \neq \Delta + 1$):
If $\SA_{j-1}$ contains $(\SA_{j}[t] - 1)$ at a position $t'$, then $\LF_{j}(t) := t'$; 
otherwise, $\LF_{j}(t) := -1$. 
\item \textbf{Special case} ($\SA_{j}[t] = \Delta + 1$): 
In this case, $\SA_{j-1}$ always contains $|T'|$ (i.e., $(n - m)$), 
and $\LF_{j}(t)$ returns the position $t'$ of $|T'|$ in $\SA_{j-1}$. 
Here, $\SA_{j}[t]$ represents the first circular shift $T^{[1]}$ of $T$.
\end{enumerate}

Similar to string insertion, 
the dynamic LF formula can be modified as follows: 

\begin{lemma}\label{lem:dynamic_LF_formula_for_string_deletion}
Let $j \in \{n, n-1, \ldots, 1\}$ and $t$ be a position in $\SA_j$.  
Define $\kappa(t) = 1$ if either  
(i) $T[i-1] < L_j[t]$, or  
(ii) $T[i-1] = L_j[t]$ and $p \leq t$;  
otherwise, $\kappa(t) = 0$. 
Here, $p$ is the row index of the circular shift $T^{[i]}$ in $\CM_{j}$. 
Then, the following equation holds: 
\begin{equation*}
  \LF_{j}(t) =   \begin{cases}
    \lexCount(L_{\RLE, j}, L_{j}[t]) + \rank(L_{\RLE, j}, t, L_{j}[t]) - \kappa(t) & \text{if $i \leq j \leq i+m-1$ and $t \neq p$} \\
    -1 & \text{if $i \leq j \leq i+m-1$ and $t = p$} \\
    \lexCount(L_{\RLE, j}, L_{j}[t]) + \rank(L_{\RLE, j}, t, L_{j}[t]) & \text{otherwise.}
  \end{cases}
\end{equation*}
\end{lemma}
\begin{proof}
    Lemma~\ref{lem:dynamic_LF_formula_for_string_deletion} can be proved by appropriately 
    modifying the proof of Lemma~\ref{lem:dynamic_LF_formula}. 
\end{proof}

We use Lemma~\ref{lem:dynamic_LF_formula_for_string_deletion} and Lemma~\ref{lem:state_stability_for_string_deletion} instead of 
Lemma~\ref{lem:dynamic_LF_formula} and Lemma~\ref{lem:state_stability}, respectively. 
Then, the algorithms for computing suffix array values are modified as follows. 

\subsubsection{Modifying \texorpdfstring{$\compSA_{X, 1}$}{computeSAX}}
The operation $\compSA_{X, 1}(L_{\RLE, j}, \SA_{s, j}, \SA_{e, j}, x_{-1})$ computes $\SA_{j-1}[x_{-1}+1]$  
from $L_{\RLE, j}$, $\SA_{s, j}$, $\SA_{e, j}$, and the row index $x_{-1}$ of a circular shift in $\CM_{j-1}$,  
identified in Step~(I) of the update algorithm. 
Here, $j \in \{i+m+1, i+m, \ldots, i-K_{\del}+1 \}$ holds 
since most of the $(n+1)$ iterations are skipped. 

Let $y$ be the insertion position in $\CM_j$ determined in Step~(III). 
Then, the value $\SA_{j-1}[x_{-1}+1]$ is computed based on the following cases:  
(i) $j = i+m+1$;  
(ii) $j \in \{ i+m, i+m-1, \ldots, i \}$;  
(iii) $j < i$. 

\textbf{Case (i): $j = i+m+1$.} 
In this case, $\SA_{j-1}[x_{-1}+1]$ is computed as follows: 
\begin{align*}
    \SA_{j-1}[x_{-1}+1] &= 
  \begin{cases}
    \SA[\ISA[i+1]+1] - m & \text{if $\SA[\ISA[i+1]+1] \geq i+m$} \\
    \SA[\ISA[i+1]+1] + \Delta & \text{otherwise.}
  \end{cases}
\end{align*}
Here, this equation follows from Lemma~\ref{lem:state_stability_for_string_deletion}, 
and $\SA[\ISA[i+1]+1]$ is precomputed by $\phi^{-1}(i+1)$. 

\textbf{Case (ii): $j \in \{ i+m, i+m-1, \ldots, i \}$}. 
In this case, 
we compute $\SA_{j-1}[x_{-1} + 1]$ using the dynamic LF function. 
The value $\SA_{j-1}[x_{-1}+1]$ is computed via $\LF_j$ by finding $h$ such that $\LF_j(h) = x_{-1}+1$ if $x_{-1} \neq |\SA_{j-1}|$, and $\LF_j(h) = 1$ otherwise. 
Such $h$ always exists since $\LF_j$ is surjective.
We then compute $\SA_{j-1}[x_{-1}+1]$ using the computed $h$ and $\LF_{j}$ as follows:
$\SA_{j-1}[x_{-1}+1] = \SA_{j}[h]-1$ if $\SA_{j}[h] \neq \Delta + 1$; otherwise, $\SA_{j-1}[x_{-1}+1] = |T'|$. 

To compute $\SA_{j}[h]$, 
we introduce two positions $\alpha$ and $\beta$ in $\SA_{j}$. 
The first position $\alpha$ is defined as follows: 
\begin{equation}
  \alpha = 
  \begin{cases}
    x+1 & \text{if } x \ne t_v + \ell_v - 1, \\
    \min \mathcal{U} & \text{if } x = t_v + \ell_v - 1 \text{ and } \mathcal{U} \ne \emptyset, \\
    \min \mathcal{U}' & \text{otherwise.}
  \end{cases}
\end{equation}
Here, $x$ is the row index in $\CM_j$ identified in Step~(I); 
$(c_v, \ell_v)$ is the run containing the $x$-th character in $L_j$,  
where $L_{\RLE, j} = (c_1, \ell_1), (c_2, \ell_2), \ldots, (c_r, \ell_r)$; 
$\mathcal{U}$ and $\mathcal{U}'$ are the sets introduced in Lemma~\ref{lem:sah_lemma}. 

Similarly, 
the second position $\beta$ is defined as follows: 
\begin{equation}
  \beta = 
  \begin{cases}
    p+1 & \text{if } x \ne t_s + \ell_s - 1, \\
    \min \mathcal{U}_{p} & \text{if } x = t_s + \ell_s - 1 \text{ and } \mathcal{U}_{p} \ne \emptyset, \\
    \min \mathcal{U}_{p}' & \text{otherwise.}
  \end{cases}
\end{equation}
Here, $p$ is the row index of the $i$-th circular shift $T^{[i]}$ of $T$ in $\CM_{j}$; 
$(c_s, \ell_s)$ is the run containing the $p$-th character in $L_j$; 
$\mathcal{U} = \{d \mid v + 1 \leq d \leq r, c_d = c_s \}$;  
$\mathcal{U}' = \{d \mid 1 \leq d \leq r, c_d = c'_{\mathsf{succ}}\}$; 
$c'_{\mathsf{succ}}$ is the smallest character greater than $c_s$ in $\{c_d \mid 1 \leq d \leq r\}$;  
if none exists, set $c'_{\mathsf{succ}} = \$$.  

The following lemma gives the relationship among the three positions $h$, $\alpha$, and $\beta$. 
\begin{lemma}\label{lem:sah_lemma_for_string_deletion}
For any $j \in \{n+1, n, \ldots, 1\}$, 
if $j \in \{ i+m, i+m-1, \ldots, i \}$ and $\alpha = p$, 
then $h = \beta$. 
Otherwise, $h = \alpha$. 
\end{lemma}
\begin{proof}
    Lemma~\ref{lem:sah_lemma_for_string_deletion} can be proved using a similar technique to that of Lemma~\ref{lem:sah_lemma}. 
\end{proof}

Lemma~\ref{lem:sah_lemma_for_string_deletion} shows that 
$\SA_{j}[h]$ is $\SA_{j}[x+1]$, $\SA_{j}[p+1]$, or a value $\SA_{s,j}[x']$ in $\SA_{s,j}$. 
$\alpha$, $\beta$, $\SA_{j}[x+1]$, and $\SA_{s,j}[x']$ can be computed in $\mathcal{O}(\log n)$ time by modifying 
the original algorithm of $\compSA_{X, 1}$ presented in Section~\ref{subsec:support_SAXY}. 
$\SA_{j}[p+1]$ can be computed in $\mathcal{O}(\log n)$ time using the same method as 
the two operations $\compISA_{\del, 1}$ and $\compISA_{\del, 2}$. 
Therefore, we can compute $\SA_{j-1}[x_{-1} + 1]$ in $\mathcal{O}(\log n)$ time. 

\textbf{Case (iii): $j < i$}. 
Similar to case (ii), 
we compute $\SA_{j-1}[x_{-1} + 1]$ in $\mathcal{O}(\log n)$ time using Lemma~\ref{lem:sah_lemma_for_string_deletion}. 
Finally, the time complexity of $\compSA_{X, 1}$ is $\mathcal{O}(\log n)$. 

\subsubsection{Modifying Other Operations}
The algorithms of $\compSA_{X, 2}$, $\compSA_{Y, 1}$, and $\compSA_{Y, 2}$ 
can be modified using a similar approach to $\compSA_{X, 1}$. 
Similar to $\compSA_{X, 1}$, 
the time complexity of these modified algorithms can be bounded by $\mathcal{O}(\log n)$.

\subsection{Overall Time Complexity}
Similar to character insertion, 
$L_{\RLE}$, $\SA_{s}$, and $\SA_{e}$ are updated by executing $\mathcal{O}(m + K_{\del})$ iterations. 
Each iteration $j \in \{ i+m, i+m-1, \ldots, i-K_{\del} \}$ can be executed in $\mathcal{O}(1)$ queries and update operations supported by 
$\mathscr{D}_{\RLE}(L_{\RLE, j})$, $\mathscr{D}_{\DI}(\SA_{s, j})$, and $\mathscr{D}_{\DI}(\SA_{e, j})$. 
Thus, the iteration $j$ can be executed in $\mathcal{O}(\log r')$, 
where $r'$ is the number of runs in $L_{\RLE, j}$, and $r' \leq n$ holds. 
In addition, we use two operations $\compISA_{\del, 1}$ and $\compISA_{\del, 2}$, 
which take $\mathcal{O}((m + i - \SA_s[\lambda]) \log r)$ time. 
Therefore, this update algorithm requires $\mathcal{O}((m + K_{\del} + i - \SA_{s}[\lambda])\log n)$ time in total. 

\subsection{Worst-case Time Complexity}
When deleting a substring $T[i..i + m - 1]$ from $T$, 
the deletion takes $\mathcal{O}((m + K_{\del} + i - \SA_s[\lambda]) \log n)$ time, 
where $(1 + m + K_{\del})$ iterations are needed for the main update algorithm, 
and $(m + i - \SA_s[\lambda])$ steps are required for auxiliary operations. 
Similar to character insertion, 
the time complexity for deletion can be bounded using $L_{\max}$. 
This is because the following lemma holds instead of Lemma~\ref{lem:bound_K_LCP_R}. 

\begin{lemma}\label{lem:bound_K_string_deletion}
    There exists a constant $\alpha$ satisfying two conditions: 
    \begin{enumerate}[(i)]
        \item $i - \SA_{s}[\lambda] \leq \alpha (m + \max \{ \LCP^{R}[p], \LCP^{R}[p+1] \})$;
        \item $K_{\del} - 1 \leq \alpha (m + \max \{ \LCP^{R}[p], \LCP^{R}[p+1] \})$.    
   \end{enumerate}
   Here, $p$ is the position of $(n - i + 1)$ in the suffix array $\SA^{R}$ of $T^{R}$. 
\end{lemma}
\begin{proof}
Lemma~\ref{lem:bound_K_string_deletion} can be proved using the same approach 
as that used in the proof of Lemma~\ref{lem:bound_K_LCP_R}.
\end{proof}

Lemma~\ref{lem:bound_K_string_deletion} indicates that 
the time complexity for deletion is $\mathcal{O}(\max \{m, \LCP^{R}[p], \LCP^{R}[p+1] \} \log n)$. 
Similar to character deletion, 
there exists a position $p'$ in $\LCP$ such that 
the time complexity for deletion is $\mathcal{O}((m + \LCP[p']) \log n)$ time, 
where $\LCP[p'] \leq L_{\max}$. 
Therefore, the time complexity for deletion can be bounded by $\mathcal{O}((m + L_{\max}) \log n)$.

\subsection{Average-case Time Complexity}
We present an upper bound on the average-case time complexity for deletion using the average $L_{\avg}$ of the values in the LCP array of $T$. 
Similar to character insertion, 
we generalize the two values $K_{\del}$ and $\lambda$ to 
functions $K_{\del}(i, m)$ and $\lambda_{\del}(i, m)$, where a substring $T[i..i + m - 1]$ is deleted from $T$, 
$m \in \{ 1, 2, \ldots, n-1 \}$, and $i \in \{ 1, 2, \ldots, n-m \}$.

When a given substring of length $m$ is deleted from $T$, 
the average-case time complexity for deletion is 
$\mathcal{O}(\log n \sum_{i = 1}^{n-m} (m + K_{\del}(i, m) + i - \SA_{s}[\lambda_{\del}(i, m)]) / (n - m))$. 
Here, 
$\sum_{i = 1}^{n-m} (m + K_{\del}(i, m) + i - \SA_{s}[\lambda_{\del}(i, m)]) \leq \alpha \sum_{i=1}^{n-m} (m + \max \{ \LCP^{R}[p_{i-1}], \LCP^{R}[p_{i-1}+1] \})$ follows from Lemma~\ref{lem:bound_K_string_insertion}, 
where $\alpha$ is a constant, and $p_{i-1}$ is the position of $(n - i + 2)$ in the suffix array of $T^{R}$. 
Similar to Theorem~\ref{theo:query_time_summary}(ii), 
$\alpha \sum_{i=1}^{n-m} (m + \max \{ \LCP^{R}[p_{i-1}], \LCP^{R}[p_{i-1}+1] \})$ 
is at most $2 \alpha \sum_{i=1}^{n} (m + \LCP[i])$. 

If $m < n/2$, 
then we obtain the average-case time complexity for deletion by the following equation:
\begin{equation*}
    \begin{split}
        \mathcal{O}(\frac{\log n \sum_{i = 1}^{n-m} (m + K_{\del}(i, m) + i - \SA_{s}[\lambda_{\del}(i, m)])}{n-m} ) &=  \mathcal{O}(\frac{nm + \sum_{i=1}^{n} \LCP[i]}{n-m} \log n) \\
        &=  \mathcal{O}(\frac{nm + \sum_{i=1}^{n} \LCP[i]}{n} \log n) \\
        &= \mathcal{O}((m + L_{\avg}) \log n).
    \end{split}
\end{equation*}

Otherwise (i.e., $m \geq n/2$), $\LCP[i] \leq 2m$ for all $i \in \{ 1, 2, \ldots, n \}$. 
Therefore, we obtain the average-case time complexity for deletion by the following equation:
\begin{equation*}
    \begin{split}
        \mathcal{O}(\frac{\log n \sum_{i = 1}^{n-m} (m + K_{\del}(i, m) + i - \SA_{s}[\lambda_{\del}(i, m)])}{n-m} ) &=  \mathcal{O}(\frac{nm + \sum_{i=1}^{n} \LCP[i]}{n - m} \log n) \\
        &= \mathcal{O}(\frac{m^{2}}{m} \log n)  \\
        &= \mathcal{O}(m \log n).
    \end{split}
\end{equation*}

Finally, we obtain Theorem~\ref{theo:string_deletion_theorem}.

\section{Experiments}\label{sec:experiments}
\newcommand{\Times}[1]{#1$\times$}
\newcommand{\TimesTBL}[1]{(#1)}
\newcommand{\Dataset}[1]{\textsf{#1}}

\begin{table*}[ht]
\footnotesize
\centering
\caption{
Statistics of datasets.
$\sigma$ is the alphabet size of each string $T$ in the datasets; 
$n$ is the length of $T$; 
$r$ is the number of runs in the BWT of $T$; 
$L_{\avg}$ is the average of the values in the LCP array;
$L_{\max}$ is the maximum value in the LCP array of $T$; 
$\occ$ is the average number of occurrences of the 10000 patterns of length 100 
in $T$, measured for locate queries (see also Table~\ref{tab:locate}).
}
\label{tab:dataset_complete}
\begin{tabular}{l||r|r|r|r|r|r|r}

Data & $\sigma$ & $n$ [$10^3$] & $r$ [$10^3$] & $L_{\avg}$ & $L_{\max}$ & $n/r$ & ave. $\occ$  \\
\hline\hline
\Dataset{cere}              & 5   & 461,287    & 11,575 & 7080 & 303,204 & 40 & 25  \\
\Dataset{coreutils}         & 236 & 205,282    & 4,684 & 149,625  & 2,960,706 & 44 & 166   \\
\Dataset{einstein.de.txt}   & 117 & 92,758     & 101 & 35,248    & 258,006  & 915 & 600   \\
\Dataset{einstein.en.txt}   & 139 & 467,627    & 290 & 59,075    & 935,920  & 1,611 & 2,567   \\
\Dataset{Escherichia\_Coli} & 15  & 112,690    & 15,044 & 11,322 & 698,433  & 7 & 7  \\
\Dataset{influenza}         & 15  & 154,809    & 3,023 & 774  & 23,483  & 51 & 296  \\
\Dataset{kernel}            & 160 & 257,962    & 2,791 & 173,308  & 2,755,550  & 92 & 51     \\
\Dataset{para}              & 5   & 429,266    & 15,637 & 3,275 & 104,177  & 27 & 20  \\
\Dataset{world\_leaders}    & 89  & 46,968     & 573 & 8,837    & 695,051  & 82 & 1269  \\
\hline
\Dataset{enwiki}            & 207 & 37,226,668 & 71,091 & 45,673 & 1,788,009  & 534 & 1,769   \\
\Dataset{chr19}        & 5   & 59,125,169 & 45,143 & 124,328 & 3,272,696  & 1,310 & 935   \\
\end{tabular}
\end{table*}

\subparagraph{Setup.}
In this section, we evaluate the effectiveness of the dynamic r-index ($\DRI$) 
when performing locate queries and character insertions, using a dataset of eleven highly repetitive strings: 
(i) nine strings from the Pizza\&Chili repetitive corpus~\cite{dataset:pc-repetitive-corpus}; 
(ii) a 37~GB string (\Dataset{enwiki}) consisting of English Wikipedia articles with complete edit history~\cite{dataset:enwiki-all-pages}; and 
(iii) a 59~GB string (\Dataset{chr19}) created by concatenating chromosome 19 from 1,000 human genomes in the 1000 Genomes Project~\cite{1000Genomes}. 
The relevant statistics for each string in the dataset are presented in Table~\ref{tab:dataset_complete}.

We compared the dynamic r-index with both a dynamic and a static self-index: the dynamic FM-index ($\DFMI$) and the r-index ($\RI$), respectively.  
We implemented the dynamic r-index and dynamic FM-index in C++, and used the publicly available implementation of r-index from \url{https://github.com/nicolaprezza/r-index}. 
For performance reasons, the dynamic data structures used in the dynamic r-index are implemented using B-trees.
We conducted all experiments on one core of an 48-core Intel Xeon Gold 6126 CPU at 2.6 GHz in a machine with 2 TB of RAM running the 64-bit version of CentOS 7.9.



\subparagraph{Results.}
\begin{table*}[tb]
\footnotesize
\centering
\caption{
Performance summary showing the average time and standard deviation for character insertion. $\DRI$, $\DFMI$, and $\RI$ represent dynamic r-index, dynamic FM-index, and r-index, respectively.
}
\label{tab:insertion}

\begin{tabular}{r||rrr}
                     & \multicolumn{3}{c}{Insertion Time (sec)} \\ 
Data                     & $\DRI$               & $\DFMI$   & $\RI$     \\
\hline\hline
\Dataset{cere}       & 0.06 $\pm$ 0.12      & 0.07 $\pm$ 0.14   & 345.78   \\
\Dataset{coreutils}  & 1.96 $\pm$ 4.27    & 2.51 $\pm$ 5.77   & 128.67   \\
\Dataset{einstein.de.txt} 
                     & 0.22 $\pm$ 0.17      & 0.56 $\pm$ 0.46   & 35.12    \\
\Dataset{einstein.en.txt} 
                     & 0.50 $\pm$ 0.38      & 1.30 $\pm$ 1.02   & 183.33   \\
\Dataset{Escherichia\_Coli} 
                     & 0.12 $\pm$ 0.47      & 0.12 $\pm$ 0.52   & 112.01   \\
\Dataset{influenza}  & 0.01 $\pm$ 0.01      & 0.01 $\pm$ 0.01   & 93.72    \\
\Dataset{kernel}     & 2.05 $\pm$ 3.44      & 3.14 $\pm$ 5.50   & 151.88   \\
\Dataset{para}       & 0.03 $\pm$ 0.03      & 0.03 $\pm$ 0.05   & 352.39   \\
\Dataset{world\_leaders} 
                     & 0.06 $\pm$ 0.27      & 0.11 $\pm$ 0.53   & 25.18    \\
\hline
\Dataset{enwiki}     & 0.89 $\pm$ 0.85      & 1.98 $\pm$ 2.00   & 40,157.89 \\
\Dataset{chr19} & 0.79 $\pm$ 2.14      & 1.37 $\pm$ 2.91   & 50,980.83 \\
\end{tabular}

\end{table*}


\begin{table*}[tb]
\footnotesize
\centering
\caption{
Performance summary on construction time in seconds. $\DRI$, $\DFMI$, and $\RI$ represent 
 dynamic r-index, dynamic FM-index, and r-index, respectively.
}
\label{tab:construction}

\begin{tabular}{l||rrr}
\hline
                     & \multicolumn{3}{c}{Construction Time (sec)} \\ 
Data               & $\DRI$      & $\DFMI$   & $\RI$     \\
\hline\hline
\Dataset{cere}       & 313    & 581   & 280   \\
\Dataset{coreutils}  & 114    & 266   & 101   \\
\Dataset{einstein.de.txt} 
                     & 35     & 108   & 29    \\
\Dataset{einstein.en.txt} 
                     & 179    & 725   & 156   \\
\Dataset{Escherichia\_Coli} 
                     & 122    & 115   & 79    \\
\Dataset{influenza}  & 75     & 171   & 64    \\
\Dataset{kernel}     & 133    & 350   & 132   \\
\Dataset{para}       & 341    & 502   & 288   \\
\Dataset{world\_leaders} 
                     & 22     & 47    & 19    \\
\hline
\Dataset{enwiki}     & 36,343 & 75,458 & 37,717 \\
\Dataset{chr19} & 49,291 & 102,485 & 48,540 \\
\hline \hline
Average              & 7,906  & 16,437 & 7,946 \\
\end{tabular}

\end{table*}

\begin{table*}[tb]
\footnotesize
\caption{
The table shows the average time and standard deviation for the backward search (left) and the occurrence computation (right) of each pattern after the backward search. $\DRI$, $\DFMI$, and $\RI$ represent dynamic r-index, dynamic FM-index, and r-index, respectively.}
\label{tab:locate}

\begin{tabular}{l||rrr|rrr}
\hline
                     & \multicolumn{3}{c|}{Backward Search (\textmu s)} & \multicolumn{3}{c}{Occurrence Computation (\textmu s)} \\ 
Data               & $\DRI$          & $\DFMI$          & $\RI$         & $\DRI$    & $\DFMI$   & $\RI$      \\
\hline\hline
\Dataset{cere}       & 434 $\pm$ 27  & 554 $\pm$ 78   & 131 $\pm$ 39    & 96 $\pm$ 66      & 187 $\pm$ 110       & 12 $\pm$ 9    \\
\Dataset{coreutils}  & 748 $\pm$ 52 & 1,345 $\pm$ 147  & 178 $\pm$ 58    & 245 $\pm$ 1,007    & 1,598 $\pm$ 7,241  & 35 $\pm$ 134    \\
\Dataset{einstein.de.txt} 
                     & 345 $\pm$ 18  & 1,113 $\pm$ 93  & 121 $\pm$ 35    & 564 $\pm$ 505  & 4,475 $\pm$ 3,836 & 76 $\pm$ 72    \\
\Dataset{einstein.en.txt} 
                     & 488 $\pm$ 31  & 1,816 $\pm$ 205  & 125 $\pm$ 33    & 2,521 $\pm$ 2,073  & 25,851 $\pm$ 20,800 & 285 $\pm$ 237     \\
\Dataset{Escherichia\_Coli} 
                     & 649 $\pm$ 166  & 693 $\pm$ 150   & 134 $\pm$ 40    & 50 $\pm$ 62        & 81 $\pm$ 71       & 4 $\pm$ 5     \\
\Dataset{influenza}  & 491 $\pm$ 104  & 686 $\pm$ 63  & 117 $\pm$ 37    & 504 $\pm$ 902  & 1,533 $\pm$ 2,808   & 52 $\pm$ 96     \\
\Dataset{kernel}     & 685 $\pm$ 41 & 1,487 $\pm$ 260  & 165 $\pm$ 51    & 94 $\pm$ 220      & 640 $\pm$ 1,603   & 14 $\pm$ 34     \\
\Dataset{para}       & 464 $\pm$ 35  & 538 $\pm$ 57  & 134 $\pm$ 37    & 83 $\pm$ 51       & 149 $\pm$ 78       & 9 $\pm$ 6     \\
\Dataset{world\_leaders} 
                     & 471 $\pm$ 49  & 964 $\pm$ 88  & 115 $\pm$ 14    & 1,635 $\pm$ 5,160 & 9,199 $\pm$ 29,617 & 151 $\pm$ 474     \\
\hline
\Dataset{enwiki}     & 1,386 $\pm$ 79 & 3,351 $\pm$ 256 & 308 $\pm$ 104   & 3,919 $\pm$ 6,190 & 26,688 $\pm$ 52,957 & 444 $\pm$ 515    \\
\Dataset{chr19} & 733 $\pm$ 49  & 2,321 $\pm$ 637  & 186 $\pm$ 33    & 3,876 $\pm$ 1,562  & 8,305 $\pm$ 3,992  & 361 $\pm$ 153     \\
\hline \hline
Average              & 627    & 1,352 & 156    & 1,235   & 7,155   & 131   \\
\end{tabular}

\end{table*}


\begin{table*}[tb]
\footnotesize
\centering
\caption{
Performance summary on working space in megabytes. $\DRI$, $\DFMI$, and $\RI$ represent dynamic r-index, dynamic FM-index, and r-index, respectively.
}
\label{tab:index_size}
\begin{tabular}{l||rrr}
\hline
                     & \multicolumn{3}{c}{Working Space (MiB)} \\ 
String               & $\DRI$    & $\DFMI$   & $\RI$      \\
\hline\hline
\Dataset{cere}       & 227 & 888 & 109    \\
\Dataset{coreutils}  & 105 & 556 & 48     \\
\Dataset{einstein.de.txt} 
                     & 4 & 187 & 2      \\
\Dataset{einstein.en.txt} 
                     & 10 & 1,251 & 4      \\
\Dataset{Escherichia\_Coli} 
                     & 408 & 235 & 127    \\
\Dataset{influenza}  & 68 & 323 & 29     \\
\Dataset{kernel}     & 64 & 698 & 29     \\
\Dataset{para}       & 437 & 827 & 145    \\
\Dataset{world\_leaders} 
                     & 15 & 97 & 6      \\
\hline
\Dataset{enwiki}     & 2,254 & 99,205 & 886   \\
\Dataset{chr19} & 1,420 & 113,113 & 540   \\
\hline \hline
Average              & 456 & 19,762 & 175    \\
\end{tabular}

\end{table*}

We focus on results only for the large data \Dataset{enwiki} and \Dataset{chr19}. 
Table~\ref{tab:insertion} presents the average time and standard deviation per character insertion.  
The results are based on 1,000 insertion operations, where each operation inserts a randomly chosen character at a randomly selected position within the input string.

Since r-index does not support character insertions, it must reconstruct the entire index for each insertion.
However, this process takes a significant amount of time for long string data.
In particular, it takes 11 and 14 hours for \Dataset{enwiki} and \Dataset{chr19}, respectively, highlighting the practical necessity of supporting update operations when handling large-scale string data. 

Although the insertion algorithm of dynamic r-index is based on that of dynamic FM-index, modified to use RLBWT instead of BWT, the average character insertion time of dynamic r-index is at most two orders of magnitude longer than that of dynamic FM-index across all string datasets.
For the two long string data, \Dataset{enwiki} and \Dataset{chr19}, dynamic r-index takes 0.89 and 0.79 seconds, respectively, which confirms its practical applicability.

The construction times for dynamic r-index, dynamic FM-index, and r-index are presented in Table~\ref{tab:construction}.

Table~\ref{tab:locate} shows the average time and standard deviation for locate queries, which consist of backward search and occurrence computation.
The r-index was the fastest in both tasks, showing the highest performance among all methods, though it lacks update support.
While dynamic r-index  was at most 11 times slower than the r-index for occurrence computation, 
the dynamic r-index was at least twice and at most 7 times faster than dynamic FM-index.
This is because dynamic r-index performs lightweight occurrence computations using sampled suffix arrays, whereas dynamic FM-index relies on $\mathcal{O}(\log^{2+\epsilon} n)$-time computations of the LF function (Section~4.6 in~\cite{DBLP:journals/jda/SalsonLLM10}).

Table~\ref{tab:index_size} shows the working space used during locate queries.  
r-index was the most space-efficient across all string datasets, as it does not support update operations and can employ space-efficient data structures.  
Dynamic r-index (based on RLBWT) was more space-efficient than dynamic FM-index (based on BWT) for most string data.
In particular, the difference was substantial for large datasets:  
the working space of dynamic r-index was 80 and 43 times smaller than that of dynamic FM-index on \Dataset{enwiki} and \Dataset{chr19}, respectively.

\section{Conclusion}
We have presented dynamic r-index, a dynamic extension of the r-index that supports update operations of string insertion and deletion. 
The average-case time complexity of these update operations is proportional to the average LCP value. 
Although some highly repetitive strings, such as Fibonacci strings, have a large average LCP value,
experiments using various highly repetitive strings demonstrated practicality of the dynamic r-index.

\bibliography{ref}
\clearpage

\appendix
\section{Complete Figure~\ref{fig:update_CM}}\label{sec:appendix}
\begin{figure}[htbp]
\begin{center}
	\includegraphics[width=1.0\textwidth]{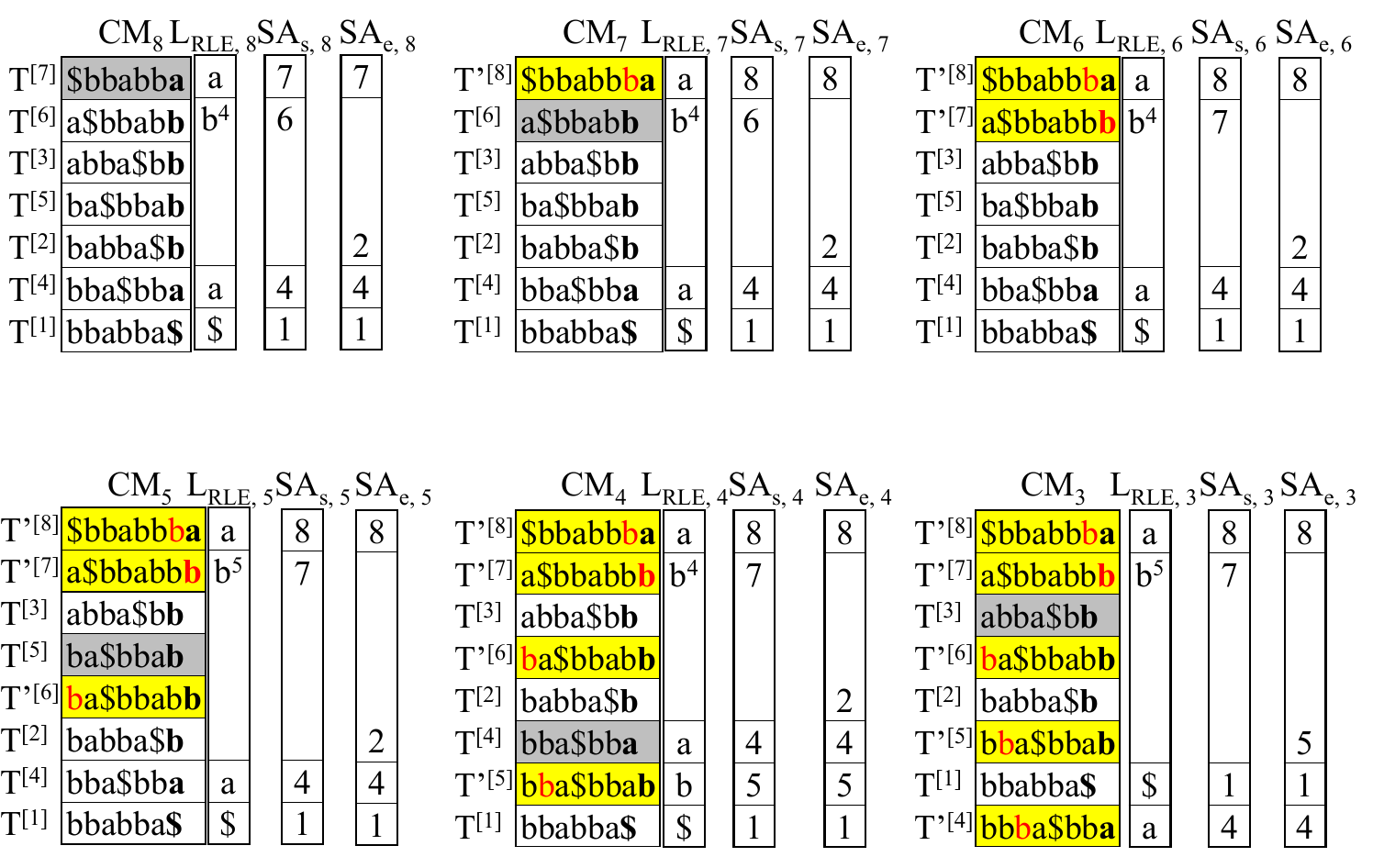}
	\includegraphics[width=1.0\textwidth]{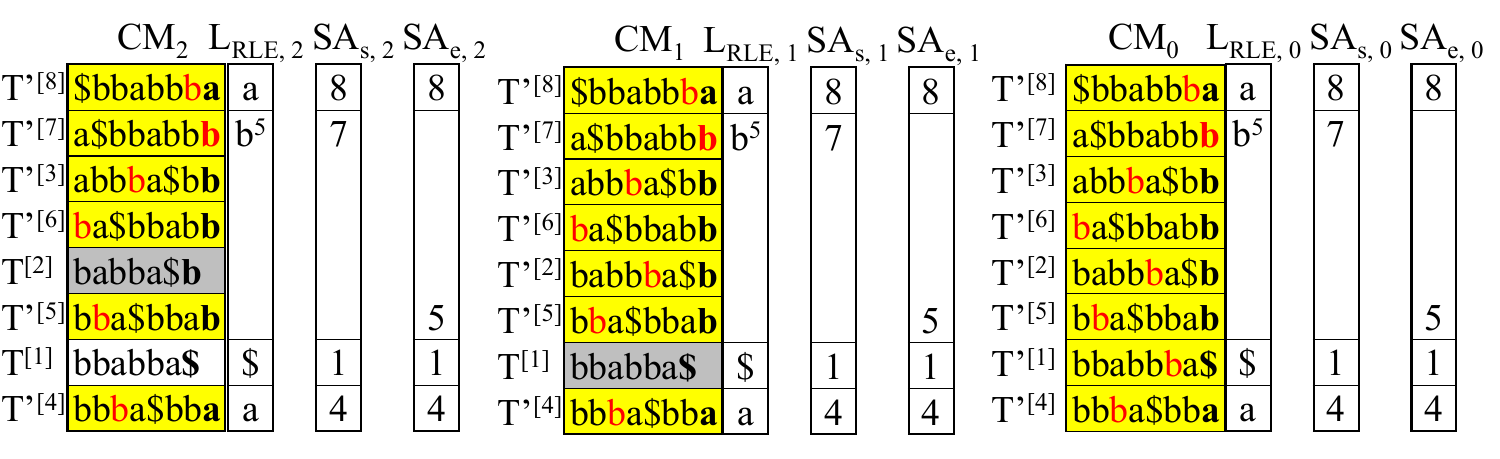}
\caption{
Illustration of the update process from $\CM$, $L_{\RLE}$, $\SA_{s}$, and $\SA_{e}$ 
to $\CM'$, $L'_{\RLE}$, $\SA'_{s}$, and $\SA'_{e}$, respectively. 
Here, $T = \texttt{bbabba\$}$ and $T' = \texttt{bbabbba\$}$, 
where $T'$ is obtained by inserting the character $b$ into $T$ at position 6. In each $\CM_j$, the circular shift to be removed is highlighted in gray, while the newly inserted circular shift is highlighted in yellow. The BWT corresponding to $\CM_{j}$ is shown with bold characters.}
\label{fig:full_update_CM}
\end{center}
\vspace{-0.8cm}
\end{figure}

\end{document}